\documentclass[12pt]{article}

\usepackage{mystyle}

\newcommand\ringring[1]{%
  {%
   \mathop{\kern0pt #1}\limits^{%
     \vbox to-1.85ex{
       \kern-2ex %
       \hbox to 0pt{\hss\normalfont\kern.1em \r{}\kern-.45em \r{}\hss}%
       \vss %
     }%
   }%
  }%
}

\newcommand{\myleft}{\mathopen{}\mathclose\bgroup\left}
\newcommand{\myright}{\aftergroup\egroup\right}
\newcommand{\matchedDelimiter}[3]{\myleft#1 #3 \myright#2}
\newcommand{\tmatchedDelimiter}[3]{#1 #3 #2} %
\newcommand{\lr}[1]{\matchedDelimiter{(}{)}{#1}}
\newcommand{\tlr}[1]{\tmatchedDelimiter{(}{)}{#1}}

\newcommand{\brackets}[1]{\matchedDelimiter{[}{]}{#1}}

\newcommand{\set}[1]{\matchedDelimiter{\lbrace}{\rbrace}{#1}}
\newcommand{\tset}[1]{\tmatchedDelimiter{\lbrace}{\rbrace}{#1}}

\newcommand{\setM}[2]{\matchedDelimiter{\lbrace}{\rbrace}{#1 : #2}}

\newcommand{\norm}[2][]{\matchedDelimiter{\Vert}{\Vert}{#2}_{#1}}
\newcommand{\tnorm}[2][]{\tmatchedDelimiter{\Vert}{\Vert}{#2}_{#1}}

\newcommand{\abs}[1]{\left|{#1}\right|} %

\newcommand{\bilin}[2][]{\matchedDelimiter{<}{>}{#2}_{#1}}

\newcommand{\without}[1]{\setminus\set{#1}}

\newcommand{\restrictTo}[2]{{#1\big\vert_{#2}}}
\newcommand{\condInd}[3]{{#1 \perp \!\!\! \perp #2 \mid #3}}
\newcommand{\exInd}[3]{{#1 \perp_e #2 \ifthenelse{\equal{#3}{}}{}{\mid #3}}}

\newcommand{\deq}[1][d]{\stackrel{#1}{=}}

\renewcommand{\t}[1][t]{^{\myleft[#1\myright]}}
\renewcommand{\k}[1][k]{^{\lr{#1}}}
\newcommand{\tk}[1][t,k]{^{\myleft[#1\myright]}}

\newcommand{\T}{^\top\!}
\newcommand{\inv}{^{-1}}

\newcommand{\pinv}{^+}
\newcommand{\I}[1][I]{_{(#1)}}

\newcommand{\N}{\mathbb{N}}

\newcommand{\Rdd}[1][d]{\mathbb{R}^{#1 \times #1}}
\newcommand{\Rd}[1][d]{\mathbb{R}^{#1}}
\newcommand{\ROdd}[1][d]{\mathring{\mathbb{R}}^{#1 \times #1}}
\newcommand{\ROd}[1][d]{\mathring{\mathbb{R}}^{#1}}

\newcommand{\oneToX}[1]{{1, \dots, #1}}
\newcommand{\vecDots}[3]{\lr{#1_{#2}, \dots, #1_{#3}}}

\renewcommand{\L}{{\mathcal{L}}}

\newcommand{\SOk}[1][k]{\mathring\Sigma^{(#1)}}

\newcommand{\Image}{\operatorname{Im}}%

\newcommand{\Cli}{{\mathcal{C}}}
\newcommand{\Sep}{{\mathcal{D}}}

\newcommand{\Qcal}{\mathcal{Q}}
\newcommand{\Pcal}{\mathcal{P}}
\newcommand{\Rcal}{\mathcal{R}}
\newcommand{\Dcal}{\mathcal{D}}

\newcommand{\GammaO}{{\mathring\Gamma}}

\newcommand{\SigmaO}{{\mathring\Sigma}}

\newcommand{\SigmaT}{{\tilde\Sigma}}

\newcommand{\undefined}{\mathrm{``{?}"}}
\newcommand{\undefinedNoQuotes}{\mathrm{?}}
\newcommand{\undefinedMatrix}{\mathrm{``{??}"}}
\newcommand{\undefinedMatrixNoQuotes}{\mathrm{??}}
\newcommand{\half}{{\tfrac{1}{2}}}

\newcommand{\LambdaVec}{\Lambda^{\!c}}

\newcommand{\forsome}{\mathrm{~for~some~}}

\newcommand{\Id}{{I_d}}
\newcommand{\inftyVec}{{\boldsymbol{\infty}}}
\newcommand{\zeroVec}{{\mathbf{0}}}
\newcommand{\oneVec}{\mathbf{1}}
\newcommand{\oneVecT}{{\mathbf{1}\T}}
\newcommand{\oneVecTT}{{\oneVec\oneVecT}}

\newcommand{\e}{\mathbf{e}}

\newcommand{\Proj}[1][]{{\Pi}_{#1}}
\newcommand{\ID}{\Proj}

\newcommand{\tr}{\operatorname{tr}}
\newcommand{\eV}[1][d]{\mathbf{e}_{#1}}

\newcommand{\Ebar}{\overline E}

\newcommand{\diff}{\mathrm{d}}

\newcommand{\Gflight}{G_{\text{flight}}}
\newcommand{\Gflow}{G_{\text{flow}}}

\newcommand{\myCases}[1]{
    \begin{cases}
        \begin{aligned}
            #1
        \end{aligned}
    \end{cases}
}
\newcommand{\myCase}[2]{&#1&\qquad&#2}

\newcommand{\myFunctionDefinition}[5]{
  #1: \; & #2 &&\longrightarrow #3
  \\
  & #4 &&\longmapsto #5
}

\newcommand{\myFunctionDefinitionInline}[5]{
  #1: \; & #2 &\longrightarrow #3
  , \quad
  #4 \longmapsto #5
}

\let\oldBeginMatrix\matrix
\let\oldEndMatrix\endmatrix
\renewenvironment{matrix}{
  \ifdefmacro{\myspacingset}{
    \myspacingset{1.0}
  }{}
  \oldBeginMatrix
}{
  \oldEndMatrix
  \ifdefmacro{\myspacingset}{
    \myspacingset{1.9}
  }{}
}

\newcommand{\myFunction}[2]{
  #1 \ifx @#2@ \else
  \mathchoice{
    \lr{#2}
  }{
    \tlr{#2}
  }{
    \lr{#2}
  }{
    \lr{#2}
  }
  \fi
}

\newcommand{\Cov}[1]{{\mathrm{Cov}\lr{#1}}}

\newcommand{\littleO}[1]{o\lr{#1}} %
\newcommand{\comp}[2][]{\myFunction{\mathfrak{C}_{#1}}{#2}}
\newcommand{\diag}[1]{\myFunction{\operatorname{diag}}{#1}}
\newcommand{\lambdaM}[2][]{\myFunction{\lambda_{#1}}{{#2}_{#1}}}

\renewcommand{\P}[2][]{\myFunction{\mathbb{P}_{#1}}{#2}}
\newcommand{\E}[2][]{\myFunction{\mathbb{E}_{#1}}{#2}}

\newcommand{\limit}[2][\infty]{\lim_{#2\rightarrow #1}}
\newcommand{\proj}[2][]{\myFunction{\sigma_{#1}}{#2}}

\newcommand{\ppinv}[1]{\myFunction{\theta}{#1}}

\newcommand{\phik}[2][k]{\myFunction{\varphi_#1}{#2}}
\newcommand{\phikt}[2][k]{\myFunction{\tilde\varphi_#1}{#2}}
\newcommand{\KLdiv}[1]{\myFunction{\mathcal{I}}{#1}}
\newcommand{\pdet}[1]{\abs{#1}_+}
\newcommand{\trace}[1]{\myFunction{\mathrm{tr}}{#1}}
\newcommand{\laSpan}[1]{\myFunction{\mathrm{span}}{{#1}}}
\newcommand{\expF}[1]{\myFunction{\exp}{#1}}
\newcommand{\ExpF}[1]{\myFunction{\mathrm{Exp}}{#1}}

\newcommand{\normal}[1]{\myFunction{\mathcal{N}}{#1}}

\newcommand{\farris}[1]{\myFunction{\gamma}{#1}}

\newcommand{\allEdges}[1][V]{\myFunction{\mathcal{E}}{#1}}
\newcommand{\diagV}[1][V]{D_{#1}}

\newcommand{\CAP}{\expandafter\MakeUppercase}
\newcommand{\HR}{Hüsler--Reiss}
\newcommand{\HRP}{Hüsler--Reiss Pareto}

\newcommand{\HC}{Hammersley--Clifford}
\newcommand{\MP}{Moore--Penrose}
\newcommand{\gP}{generalized Pareto}

\newcommand{\multP}{multivariate Pareto}
\newcommand{\multGP}{multivariate generalized Pareto}

\newcommand{\KL}{Kullback--Leibler}

\newcommand{\PD}{positive definite}
\newcommand{\PSD}{positive semi-definite}
\newcommand{\CND}{conditionally negative definite}

\newcommand{\PCND}{partially conditionally negative definite}

\newcommand{\eglearn}{\texttt{EGlearn}}

\newcommand{\quotes}[1]{``#1''}

\providebool{cleanbuild}
\newcommand{\neutralshorten}[1]{%
  \ifbool{cleanbuild}{\ifvmode\else\unskip\fi\ignorespaces}{#1}%
}
\newcommand{\comment}[2]{%
  \neutralshorten{%
    \ifmmode%
      \text{\textcolor{#1}{\sffamily\small[{#2}]}}%
    \else%
      \noindent%
      \textcolor{#1}{\sffamily\small[{#2\unskip}]}%
    \fi%
  }%
}

\newcommand{\maybecolor}[1]{\ifbool{cleanbuild}{}{\color{#1}}}
\definecolor{darkred}{rgb}{0.8,0,0}

\definecolor{darkgreen}{rgb}{0,0.8,0}

\definecolor{applegreen}{rgb}{0.58,0.7,0}

\newenvironment{align**}{
  \(
}{
  \)
}

\newcommand{\provenLabel}[1]{
  \label{#1}
  \ifbool{cleanbuild}{}{%
  \hyperlink{proof:#1}{(See Proof.)}%
  }
}

\newenvironment{proofOf}[1][]{
  \ifx @#1@
    \begin{proof}
  \else
    \begin{proof}[Proof of \hypertarget{proof:#1}{\cref{#1}}]
  \fi
}{
  \end{proof}
}

\usepackage{geometry}

\newcommand{\valUsedP}{0.95}
\newcommand{\valUsedPClustering}{0.95}
\newcommand{\valUsedPMarginals}{0.95}
\newcommand{\valDateEstStart}{2005-01-01}
\newcommand{\valDateEstEnd}{2010-12-31}
\newcommand{\valDateValStart}{2011-01-01}
\newcommand{\valDateValEnd}{2020-12-31}
\newcommand{\valYearStart}{2005}
\newcommand{\valYearEnd}{2020}
\newcommand{\valMinYearlyAirportFlights}{1000}

\newcommand{\valNDaysEst}{1764}
\newcommand{\valNDaysVal}{1839}

\newcommand{\valNYearsAll}{16}

\newcommand{\valNAirportsFiltered}{169}
\newcommand{\valNAirportsCluster}{29}

\newcommand{\valNDaysFiltered}{5601}

\newcommand{\valNClustersWord}{four}

\newcommand{\valRhoStar}{0.15}

\newcommand{\valRhoMax}{0.5}
\newcommand{\valNEdgesEglearn}{101}
\newcommand{\valNEdgesTree}{28}
\newcommand{\valNEdgesFlightGraph}{250}
\newcommand{\valNEdgesComplete}{406}
\newcommand{\valTreeGraphLoglik}{-6435}

\newcommand{\danUsedP}{0.9}
\newcommand{\danYearAllStart}{1960}
\newcommand{\danYearAllEnd}{2010}
\newcommand{\danYearEstStart}{1960}
\newcommand{\danYearEstEnd}{1985}
\newcommand{\danYearValStart}{1986}
\newcommand{\danYearValEnd}{2010}
\newcommand{\danNObsAll}{428}
\newcommand{\danNObsEst}{220}
\newcommand{\danNObsVal}{208}
\newcommand{\danNStations}{31}
\newcommand{\danRhoStar}{0.06}

\newcommand{\danRhoMax}{0.1}
\newcommand{\danNEdgesEglearn}{69}
\newcommand{\danNEdgesTree}{30}
\newcommand{\danNEdgesFlowGraph}{30}
\newcommand{\danNEdgesComplete}{465}

\newcommand{\danCompleteGraphLoglik}{-1810}

\newcommand{\danFlowGraphLoglik}{-329}
\newcommand{\danTreeGraphLoglik}{-265}

\newcommand{\pdfPath}[1]{figures/\detokenize{#1}.pdf}
\newcommand{\inputTikz}[2][1\textwidth]{
  \ifblank{#1}{
    \includegraphics{\pdfPath{#2}}
  }{
    \includegraphics[width=#1]{\pdfPath{#2}}
  }
}

\setbool{cleanbuild}{true} %

\begin{document}

\newgeometry{margin=1in}

\title{\bf Statistical Inference for \HR{} Graphical Models Through Matrix Completions}
\author{
Manuel Hentschel \\
Research Center for Statistics, University of Geneva\\
and \\
Sebastian Engelke \\
Research Center for Statistics, University of Geneva\\
and \\
Johan Segers\\
ISBA/LIDAM, UCLouvain}
\maketitle
\bigskip
\begin{abstract}

The severity of multivariate extreme events is driven by the dependence between the largest marginal observations. The \HR{} distribution is a versatile model for this extremal dependence, and it is usually parameterized by a variogram matrix.
In order to represent conditional independence relations and obtain sparse parameterizations, we introduce the novel \HR{} precision matrix. Similarly to the Gaussian case, this matrix appears naturally in density representations of the \HR{} Pareto distribution and encodes the extremal graphical structure through its zero pattern.
For a given, arbitrary graph we prove the existence and uniqueness of the completion of a partially specified \HR{} variogram matrix so that its precision matrix has zeros on non-edges in the graph. Using suitable estimators for the parameters on the edges, our theory provides the first consistent estimator of graph structured \HR{} distributions.
If the graph is unknown, our method can be combined with recent structure learning algorithms to jointly infer the graph and the corresponding parameter matrix.
Based on our methodology, we propose new tools for statistical inference of sparse \HR{} models
and illustrate them on large flight delay data in the U.S.,
as well as Danube river flow data.

\end{abstract}
\noindent%

\vfill

\restoregeometry

\newpage

\etocdepthtag.toc{mtchapter}
\etocsettagdepth{mtchapter}{subsection}
\etocsettagdepth{mtappendix}{none}

\section{Introduction}

\newcommand{\topic}[1]{}

In statistical modelling,
conditional independence and graphical models
are well-established concepts for analyzing structural relationships in data
\citep{lauritzen1996, wainwright2008}.
Particularly important are Gaussian graphical models, also known as Gaussian Markov random fields \citep{rue2005}.
The graphical structure of a multivariate normal distribution
with positive definite covariance matrix $\Sigma$ can be read off from the
zeros of its precision matrix $\Sigma\inv$.

For risk assessment in fields such as climate science, hydrology, or finance, the primary interest
is in extreme observations, with attention to both the marginal tails and the dependence between multiple risk factors.
In view of the growing complexity and dimensionality of modern data sets,
sparsity and graphical models are becoming crucial notions for
the analysis of extremes \citep[e.g.,][]{engelke2021a}.

There are two different ways of defining graphical models for extreme value distributions.
The first is based on max-linear models \citep{gissibl2018} and the second one studies \multP{} distributions \citep{engelke2020}.
We follow the second approach since their new notions of conditional independence and extremal graphical models link naturally to 
the well-known \HC{} theorem for density factorizations.
Moreover, in the case of tree graphs, \cite{segers2020} shows that the extremes of regularly varying Markov trees converge to these extremal tree models.
\cite{lee2018} propose parsimonious models for extreme value copulas; the link with extremal graphical models is made in \cite{asenova2021inference}.

The class of extremal graphical models with
\HRP{} distributions is of particular interest.
In $d$ dimensions, the parameter of this family is
a variogram matrix $\Gamma \in \Rdd[d]$. 
Because of their flexibility and stochastic properties,
\HR{} distributions can be seen as the counterpart of the Gaussian family for multivariate extremes.
In combination with extremal graphical models, the \HR{} family
constitutes a powerful tool for sparse extreme value modelling, with many open questions
still to explore.

For a connected, undirected graph $G= (V,E)$ with nodes $V= \{1,\dots, d\}$ and edges $E$,
\cite{engelke2020} show that such a distribution's graphical structure
can be read off
from a set of $(d-1) \times (d-1)$ precision matrices $\Theta\k$, for $k \in V$.
While zeros in $\Theta\k$ correspond to extremal conditional independence of nodes $i,j \neq k$, the information on edges involving the $k$th node is encoded only indirectly through the row sums of this matrix.
A natural question, appearing also in the discussion of \cite{engelke2020}, is if there exists a symmetric approach involving a single $d\times d$ precision matrix.

Statistical inference for \HR{} graphical models is limited so far to the simple structures of trees and block graphs \citep{engelke2020a, asenova2021extremes}.
The parameter matrix $\Gamma$ is then additive on the graph and the maximum likelihood estimator 
is an explicit
combination of the bivariate estimators on the edges.
Since block graphs 
lack flexibility for
general applications, several discussion contributions of \cite{engelke2020} have emphasized the need for estimators suitable for more general graphs.

In this paper, we obtain new theoretical results on \HR{} distributions that
answer the two open questions above and enable statistical inference
on extremal graphical models on decomposable and non-decomposable graphs.
Our main contributions are threefold. Firstly, in \cref{sec:HrPrecisionMatrix}, we introduce the \HR{} precision matrix $\Theta \in \Rdd$ as
$\Theta_{ij} = \Theta\k_{ij}$ for some $k \neq i, j$, a definition
which---surprisingly---is independent of the particular choice of $k \in V$.
This positive semi-definite matrix indeed reflects the sparsity of the
extremal graph by zero off-diagonal entries.
We give several characterizations of this matrix, one of them
as the Moore--Penrose inverse of
a projection of the parameter matrix $\Gamma$.

Secondly, we study how a \HR{} distribution on a given, general graph $G=(V,E)$ and fixed marginal distributions on the edges of the graph can be constructed. Thanks to the new \HR{} precision matrix, this task can be framed as a matrix completion  problem, which aims to find
a conditionally negative definite variogram matrix $\Gamma$
that has specified values $\Gamma_{ij}$ in the entries corresponding to the edges $(i,j) \in E$
of graph $G$ 
and whose precision matrix $\Theta$ has zeros in the remaining entries, i.e., $\Theta_{ij} = 0$ for $(i,j) \notin E$. A concrete example in dimension $d=4$ for the graph in \autoref{subfig:exampleBlock} is
\begin{align}
    \label{eq:simpleGammaThetaPartial}
    \Gamma
    =
    \lr{
        \begin{matrix}
     0 & 3 & \mathrm{?} & 1 \\
     3 & 0 & 10 & 2 \\
     \mathrm{?} & 10 & 0 & \mathrm{?} \\
     1 & 2 & \mathrm{?} & 0
\end{matrix}
    }
    , \qquad
    \Theta
    =
    \lr{
        \begin{matrix}
     \mathrm{?} & \mathrm{?} & 0 & \mathrm{?} \\
     \mathrm{?} & \mathrm{?} & \mathrm{?} & \mathrm{?} \\
     0 & \mathrm{?} & \mathrm{?} & 0 \\
     \mathrm{?} & \mathrm{?} & 0 & \mathrm{?}
\end{matrix}
    }
    ,
\end{align}%
where the completion problem corresponds to finding the entries with $\undefined$.
In \cref{sec:completion}, we show that such a completion exists and is unique. 
Our results can be seen as a semi-definite extension of matrix completion problems for
the covariance matrix of Gaussian distributions studied in \cite{speedKivery1986} and \cite{bakonyi2011}.

Thirdly, we leverage our theoretical results to provide effective statistical tools for extremal graphical models. The \HR{} precision matrix
allows us to represent the maximum (surrogate-)likelihood estimate of $\Gamma$ on the graph $G$ as the maximizer of the constrained optimization problem
\begin{equation*}
     \log \pdet{\Theta} + \tfrac{1}{2}\tr(\widehat{\Gamma}\Theta), \quad \text{s.t. } \Theta_{ij} = 0 \text{ if } (i,j) \notin E,
\end{equation*}
where $\widehat \Gamma$ is the empirical variogram \citep{engelke2020a} and $\pdet{\, \cdot\,}$ denotes the pseudo-determinant. In \cref{sec:statinf}, we prove that the solution to this optimization problem is given by the solution of the above matrix completion problem. We further combine recent structure learning methods \citep{engelke2022a} with our completion to yield the first estimator that is jointly sparsistent for the extremal graph and consistent for the model parameters. 
Our methodology enables new model assessment plots that allow for model interpretation and comparison of models with different degrees of sparsity. 
This is illustrated in a case study of large delays in the domestic U.S.~air travel network in \cref{sec:applic}.
The supplementary material contains another case study, together with mathematical details and proofs.

The \HR{} precision matrix and the properties we derive in this paper have already been used for multiple purposes, including the parameterization of sparse statistical models \citep{roettgerSchmitz2023, roettger2023parametric}, structure estimation \citep{engelke2022a,wan2023graphical}, efficient statistical inference \citep{roettger2023,lederer2023extremes}, and a characterization in terms of pairwise interaction models \citep{lalancette2023pairwise}.

\section{Extremal graphical models}
\label{sec:extrgraphmod}
\subsection{Multivariate generalized Pareto distributions}

Multivariate extreme value theory studies the tail behavior of a random vector $X = \vecDots{X}{1}{d}$. A first summary of the extremal dependence structure of the bivariate margins for $i,j \in V := \{1,\dots, d\}$ is the extremal correlation
\begin{align}
    \label{eq:defChi}
    \chi_{ij}
    :=
    \limit[0]{p} \chi_{ij}(p)
    :=
    \limit[0]{p}
    \P{
        F_i(X_i) > 1-p
        \mid
        F_j(X_j) > 1-p
    } \in [0,1]
    ,
\end{align}
defined whenever the limit exists and where $F_i$ is the distribution function of $X_i$. When $\chi_{ij} > 0$ we say that $X_i$ and $X_j$ are  asymptotically dependent, and when $\chi_{ij} = 0$ we speak of asymptotic independence.
In the former case,
there are two different,
but closely related approaches for modelling 
extremal dependence:
through component-wise maxima of independent copies of $X$ leading to max-stable distributions \citep{deHaan1977};
and through threshold exceedances of $X$ resulting in multivariate generalized Pareto distributions \citep{rootzen2006}.
Here,
we concentrate on the threshold exceedance approach since it 
is
well-suited for graphical modelling \citep{engelke2020, segers2020}.
For statistical models for asymptotic independence, we refer to \cite{heffernan2004}, for instance, and to \cite{papastathopoulos2017} in the context of extremes of Markov chains as well as \cite{casey2023decomposable} for extremes in decomposable graphical models.

To make abstraction of the univariate marginal distributions and 
concentrate on the extremal dependence,
it is usually assumed that all 
variables $X_i$ follow
a given continuous distribution.
Throughout, 
we use standard exponential margins, that is, $\P{X_i \leq x} = 1- \exp(-x)$ for $x \geq 0$ and $i \in V$.
Let $\zeroVec$, $\oneVec$, and $\inftyVec$ denote vectors of adequate size
with all elements equal to~$0$, $1$, and~$\infty$ respectively.
A random vector $Y = (Y_1,\dots, Y_d)$ is said to follow a multivariate generalized Pareto distribution \citep{rootzen2006} if for any $z \in \L = \set{x \in \Rd: x \not\leq \zeroVec}$, we have
\begin{align}
    \label{eq:mpd}
    \P{Y \leq z}
    :=
    \limit{u} \P{
        X - u\oneVec \leq z
        \mid
        X \not\leq u \oneVec
    }
    =
    \frac{
        \LambdaVec\lr{z \land \zeroVec} - \LambdaVec\lr{z}
    }{
        \LambdaVec\lr{\zeroVec}
    }
    , \quad
\end{align}
for some random vector $X$, which is then said to be in the domain of attraction of~$Y$; the simple normalization by subtracting $u\oneVec$ comes from the assumption of exponential margins of $X$. 
Multivariate generalized Pareto distributions are threshold-stable
in the sense that for a vector $a \geq \zeroVec$,
the conditional random vector $Y - a$ given $Y \nleq a$
is again \multGP{}.
This also implies useful stochastic representations; see \citet{rootzen2018} for details.
The so-called exponent measure $\Lambda$ is a 
measure on $[-\infty, \infty)^d \setminus \{-\inftyVec\}$
that is finite on sets bounded away from $-\inftyVec$,
and we write $\LambdaVec(z) := \Lambda\lr{[-\infty,\infty)^d \setminus [-\inftyVec, z]}$.
Multivariate generalized Pareto distributions are the only ones that can arise as limits of threshold exceedances as in~\cref{eq:mpd}.

We assume $\Lambda$ to be absolutely continuous
with respect to the $d$-dimensional Lebesgue measure
and let $\lambda$ denote its Radon--Nikodym derivative.
The set of valid exponent measure densities $\lambda$ is 
characterized by the following two properties:
\begin{alignat*}{3}
    \lambda\lr{y + t\oneVec}
    &=
    \expF{-t} \lambda\lr{y}
    , \qquad &&\forall
    t \in \Rd[], y \in \Rd
    ,
    \\
    \int_{y_i>0}
    \lambda\lr{y} \, \diff y
    &=
    1
    , \qquad &&\forall
    i \in V
    .
\end{alignat*}
Since the distribution of $Y$ is proportional to the restriction of $\Lambda$ to $\mathcal L$,
its density $f$ then also exists and is proportional to the exponent measure density $\lambda$ as $f\lr{y} = \lambda\lr{y} / \LambdaVec\lr{\zeroVec}$ for all $y \in \L$, since $\LambdaVec\lr{\zeroVec} = \Lambda(\mathcal L)$. 

When considering marginal distributions of \multGP{} distributions,
it is often more useful to work with the limit distributions that arise in \cref{eq:mpd}
by replacing $X$ with the sub-vector $X_I$, also in the conditioning event, for some non-empty $I \subset \set{\oneToX{d}}$.
The resulting ``\gP{} marginal distribution'', here denoted $Y\I$,
is supported on $\L_I = \setM{x \in \Rd[\abs{I}]}{x \not\leq \zeroVec}$,
and its corresponding exponent measure and density are the actual marginals
$\Lambda_I$ and $\lambda_I$.
Note that $Y\I$ is equal in distribution to $Y_I$ conditioned on the event $Y_I \not\leq \zeroVec$.
See \citet[Section~4]{rootzen2006} or \citet[Section~6]{rootzen2018} for a detailed discussion of these marginals.
Due to the choice of standard exponential univariate margins for $X$,
the univariate ``\gP{} marginals'' $Y\I[\set{i}] \deq (Y_i \mid Y_i > 0)$ are standard exponential as well.
For a discussion how to transform exponent measure densities between different choices of marginals, see \cref{remark:densitiesForDifferentMargins}.

\label{subsec:mgp}
\subsection{\HR{} distributions}
\label{subsec:HRdefinition}

A popular class of \multGP{} distributions is the class of \HR{} distributions,
which are parametrized by symmetric \CND{} matrices, here denoted by $\Gamma$ \citep{hueslerReiss1989, kabluchkoSchlatherDeHaan2009}.
The set of these so-called variogram matrices is defined as
\begin{align}
    \label{eq:defVariogramSet}
    \Dcal_d
    =
    \setM{
        \Gamma \in \Rdd
    }{
        \Gamma = \Gamma\T
        ,\;\;
        \diag{\Gamma} = \zeroVec
        ,\;\;
        v\T\Gamma v < 0
        \;\forall\,
        \zeroVec \neq v \perp \oneVec
    }
    .
\end{align}
Variogram matrices are closely related to covariance matrices
and can be constructed from a centered random vector $W$ with
covariance matrix $\Sigma$ as
\begin{align}
    \label{eq:FarrisTransformStochastic}
    \Gamma_{ij}
    =
    \E{}\lr{W_i - W_j}^2
    =
    \Sigma_{ii} + \Sigma_{jj} - 2 \Sigma_{ij},
    \qquad \forall i,j \in V := \set{\oneToX{d}}
    .
\end{align}
Notably, this so-called covariance transform \citep{farris1970}
is not injective.
One collection of covariance matrices corresponding to a given $\Gamma$,
used for example in \cite{engelke2020},
can be obtained for $k \in V$ as
\begin{align}
    \label{eq:defPhiKSigmaTK}
    \SigmaT\k
    =
    \phikt{\Gamma}
    :=
    \half\lr{
        \Gamma_{ik} + \Gamma_{jk} - \Gamma_{ij}
    }_{i,j = \oneToX{d}}
    \in \Rdd
    .
\end{align}
The $k$th row and column of these matrices are identically zero.
Let $\Sigma\k \in \Rdd[(d-1)]$ denote
the matrices with these entries removed,
conveniently indexed by $V\without{k}$,
and $\phik{}$ the corresponding mapping.

\begin{definition}
    \label{def:HRdensity}
    Let $\Gamma \in \Dcal_d$ be a variogram matrix,
    $k \in V$,
    $\Sigma\k = \phik{\Gamma}$,
    and
    $\mu\k = \lr{-\half \Gamma_{ik} }_{i \neq k} \in \Rd[d-1]$.
    A \multGP{} random vector $Y = \vecDots{Y}{1}{d}$
    is \HRP{} distributed with parameter matrix $\Gamma$
    if its exponent measure density $\lambda$ satisfies
    \begin{align*}
        \lambda(y; \Gamma)
        &=
        \frac{
            \expF{-y_k}
        }{
            \sqrt{\lr{2\pi}^{d-1}\abs{\Sigma\k}}
        }
        \expF{
            -\half
            \norm[\Theta\k]{y_{V \without{k}} - \oneVec y_k - \mu\k}^2
        }
        ,
        \qquad
        y \in \Rd
        ,
    \end{align*}
    with $\Theta\k = \lr{\Sigma\k}\inv$ and using the notation $\norm[M]{v}^2 = v\T Mv$.
\end{definition}
Notably, this expression does not depend on the choice of $k$ \citep[e.g.,][Theorem~3.3]{engelke2014}.
The marginals $Y\I$ (in the \gP{} sense of \cref{subsec:mgp})
of a \HRP{} distributed vector with variogram matrix $\Gamma$
are also \HRP{} distributed with variogram matrix $\Gamma_{I\times I}$ \citep[Example~7]{engelke2020}.

The popularity of \HR{} distributions has theoretical and practical origins.
From known properties about this model class, it becomes obvious that the \HR{} family takes the same role among \multGP{} distributions as the Gaussian distribution in the non-extreme setting. The results of our paper will further shed light on this analogy. Moreover, \cite{lalancette2023pairwise} recently showed that \HR{} distributions constitute the only pairwise interaction models among \multGP{} distributions.
From a practical point of view, \HR{} distributions have become a standard model in multivariate extreme value modeling because of their flexibility and tractability. Importantly, they are the finite dimensional distributions of the widely used spatial Brown--Resnick \citep{kabluchkoSchlatherDeHaan2009} processes. We refer to \cite{dav2012b} and \cite{engelke2021a} for overviews of applications of these models.

\subsection{Extremal conditional independence}
\label{subsec:extremalCondInd}

In classical statistics,
graphical models are defined through conditional independence relations between
components of a random vector \citep[e.g.,][]{lauritzen1996}.
Since the support $\L$ of a \multGP{} distribution is not a product space,
this definition 
is impractical
in this case.
Instead, \cite{engelke2020} introduce a novel notion of extremal conditional independence
in terms of the vectors $Y\k = (Y_{1}^{(k)},\ldots,Y_{d}^{(k)})$ for $k \in V  := \set{\oneToX{d}}$,
defined as the \multGP{} vector $Y$ conditioned on the event $\set{Y_k > 0}$,
and 
supported on the product spaces $\L\k = \setM{x \in \L}{x_k>0} = \setM{x \in \Rd}{x_k > 0}$.
The extremal version of conditional independence between sub-vectors $Y_A$ and $Y_B$ given $Y_C$,
for non-empty, disjoint sets $A \cup B \cup C = V$,
is then defined as
\begin{align}\label{eci}
    \forall k \in V: \quad \condInd{Y\k_A}{Y\k_B}{Y\k_C}
    ,
\end{align}
and it is denoted by $\exInd{Y_A}{Y_B}{Y_C}$.

We consider the index set $V$ as a set of nodes and let $E\subseteq \allEdges$ be the set
of undirected edges of a connected graph $G = (V,E)$,
where $\allEdges$ denotes the set of all possible edges.
\cref{fig:exampleGraphs} shows four different graph structures in increasing generality;
see \cref{sec:graphTheory} for definitions of the related terms.
We say that $Y$ follows an extremal graphical model on $G$ if for $i,j \in V$ with $i \ne j$ we have
\begin{align*}
    (i,j) \notin E
    \quad \implies \quad
    \exInd{Y_i}{Y_j}{Y_{V\without{i,j}}}
    .
\end{align*}
\cite{engelke2020} show that the existence of a positive and continuous
exponent measure density $\lambda$ of $Y$ implies
connectedness of the corresponding graph $G$
and
the equivalence of~\cref{eci} to the factorization
\begin{align*}
    \lambdaM{y} \, \lambdaM[C]{y}
    =
        \lambdaM[A \cup C]{y} \,
        \lambdaM[B \cup C]{y},
    \quad y \in \L .
\end{align*}
If $Y$ is an extremal graphical model on a decomposable graph $G$, this result, applied recursively, yields a Hammersley--Clifford theorem that states that
the exponent measure density of $Y$
factorizes on the graph into lower-dimensional marginal densities;
for an extension to unconnected graphs see~\cite{eng_iva_kir}. 

The extremal conditional independence $\exInd{Y_A}{Y_B}{Y_C}$ has an intuitive interpretation in terms of classical conditional independence. Suppose that we observe an extreme event in variables corresponding to $C$, that is, $\max_{j\in C} Y_j > 0$, then, conditionally on $Y_C$, the sub-vectors $Y_A$ and $Y_B$ are independent in the usual sense. 
For an extremal graphical structure $G$, this translates into an extremal local prediction property of the form
\[
	\lr{Y_i \mid Y_{\delta(i)} = y_{\delta(i)}}
	\stackrel{d}{=}
	\lr{Y_i \mid Y_{\without{i}} = y_{\without{i}}},
	\qquad
	y \in \Rd
    , \;
    \max_{j \in \delta(i)} y_j > 0
    ,
\]
where $\delta(i)$ denotes the neighborhood of vertex $i$ in $G$ (see \cref{sec:graphTheory}). Thus, knowing that an extreme event has occurred in the neighborhood of $Y_i$, is suffices to know the neighbors of $Y_i$ for predicting its value. 

\begin{figure}
    \centering
    \begin{subfigure}[b]{.19\textwidth}
        \inputTikz{G_02_graph}
        \setlength\abovecaptionskip{-1\baselineskip}
        \caption{}
        \label{subfig:exampleTree}
    \end{subfigure}
    \begin{subfigure}[b]{.19\textwidth}
        \inputTikz{G_03_graph}
        \setlength\abovecaptionskip{-1\baselineskip}
        \caption{}
        \label{subfig:exampleBlock}
    \end{subfigure}
    \begin{subfigure}[b]{.19\textwidth}
        \inputTikz{G_04_graph}
        \setlength\abovecaptionskip{-1\baselineskip}
        \caption{}
        \label{subfig:exampleDecomposable}
    \end{subfigure}
    \begin{subfigure}[b]{.19\textwidth}
        \inputTikz{G_05_graph}
        \setlength\abovecaptionskip{-1\baselineskip}
        \caption{}
        \label{subfig:exampleGeneral}
    \end{subfigure}
    \caption{
        Different undirected, connected graphs on four nodes:
        (\subref*{subfig:exampleTree}) is a tree,
        (\subref*{subfig:exampleBlock}) is a block graph,
        (\subref*{subfig:exampleDecomposable}) is decomposable,
        (\subref*{subfig:exampleGeneral}) is non-decomposable.
    }
    \label{fig:exampleGraphs}
\end{figure}

\section{\HR{} precision matrix}
\label{sec:HR}
\label{sec:HrPrecisionMatrix}

An important property of Gaussian graphical models with covariance matrix $\Sigma$ is that
the conditional independence structure
can be read off from the zeros in the (Gaussian) precision matrix
$\Sigma\inv$.
\HR{} graphical models satisfy a similar property. Indeed, let $Y$ be a \HRP{} vector with parameter matrix $\Gamma \in \Dcal_d$ that is an extremal graphical model on the undirected graph $G = (V,E)$. \cite{engelke2020} show that the graph $G$ is necessarily connected and that the graphical structure can be read off from the precision matrices
$\Theta\k = \tlr{\Sigma\k}\inv$
for any $k \in V$; in the following we index these $(d-1)\times(d-1)$ matrices by $V \setminus \{k\}$ for the sake of simpler notation.
Two distinct nodes $i, j \neq k$ are extremal conditionally independent in the sense that $\exInd{Y_i}{Y_j}{Y_{V\without{i,j}}}$
if and only if the corresponding entry $\Theta\k_{ij}$ is zero.
If one of the nodes, say $j$, is equal to $k$,
extremal conditional independence is equivalent to the row sum $\sum_{l\neq k} \Theta\k_{il}$ being zero.

Each $\Theta\k$ thus contains all information on conditional independence of $Y$.
A natural question,
already raised in the discussion of \cite{engelke2020},
is whether there is a single $d \times d$ precision matrix
that contains this information but does not require a choice of $k$.
In fact, such a matrix can be defined as follows.

\begin{definition}
    \label{def:Theta}
    Let $\Gamma \in  \Dcal_d$ be a variogram matrix. For $d\geq 3$
    define $\Theta \in \Rdd$ by
    \begin{align*}
        \Theta_{ij}
        =
        \Theta\k_{ij}
        , \qquad
        \forsome
        k \in V\setminus \{i,j\},
    \end{align*}
    with $\Theta\k$ as above.
    In the case $d=2$,
    set $\Theta_{11} = \Theta_{22} = -\Theta_{12} = -\Theta_{21} = 1/\Gamma_{12}$.
\end{definition}

Lemma~1 and Proposition~3 in \cite{engelke2020}
imply that the matrix $\Theta$ is well-defined and represents extremal conditional independence of distinct nodes $i,j \in V$ in the corresponding \HR{} model by
\begin{align}
    \label{theta0}
    \Theta_{ij} = 0 \quad \iff  \quad \exInd{Y_i}{Y_j}{Y_{V\without{i,j}}}.
\end{align}
While \cref{def:Theta} is a natural way to jointly represent the information contained in the $\Theta\k$ matrices,
it remains to be shown that the matrix $\Theta$
allows for useful mathematical representations.
In order to give a first such representation,
let $\lr{\,\cdot\,}\pinv$ denote the
\MP{} inverse (see \cref{subsec:moorePenrose} for details),
let $\Id$ be the $d\times d$ identity matrix,
write $\eV = d\inv \oneVec$,
and let $\Proj = \Id - \oneVec\eV\T$ be the $d \times d$ centering matrix, i.e., the projection matrix onto the orthogonal complement of $\oneVec$.
Recall \cref{eq:FarrisTransformStochastic},
which can be written as a mapping for arbitrary matrices $S \in \Rdd$ as
\begin{align}
    \label{eq:FarrisTransform}
    \farris{}:
    S
    \longmapsto
    \Gamma
    =
    \oneVec \diag{S}\T
    +
    \diag{S} \oneVec\T
    -
    2 S
    ,
\end{align}
see \cref{def:gamma,lemma:gammaMapsToCND} for details.

\begin{proposition}
    \provenLabel{prop:ThetaReps}
    Let $\Gamma \in  \Dcal_d$
    and $S \in \Rdd$ satisfy $\farris{S} = \Gamma$.
    Then
    \begin{align}
        \label{eq:conditionThetaReps2}
        \ID S \ID
        =
        \ID \lr{- \half \Gamma} \ID
        .
    \end{align}
    Furthermore,
    the matrix $\Theta$ from \cref{def:Theta} satisfies
    \begin{align}
        \Theta
        &= \lr{\Proj S \Proj}\pinv
        \label{eq:ThetaRepPinv}
        \\
        &= \limit{t} \lr{t \oneVecTT + S}\inv.
        \label{eq:ThetaRepLimit}
    \end{align}
\end{proposition}

Since the mapping $\gamma$ preserves (a)symmetry
and $\Gamma$ is symmetric,
the condition $\farris{S} = \Gamma$
requires $S$ to be a symmetric matrix as well.
It turns out that,
restricted to symmetric matrices,
\cref{eq:conditionThetaReps2}
is in fact equivalent to $\farris{S} = \Gamma$;
see \cref{lemma:farrisSeqG}.

In practice,
the matrix $S$ is usually a (definite or semi-definite)
covariance matrix,
so it is natural to only consider symmetric matrices here.
For the sake of completeness 
we show in the proof of \cref{prop:ThetaReps}
that \cref{eq:conditionThetaReps2}
is a sufficient and necessary
condition for \cref{eq:ThetaRepPinv},
also allowing asymmetric matrices $S$.
\cref{lemma:SisGeneralizedInverse} shows that \cref{eq:conditionThetaReps2}
is further equivalent to $S$ being a generalized inverse of $\Theta$
in the sense $\Theta S \Theta = \Theta$.

In light of \cref{eq:conditionThetaReps2}, %
an obvious choice of $S$ for a given $\Gamma$ is setting $S=-\half \Gamma$,
while other interesting choices for $S$ are the matrices $\SigmaT\k = \phikt{\Gamma}$, defined in \cref{eq:defPhiKSigmaTK}.
Furthermore, these matrices satisfy 
$\ID \lr{-\half\Gamma} \ID
    =
    d^{-1} \sum_{k=1}^d \tilde\Sigma^{(k)} - t(\Gamma) \oneVecTT$
where $t(\Gamma) = \half d^{-2} \oneVec\T \Gamma \oneVec$ is the largest scalar $t$ such that
$d\inv \sum_{k=1}^d \tilde{\Sigma}^{(k)} - t \oneVec \oneVec\T$ is \PSD{}.
Further connections between $\Gamma$, $\tilde{\Sigma}^{(k)}$ and $\Theta$ are investigated in \citet[Propositions~3.1 and~3.2]{wan2023graphical}.

To better understand representation~\cref{eq:ThetaRepPinv},
let $\Pcal_d^\oneVec \subset \Rdd$ denote the set of symmetric \PSD{} matrices with kernel equal to $\laSpan{\set{\oneVec}}$.
Let $\sigma$ and $\theta$ be the mappings from a variogram matrix $\Gamma$
to the corresponding matrices $\Sigma$ and $\Theta$, respectively:
\begin{equation}
    \begin{aligned}
        \label{eq:thetaSigma}
        \sigma: \;
        &\Gamma \longmapsto
        \Sigma := \Proj \lr{-\half\Gamma} \Proj
        , \\
        \theta: \;
        &\Gamma \longmapsto
        \Theta := \Sigma\pinv = \sigma(\Gamma)\pinv = \lr{\ID\lr{-\half\Gamma}\ID}\pinv
        .
    \end{aligned}
\end{equation}
\begin{proposition}
    \provenLabel{prop:thetaHomeo}
    The mappings $\sigma$ and $\theta$ are homeomorphisms between $\Dcal_d$ and $\Pcal_d^\oneVec$.
    Using $\gamma(\cdot)$ defined in \cref{eq:FarrisTransform}, their inverses are
    \begin{alignat*}{3}
        \sigma\inv\lr{\Sigma}
        &=
        \farris{\Sigma}
        , \qquad
        \theta\inv\lr{\Theta}
        &=
        \farris{\Theta\pinv}
        .
    \end{alignat*}
\end{proposition}

From this proposition it follows that
the matrix $\Theta$ from \cref{def:Theta} is always \PSD{} with kernel equal to $\laSpan{\set{\oneVec}}$.
Furthermore, the class of \HR{} distributions,
parametrized by the set $\Dcal_d$ of \CND{} matrices,
can just as well be parametrized by $\Pcal_d^\oneVec$,
interpreted either in the role of $\Theta$ or $\Sigma$.
As discussed after \cref{prop:HrDensitySimple}, the matrix
$\Sigma$ is the degenerate covariance matrix of a particular transformation of the \HRP{} vector $Y$.
Similarly to the Gaussian case, the precision matrix $\Theta$ can be obtained
from this covariance matrix as $\Theta = \Sigma\pinv$,
using the \MP{} inverse due to its singularity.

For a given $\Gamma$,
the characterizations in \cref{prop:ThetaReps} give a straightforward way to compute $\Theta$ in \cref{def:Theta} as $\Theta = \theta(\Gamma) = \lr{\ID\tlr{-\half\Gamma}\ID}\pinv$
and thus retrieve the conditional independence structure of the corresponding extremal graphical model.
In a Gaussian model,
another property of the precision matrix is
that it can be used to express its probability density function in a concise way.
The following result shows that a similar expression is also possible for \HR{} distributions.

\begin{proposition}
    \provenLabel{prop:HrDensitySimple}
    Let $\Gamma \in \Dcal_d$ and $\Theta = \theta\lr{\Gamma}$.
    Then the exponent measure density $\lambda(\,\cdot\,;\Gamma)$
    from \cref{def:HRdensity} can be expressed as
    \begin{align}
        \lambda(y; \Theta)
        &=
        c_1
        \cdot
        c_2
        \cdot
        \expF{
            -\half
            y\T \Theta y
            +
            y\T r_\Theta
            -
            y\T\eV
        }
        , \qquad
        y \in \Rd
        ,
    \end{align}
    with $r_\Theta = -\half \Theta\Gamma\eV$,
    $c_1 = \lr{2\pi}^{-(d-1)/2}\lr{d\inv\pdet{\Theta}}^{1/2}$,
    and
    $c_2 = \expF{ - \tfrac{1}{8} \eV\T \lr{ \Gamma \Theta \Gamma + 2\Gamma } \eV }$.
\end{proposition}

This expression is remarkably similar to 
a degenerate Gaussian density on the hyperplane $\set{\oneVec}^\perp$,
differing only by the value of the constant and the term $\expF{-y\T\eV}$.
In fact,
when restricting the support to the half-space $\setM{x \in \Rd}{\oneVecT x \geq 0}$,
this density is proportional to the probability density function of a random variable
$Y_\oneVec = W_\Sigma + R \oneVec$,
with $R \sim \ExpF{1}$
and $W_\Sigma$ a degenerate normal vector with covariance matrix
$\Sigma = \Proj \lr{-\half\Gamma} \Proj$, independent of $R$.

A careful look at other characterizations of $\lambda$ such as in \cite{wadsworthTawn2014},
or the slightly more general definitions of the \HR{} exponent measure density
in \cite{hoDombry2017} and \cite{kiriliouk2018POT},
shows that the precision matrix $\Theta$ appears naturally in these parameterizations as well.

\begin{remark}
    \label{remark:otherApplicationsTheta}
    Besides the extremal conditional independence structure and exponent measure density,
    the precision matrix $\Theta$ is 
    also useful to describe other stochastic properties of \HR{} distributions.
    Based on the results of the present paper,
    \cite{roettger2023} and \cite{roettgerSchmitz2023}
    show that multivariate total positivity of order two, a notion of positive dependence, can be encoded as 
    $\Theta_{ij} \leq 0$ for all $i\neq j$.
    \cite{engelke2022a} suggest that the precision matrix can be used to estimate extremal graphical structures by penalizing its $L^1$-norm $\| \Theta \|_1$,
    and a similar approach is taken in \cite{wan2023graphical} and \cite{lederer2023extremes}.
    Compared to the Gaussian case,
    a difficulty of dealing with $\Theta$ in mathematical derivations and statistical implementations is that it is
    not invertible.
\end{remark}

\section{Matrix completion problems}
\label{sec:matcomprob}
\label{sec:completion}
A well-studied problem related to Gaussian graphical models is
the task of constructing
a model with a given graphical structure and specified marginal distributions on the fully connected
subsets of vertices,
or equivalently, to complete a partially defined covariance matrix
such that its precision matrix has zeros where there are no edges in the specified graph
\citep{speedKivery1986,bakonyi2011}.
For extremal \HR{} models,
the same problem can be posed,
expressed as a matrix completion problem on the variogram and precision matrix.

To formalize the notion of a partially specified matrix, we define the set
$\ROd[] := \Rd[] \cup \set{\undefined}$,
consisting of the real numbers and the placeholder
$\undefined$ for unspecified values
\citep[see e.g.,][for this use of $\undefined$]{bakonyi2011}.
For an undirected graph $G=(V,E)$,
a matrix is said to be \quotes{specified on $G$}
if it is specified on the diagonal
and the entries corresponding to the edges of $G$.
A matrix is said to be \quotes{\PCND{}}
if it is symmetric, its diagonal is fully specified,
and all fully specified principal submatrices are \CND{}
(a principal submatrix is any submatrix obtained by removing the same index set from both the columns and rows of the matrix).
In computations involving partially specified matrices,
an entry in the result is $\undefined$
as soon as any of the entries used to compute it
is itself $\undefined$.
For definitions of graph theoretical concepts used in this section,
we refer to \cref{sec:graphTheory}.

\begin{definition}
    \label{def:matrixCompletionProblem}
    Let $G=(V,E)$ be an undirected graph and let $\mathring{\Gamma} \in \ROdd$ be a partially conditionally negative definite matrix, specified on $G$. The corresponding matrix completion problem is to
    find a \CND{} matrix $\Gamma \in \Dcal_d$
    and $\Theta = \ppinv{\Gamma}$ such that
    \begin{gather}
    \begin{alignedat}{2}
        \label{eq:matrixCompletionProblem}
        \Gamma_{ij}
        &=
        \GammaO_{ij}
        \qquad &&
        \forall (i,j) \in \Ebar
        ,
        \\
        \Theta_{ij}
        &=
        0
        \qquad &&
        \forall (i,j) \notin \Ebar
        ,
    \end{alignedat}
    \end{gather}
    where $\Ebar$ denotes the edge set $E$
    augmented by the diagonal entries $\setM{(i,i)}{i \in V}$.
\end{definition}

An example for such a matrix completion problem on the graph in
\autoref{subfig:exampleBlock} is given in \cref{eq:simpleGammaThetaPartial}.
In the Gaussian case, a similar problem can be posed with the
covariance matrix $\Sigma$ instead of $\Gamma$ and with $\Theta = \Sigma\inv$.
To this \PD{} completion problem,
an explicit solution for decomposable graphs
and a convergent algorithm for general graphs is given
in \cite{speedKivery1986}.
In this section we discuss semi-definite matrix completion problems for \HR{} models as in \cref{def:matrixCompletionProblem}, starting from simple graph structures such as trees and finishing
with general graphs.
We assume throughout that the graph $G$ is connected since only those can be associated to non-degenerate \HR{} models;
see \cref{subsec:extremalCondInd}.

\begin{example}
Block graphs are simple graph structures where the separator sets contain only single nodes.
Trees are a special case of this class where 
all cliques consist of exactly two nodes.
The particular structure of block graphs makes them appealing for the construction of parametric models and for statistical inference \citep{engelke2020, asenova2021inference}. In fact, this structure also yields a simple explicit solution to 
the matrix completion problem in \cref{def:matrixCompletionProblem}.
If $\GammaO$ is a \PCND{} %
matrix on the connected block graph $G = (V,E)$,
then a unique completion exists \citep[][Proposition~4]{engelke2020} and can be expressed as
\begin{align}
    \label{eq:tree_metric}
    \Gamma_{ij}
    =
    \sum_{(s,t) \in \text{path}(i,j)} \GammaO_{st}
    ,
\end{align}
where $\text{path}(i,j)$ is the unique shortest path between $i$ and $j$ in $G$; see also \cite{engelke2020a} and \cite{asenova2021extremes}.
Using this result, the missing entries in \cref{eq:simpleGammaThetaPartial}
can be computed to be $\Gamma_{13} = \Gamma_{31} = 13$ and $\Gamma_{34} = \Gamma_{43}=12$.
\end{example}

\subsection{Decomposable graphs}

In this section, a solution to the matrix completion problem will be given for connected, decomposable graphs.
First, consider the simplest (non-trivial) example from this class of graphs, a graph consisting of exactly two cliques.
Recall
$\Sigma\k$, $\SigmaT\k$ and $\varphi_k$ from \cref{eq:defPhiKSigmaTK},
$\gamma$~from \cref{eq:FarrisTransform},
and observe that %
$\varphi_k\inv\tlr{\Sigma\k} = \gamma\tlr{\SigmaT\k}$.

\begin{lemma}
    \provenLabel{lemma:completeTwoCliquesGamma}
    Let $G=(V,E)$ be a connected decomposable graph consisting of two cliques $C_1$ and $C_2$,
    separated by $D_2 = C_1 \cap C_2 \neq \emptyset$.
    Let $\GammaO$ be a \PCND{} matrix, specified on $G$.
    For some $k\in D_2$, let $\SOk = \phik{\GammaO}$
    and let $\Sigma\k$ be its unique \PD{} completion with graphical structure $\restrictTo{G}{V\without{k}}$.
    Then
    \begin{align*}
        \Gamma
        &:=
        \varphi_k\inv
        \tlr{
            \Sigma\k
        }
        =
        \gamma\tlr{\SigmaT\k}
        ,
    \end{align*}
    is the unique solution of the matrix completion problem~\ref{def:matrixCompletionProblem}
    for $\GammaO$ and graph $G$.
\end{lemma}
The existence and uniqueness of such a \PD{} completion for general graph structures is shown in \cite{speedKivery1986},
and a short proof of the case with two cliques is given in \cref{lemma:completeTwoCliques}.
Using this result, a completion for a general decomposable graph $G$
can be constructed as follows.
Let $\set{C_1, \dots, C_N}$ be the cliques of $G$,
ordered according to the running intersection property (cf. \cref{sec:graphTheory}),
and, without losing generality, assume that the vertices in $V$ are ordered accordingly,
in the sense $i \leq j$ for all $i \in C_k$ and $j \in C_l$ with $k < l$.
Let $\Gamma_1 = \restrictTo{\GammaO}{C_1}$,
$V_n = C_1 \cup \hdots \cup C_n$,
and $K_n = V_n \times V_n$.
Further, for $n = 2, \dots, N$, iteratively define
$\GammaO_n \in \ROdd[\abs{V_n}]$ with entries
\begin{alignat*}{2}
    \tlr{\GammaO_n}_{ij}
    &=
    \myCases{
        \myCase{
            \lr{\Gamma_{n-1}}_{ij}
        }{
            (i,j) \in K_{n-1}
            ,
        }
        \\
        \myCase{
            \GammaO_{ij}
        }{
            (i,j) \in \restrictTo{\Ebar}{V_n} \backslash K_{n-1}
            ,
        }
        \\
        \myCase{
            \undefined
        }{
            \text{otherwise}
            ,
        }
    }
\end{alignat*}
and $\Gamma_n = \comp{}\tlr{\GammaO_n} \in \Rdd[\abs{V_n}]$,
where $\comp{}$ denotes the completion from \cref{lemma:completeTwoCliquesGamma}.

\begin{proposition}
    \provenLabel{prop:completeDecomposable}
    Let $G=(V,E)$ be a connected, decomposable graph and $\GammaO$
    a \PCND{} matrix, specified on $G$.
    Then there exists a unique solution to the matrix completion problem
    \ref{def:matrixCompletionProblem} with $\GammaO$ and $G$.
    This solution can be computed explicitly
    as $\Gamma_N$ in the construction above.
\end{proposition}
The class of decomposable graphs is a significant extension of the class of block graphs.
For instance,
a non-decomposable graph can be approximated in a non-trivial fashion
by computing a so-called decomposable completion.
In contrast, the only block graph completion of any biconnected graph
(i.e., a graph that remains connected after removal of any one vertex)
is already the complete graph,
since the only biconnected block graph is the complete graph,
and adding edges to a graph preserves connectivity.
See for instance \citet[][Section~5.3]{lauritzen1996} for further details about decomposable graphical models.

\begin{example}
	\label{ex:exampleDecomposableCompletion}
    \cref{fig:exampleDecomposableCompletion}
    shows an example of the matrix completion algorithm from \cref{prop:completeDecomposable}.
    Starting with the clique $\set{1, 2, 3}$,
    the missing values are computed clique by clique.
    Edges whose corresponding matrix entries were already computed,
    are considered as part of the graph in subsequent steps,
    such that each computation is a direct application of \cref{lemma:completeTwoCliquesGamma}.
    Thanks to the running intersection property,
    the required conditional independence structure is preserved.
\end{example}

\begin{figure}
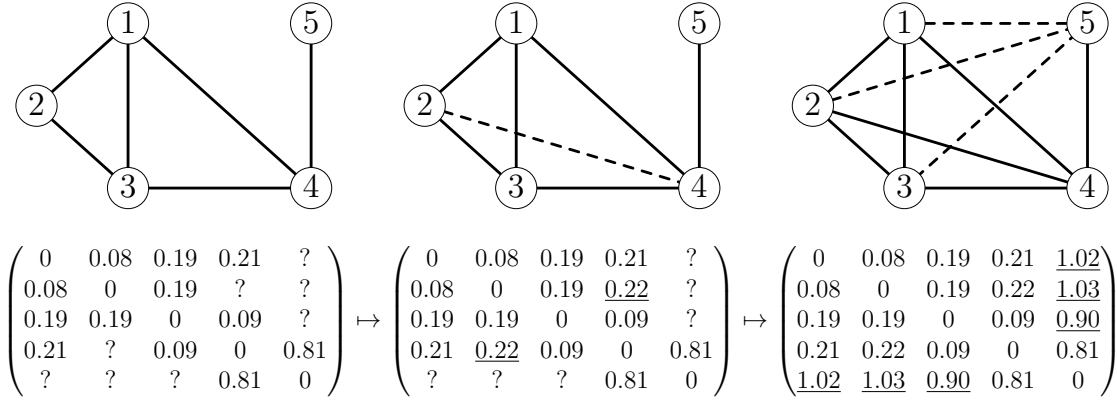

    \center
    \begin{minipage}{.3\linewidth}
        \inputTikz{GD_01_graph}
    \end{minipage}
    \hspace{10pt}
    \begin{minipage}{.3\linewidth}
        \inputTikz{GD_02_graph}
    \end{minipage}
    \hspace{10pt}
    \begin{minipage}{.3\linewidth}
        \inputTikz{GD_03_graph}
    \end{minipage}
    \\
    \vspace{-15pt}
    \noindent\resizebox{1\textwidth}{!}{
        \begin{minipage}{1\textwidth}
            \begin{align*}
                \lr{\input{examples/GD_01_matrix.tex}}
                \mapsto
                \lr{\input{examples/GD_02_matrix.tex}}
                \mapsto
                \lr{\input{examples/GD_03_matrix.tex}}
            \end{align*}
        \end{minipage}
    }
    \caption{
        Illustration of \cref{ex:exampleDecomposableCompletion}.
        On the left-hand side is the initial partial matrix $\GammaO$ and the corresponding decomposable graph.
        In each of the following steps, dashed edges (top) correspond to newly computed matrix entries (underlined, bottom).
    }
    \label{fig:exampleDecomposableCompletion}
\end{figure}

\subsection{General graphs}

The class of general connected graphs is much larger than the class of decomposable ones; see \autoref{subfig:exampleGeneral} for an example. In applications, when the graph is estimated from data without restrictions, non-decomposable structures often arise.
The following results provide a more general
but slightly weaker solution to the matrix completion problem in \cref{def:matrixCompletionProblem},
where, as before, $\Ebar$ is equal to $E \cup \setM{(i, i)}{i \in V}$.

\begin{proposition}
    \provenLabel{prop:completeGeneral}
    Let $G = (V,E)$ be a connected graph
    and let
    $\GammaO \in \ROdd$ be specified on $G$,
    such that there exists a fully specified \CND{} matrix
    that agrees with $\GammaO$ on the entries $(i,j) \in \Ebar$.
    Then there exists a unique \CND{} matrix $\Gamma$
    that solves the matrix completion problem
    \cref{eq:matrixCompletionProblem}
    for $\GammaO$ and $G$.
\end{proposition}

This result provides the same theoretical existence of a unique solution
as \cref{prop:completeDecomposable}
and allows the definition of the following mapping.
We use the notation $\restrictTo{\Gamma}{G} \in \ROdd$
to denote the restriction of a fully specified matrix $\Gamma$
to a graph $G = (V, E)$,
in the sense
\begin{align}
    \label{eq:defRestrictToGraph}
    \lr{\restrictTo{\Gamma}{G}}_{ij}
    &=
    \myCases{
        \myCase{\Gamma_{ij}}{(i,j) \in \Ebar,} \\
        \myCase{\undefined}{(i,j) \notin \Ebar.}
    }
\end{align}

\begin{definition}
    \label{def:notationComp}
    For a connected graph $G=(V,E)$ let
    $
        \mathring\Dcal_G
        =
        \setM{
            \lr{\restrictTo{\Gamma'}{G}}
        }{
            \Gamma' \in \Dcal_d
        }
    $
    be the restriction of \CND{} matrices to $G$.
    The function
    \begin{alignat*}{3}
        \myFunctionDefinitionInline{\comp[G]{}}{
            \mathring\Dcal_G
        }{
            \Dcal_d
        }{
            \GammaO
        }{
            \Gamma
        },
    \end{alignat*}
    maps a partial matrix $\GammaO \in \mathring\Dcal_G$
    to its unique completion $\Gamma$
    satisfying the matrix completion problem in \cref{def:matrixCompletionProblem}
    with respect to $G$.
\end{definition}

The uniqueness of this completion
can be understood intuitively by counting the number of parameters and constraints involved in \cref{eq:matrixCompletionProblem}.
Notably, the existence of any \CND{} completion of $\GammaO$
is sufficient for the existence of a graphical completion.
In the decomposable case, this can be guaranteed by verifying the
definiteness of all fully specified principal submatrices,
but in the general case this criterion does not work,
as the following counter-example shows.
\begin{example}
    \provenLabel{exmp:noCompletion}
    For $d \geq 4$, let $\GammaO \in \ROdd$ be a partial matrix on the $d$-dimensional ring graph with entries
    $\GammaO_{ij} = 1$ if $\abs{i-j}=1$, $\GammaO_{1d} = \GammaO_{d1} = (2d)^2$,
    zeros on the diagonal, and $\undefined$ elsewhere.
    Then all fully specified principal submatrices of this $\GammaO$ are \CND{},
    but there exists no \CND{} completion of the entire matrix
    (see \cref{subsubsec:proofNonCompletableGraphs} for details).
\end{example}

In general, computing the completion is less straightforward than
in the decomposable case,
but the following algorithm follows from the proof of \cref{prop:completeGeneral}.
This procedure is similar to the one described in \citet[Section~4]{speedKivery1986},
using the above results instead of \PD{} matrix completions.
\begin{corollary}
    \label{cor:completeGeneral}
    Let $G$ and $\GammaO$ be as in \cref{prop:completeGeneral},
    and let $\Gamma_0\in \mathcal D_d$ be such that $\restrictTo{\lr{\Gamma_0}}{G} = \GammaO$.
    For some fixed $m \in \mathbb N$,
    let $G_i = (V, E_i)$, $i=1, \dots, m$, be decomposable graphs
    such that $\bigcap_i E_i = E$ and define $(\Gamma_n)_{n \ge 1}$ recursively by
    \begin{align*}
        \Gamma_n
        &=
        \comp[G_{t_n}]{
            \restrictTo{\lr{\Gamma_{n-1}}}{G_{t_n}}
        }
        ,
    \end{align*}
    with $t_n \in \set{\oneToX{m}}$, $t_n \equiv n \mod m$, and
    $\comp[G_{t_n}]{}$ computed as in \cref{prop:completeDecomposable}.
    Then $\Gamma_n$ converges to the unique completion $\comp[G]{\GammaO}$ as $n\to \infty$.
\end{corollary}

\begin{remark}
    \label{remark:generalCompletion}
    An easy way to construct a suitable set of graphs $\set{G_1,\ldots,G_m}$ is to
    put $\set{e_1, \hdots, e_m} = \allEdges\setminus E$
    and use
    $E_i = \allEdges\setminus\set{e_i}$, where $\allEdges$ denotes the set of all possible edges, see \cref{sec:graphTheory}.
    However, this choice leads to slow convergence
    since each iteration only updates one entry in $\Gamma$
    while requiring inversion of a matrix of dimension  $(d-3)\times(d-3)$.
    Better performance can be achieved by choosing the $G_i$ to be decomposable completions of $G$,
    with separator sets being as small as possible; see \cite{baz2022} for details and further optimizations.
\end{remark}

In order to apply \cref{cor:completeGeneral},
an initial (non-graphical) completion of $\GammaO$,
to be used as $\Gamma_0$, is required.
The problem of finding such a matrix
is also known as the Euclidean distance matrix completion problem,
with solution algorithms for example in
\cite{bakonyiJohnson1995} and \cite{fangEtAl2011}.

\begin{example}
    \label{exmp:completionGeneral}
    To illustrate the algorithm from \cref{cor:completeGeneral},
    consider the matrix $\Gamma_0$ below,
    which does not have any non-trivial graphical structure,
    and its \quotes{completion} $\Gamma$,
    whose conditional independence structure is
    described by the graph in \cref{fig:compConv}.
    \begin{align*}
        \Gamma_0 =
        \lr{
            \begin{matrix}
     0 & 0.23 & \underline{0.08} & 0.09 & 0.21 \\
     0.23 & 0 & 0.14 & \underline{0.23} & \underline{0.19} \\
     \underline{0.08} & 0.14 & 0 & 0.11 & \underline{0.20} \\
     0.09 & \underline{0.23} & 0.11 & 0 & 0.16 \\
     0.21 & \underline{0.19} & \underline{0.20} & 0.16 & 0
\end{matrix}
        }
        \mapsto
        \Gamma =
        \lr{
            \begin{matrix}
     0 & 0.23 & \underline{0.17} & 0.09 & 0.21 \\
     0.23 & 0 & 0.14 & \underline{0.20} & \underline{0.35} \\
     \underline{0.17} & 0.14 & 0 & 0.11 & \underline{0.26} \\
     0.09 & \underline{0.20} & 0.11 & 0 & 0.16 \\
     0.21 & \underline{0.35} & \underline{0.26} & 0.16 & 0
\end{matrix}
        }
    \end{align*}
    These two matrices differ only in the highlighted entries, corresponding to non-edges in $G$.
    During the computation of $\Gamma$,
    the entries of $\Gamma_0$ that correspond to edges in $G$ do not change,
    and the convergence of the remaining entries in $\Theta$ to zero is plotted in \cref{fig:compConv}.
    It can be seen that the entries corresponding to edges $(2,5)$ and $(3,5)$
    stay at zero (up to numerical precision of magnitude $10^{-15}$)
    after the first few iterations,
    whereas the maximum of entries $(2,4)$ and $(1,3)$ converges to zero at a slower rate.
    This observation is in line with the fact that vertices $\set{1, 2, 3, 4}$
    induce the chordless cycle that makes the graph non-decomposable.
\end{example}

\begin{figure}
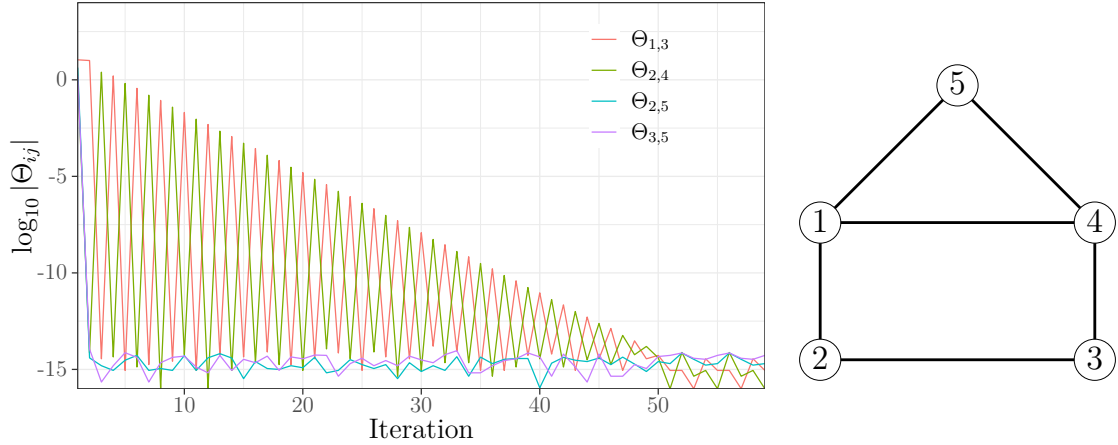

    \begin{subfigure}[c]{.68\textwidth}
        \inputTikz{completionConvergence}
    \end{subfigure}
    \begin{subfigure}[c]{.3\textwidth}
        \inputTikz{completionConvergenceGraph}
    \end{subfigure}
    \setlength\abovecaptionskip{-1\baselineskip}
    \caption{
        Illustration of \cref{prop:completeGeneral} in \cref{exmp:completionGeneral}.
        A non-decomposable graph $G = (V,E)$ (right),
        and the convergence of $\Theta_{ij}$ to zero for $(i,j) \notin E$
        as the graphical completion of $\Gamma$ is computed (left).
        In each iteration,
        the graph defined by edge set
        $E_1 = E \cup \set{(1,3)}$
        or
        $E_2 = E \cup \set{(2, 4)}$ 
        is completed using \cref{prop:completeDecomposable}.
    }
    \label{fig:compConv}
\end{figure}

\section{Statistical inference}
\label{sec:statinf}

Estimation of a \HR{} parameter matrix $\Gamma$ that is an extremal graphical model on a given graph $G=(V,E)$
is currently restricted to the simple structures of trees \citep{engelke2020a, hu2022}, latent trees \citep{asenova2021inference} and block graphs \citep{engelke2020, asenova2021extremes}. In the discussion of \cite{engelke2020}, this has been pointed out as too restrictive in practice, since many data sets require more general graph structures.

In this section we solve this issue by applying our results on matrix completions for variograms. In particular, we show how any consistent estimator $\widehat \Gamma$ can be transformed into a consistent estimator $\widehat \Gamma^G$ with desired graph structure.
When the graph is unknown, it can be recovered by existing structure learning methods $\widehat G$. Our completion then produces the first estimator $\widehat \Gamma^{\widehat G}$ that jointly estimates the graph and the variogram matrix for general graphs $G$ in a consistent way.

\subsection{Matrix completion as likelihood optimization}

For a more statistical perspective on the matrix completions in Section~\ref{sec:matcomprob}, we first characterize them as constrained maximum likelihood estimators; see \citet{uhler2017} for the Gaussian case.
For a variogram matrix $\overline{\Gamma} \in \Dcal_d$ and a connected graph $G=(V,E)$,
we will show that the maximizer of the \PSD{} Gaussian log-likelihood
\begin{align}\label{logdet}
  \log \pdet{\Theta} + \half \trace{\overline{\Gamma} \Theta},
\end{align}
under suitable graph constraints, is equal to our completion operator $\comp[G]{}$ from \cref{def:notationComp} applied to $\restrictTo{\overline{\Gamma}}{G}$.
The connection to likelihood estimation for \HR{} distributions arises by choosing for $\overline{\Gamma}$ the empirical variogram $\widehat \Gamma$ \citep{engelke2020a}. In this case,~\cref{logdet}~is the (surrogate) log-likelihood of the \HR{} model parameterized in terms of the precision matrix $\Theta$ \citep[][Section 5.1]{roettger2023}; see \cref{subsec:completionAsOpt} for a simple derivation. Recall $\Ebar = E \cup \setM{(i,i)}{i \in V}$ for $E \subseteq \allEdges$ as well as the map $\theta$ and its inverse in \cref{prop:thetaHomeo}.

\begin{proposition}
    \provenLabel{prop:completionOptProblem}
    Let $\overline{\Gamma} \in \Dcal_d$ be a variogram matrix and $G=(V,E)$ a connected graph.
    Then $\comp[G]{\restrictTo{\overline{\Gamma}}{G}} = \theta^{-1}\tlr{\overline{\Theta}_G}$ where $\overline{\Theta}_G$ is the unique maximizer of~\cref{logdet} over all \HR{} precision matrices $\Theta \in \Pcal_d^{\oneVec}$ under the constraint
    that $\Theta_{ij} = 0$ for all $(i,j) \notin \Ebar$.
\end{proposition}

We note that the result holds for any variogram matrix $\overline{\Gamma}$, but only for the empirical variogram $\widehat \Gamma$ there is an interpretation in terms of maximum (surrogate) likelihood estimation.
In practice,
solving this optimization provides an alternative to \cref{cor:completeGeneral} to compute the graphical completion of a matrix. We %
leave this approach for future research.

\subsection{Consistency}

In order to show consistency results, a useful property of the completion $\comp[G]{}$ is that it is a continuous mapping from the space $\mathring\Dcal_G$ of partially specified variogram matrices to the space of variogram matrices $\Dcal$.

\begin{lemma}
    \provenLabel{lemma:matrixCompletionContinuous}
    The mapping $\comp[G]{}$ from \cref{def:notationComp} is continuous.
  \end{lemma}

By continuity of $\comp[G]{}$, a consistent estimator on the edges of
a \HR{} graphical model can be extended to a consistent estimator
of the whole parameter matrix.

\begin{theorem}
    \provenLabel{thm:matrixCompletionConsistent}
    Consider a \HR{} graphical model with graphical structure $G = (V,E)$
    and variogram matrix $\Gamma$.
    Let $\widehat\GammaO \equiv \widehat\GammaO_n$
    be an estimator sequence of $\restrictTo{\Gamma}{G}$
    which satisfies $\diag{\widehat\GammaO} \equiv \zeroVec$,
    is symmetric,
    and is consistent
    in the sense that
    for all $\varepsilon > 0$, we have
    \begin{align}\label{partial_consist}
        \P{
            \max_{ (i,j) \in E }
            \abs{ \widehat\GammaO_{ij} - \Gamma_{ij} }
            <
            \varepsilon
        } \to 1
        , \quad
        n \to \infty
        .
    \end{align}
    Then with probability tending to one there exists a completion of $\widehat\GammaO$, that is, $\mathbb P( \widehat\GammaO \in \mathring\Dcal_G) \to 1$ as $n\to \infty$.
    Let $\widehat \Gamma^G = \comp[G]{\widehat\GammaO}$ denote this completion and $\widehat \Theta^G$ the corresponding precision matrix (and set both matrices to the zero matrix if the completion does not yet exist).
    Then $\widehat \Gamma^G$ is a consistent estimator for $\Gamma$ with the correct graph structure,
    that is, $\widehat\Theta^G_{ij} = 0$ if $(i,j) \notin \Ebar$ and for all $\varepsilon > 0$,
    \begin{align*}
        \P{
            \max_{(i,j) \in V \times V}
            \abs{\widehat \Gamma^G_{ij} - \Gamma_{ij}}
            <
            \varepsilon
        } \to 1
        , \quad
        n \to \infty
        .
    \end{align*}
\end{theorem}

\cref{thm:matrixCompletionConsistent} provides the first consistent estimator of a \HR{} variogram matrix that respects the structure of a general graph $G=(V,E)$.
Indeed, the only ingredient that is needed is a consistent estimator of $\Gamma$ on the edge set $E$.
There are many different possibilities for such estimators in the literature.
A natural and computationally efficient estimator is the empirical variogram \citep{engelke2014, engelke2020a}, which for the \HR{} distributions is the empirical version of the parameter matrix $\Gamma$.
Other proposals include M-estimators \citep{einmahl2012, einmahl2016, lalancette2021}, proper scoring rules \citep{deFondeville2018} and likelihood methods.
The latter have the advantage that they can incorporate censoring of components of the data that are not extreme \citep{ledfordTawn1996, wadsworthTawn2014}.
Pairwise likelihood methods \citep{padoan2010} reduce the high computational cost of censoring.

In order to guarantee that even for a fixed sample size $n$ there exists a completion, it can be advantageous to start with a consistent estimator $\widehat \Gamma \in \Dcal$ of the entire parameter matrix $\Gamma$. Since $\restrictTo{\widehat \Gamma}{G} \in \mathring\Dcal_G$ by definition, \cref{prop:completeGeneral} ensures that
\begin{align}
    \label{eq:completionFull}
    \widehat{\Gamma}^G
    =
    \comp[G]{\restrictTo{\widehat \Gamma}{G}}
\end{align}
exists and is a consistent estimator with graph structure $G$. In larger dimensions $d = |V|$, estimating all entries of $\Gamma$ by likelihood optimization or censoring is often infeasible, and the only option is then the empirical variogram or estimators based on simple summary statistics such as extremal coefficients \citep{einmahl2018continuous}.

If the graph $G$ is sparse, then there is an efficient alternative to estimating every entry of $\Gamma$, which requires only estimation on all cliques separately; see \cref{subsec:estimationOnSparseGraphs}.
Furthermore,
in \cref{sec:simulation}, we perform a simulation study that illustrates the performance gains achieved by this method,
while yielding a similar quality of estimation of $\Gamma$.

The graph $G$ is in practice often unknown and has to be estimated from the data.
Consistent structure estimation methods for extremal graphs exist for trees \citep{engelke2020a, hu2022}, and for general graphs based on lasso-type $L^1$ penalization \citep{engelke2022a,wan2023graphical}.
The latter paper proposes the \eglearn{} method and shows, under certain conditions, its sparsistency even in the high-dimensional case where the dimension may grow with the sample size.
More precisely,
if $G$ is the graph implied by the zero entries in the true $\Theta$,
and $\widehat G = (V,\widehat E)$ is the estimated graph from \eglearn{} then
\begin{align}\label{sparsistency}
  \mathbb P( \widehat G = G ) \to 1, \quad n \to\infty.
\end{align}
While they consistently recover a general graph structure, they do not obtain an estimate of the $\Gamma$ matrix on the estimated graph structure.

Our theory complements the structure estimation in \citet{engelke2022a}.
In combination, we are now able to estimate jointly any graph structure and the corresponding \HR{} parameter matrix consistently.

\begin{corollary}
    \label{cor:matrixCompletionSparsistent}
    Consider a \HR{} graphical model with graphical structure $G = (V,E)$
    and variogram matrix $\Gamma$. Let $\widehat\Gamma$ be a consistent estimator for $\Gamma$ and $\widehat G = (V,\widehat E)$ a sparsistent estimator of $G$ as in~\cref{sparsistency}.
    If $\widehat \Gamma^{\widehat G} = \comp[\widehat{G}]{\restrictTo{\widehat \Gamma}{\widehat{G}}}$ is the completion on $\widehat G$ and $\widehat \Theta^{\widehat G}$ its precision matrix, then $\widehat \Gamma^{\widehat G}$ is a consistent and sparsistent estimator for $\Gamma$, that is, for all $\varepsilon > 0$,
    \begin{align*}
        \P{
            \max_{(i,j) \in V \times V}
            \abs{\widehat \Gamma^{\widehat G}_{ij} - \Gamma_{ij}}
            <
            \varepsilon
            , \;\;
            \widehat\Theta^{\widehat G}_{ij} = 0
            \;
            \forall (i,j) \notin \Ebar
        } \to 1
        , \quad
        n \to \infty.
    \end{align*}
\end{corollary}

In \cref{cor:matrixCompletionSparsistent}, not only is $\widehat\Gamma^{\widehat G}$ a consistent estimator of $\Gamma$, it is also, with probability going to one, a graphical model with respect to the true graph $G$.

\section{Application}
\label{sec:applic}

Evaluation of the risk of climate extremes such as heatwaves \citep[e.g.,][]{reich2014hierarchical}
or floods \citep[e.g.,][]{asadi2015extremes, coo2018} is the bread and butter of extreme value analysis,
and our methods are perfectly suitable for such data.
In order to illustrate our methods on a widely used data set and to underline their applicability to a broad range of domains, we study the river flow data of \citet{asadi2015extremes} in the Supplementary Material (\cref{sec:danube}).
In this section, we extend the range of possible applications and
use our methodology to study 
large flight delays in the U.S.\ flight network.
Excessive delays have a variety of negative effects,
ranging from inconveniences for passengers
and congestion of critical airport infrastructure
to financial losses for involved parties.

\subsection{Data and exploratory analysis}

The United States Bureau of Transportation Statistics%
\footnote{\url{https://www.bts.dot.gov/browse-statistical-products-and-data/bts-publications/airline-service-quality-performance-234-time}}
provides records of domestic flights in the U.S.\ that are operated by major carriers (at least $1\%$ market share)
at airports accounting for at least $1\%$ of domestic enplanements.
We use this data from \valYearStart{} to \valYearEnd{},
selecting only airports from the contiguous U.S.\ with
a minimum of \valMinYearlyAirportFlights{} flights per year.
For each airport, we compute the accumulated (positive) flight delays (in minutes) on a given day.
As there are a few days for which no data are available,
this results in a data set with $n=\valNDaysFiltered{}$ observations
$x_1,\dots, x_n \in \mathbb R^d$
of daily accumulated flight delays for the $\valNYearsAll{}$ years at $d=\valNAirportsFiltered{}$ airports.
This pre-processed data set is available in the R-package \texttt{graphicalExtremes}
\citep{graphicalExtremes2022}.

The models in this paper are suitable for variables where the largest observations are asymptotically dependent. For our data, such dependence between the largest daily accumulated flight delays at different airports may result from several factors. For instance, large delays at a hub in the network may induce delays at other airports that have frequent connections to that hub.
On the other hand, independently of flight connections, meteorological events such as severe storms or heavy snowfall can cause simultaneous delays at airports in the same geographical area. 

In order to find groups of airports at risk of simultaneous large delays, we first run a $k$-medoids clustering \citep{kaufman2009finding}, where similarities are defined 
in terms of
the strength of extremal dependence. Clustering is frequently used in the extreme value literature to identify regions that are homogeneous in terms of dependence properties \citep[e.g.,][]{bernard2013, saunders2019, vignotto2021}. We follow these approaches but use a %
different dissimilarity measure, namely the empirical version $\hat\chi_{ij}(p)$ of the extremal correlation~\cref{eq:defChi} at probability level~$p=\valUsedPClustering$.
The clusters and subsequent results are stable with respect to the exact probability level
in a range of about $p \in [0.8,0.95]$, yielding both a sufficiently high threshold and enough exceedances.

The resulting clusters are shown in \cref{fig:applicationClusters}, for
an ad-hoc choice of \valNClustersWord{} clusters. Even though no information on the airport locations is used, %
the resulting groups exhibit strong geographical proximity, supporting our intuition on the importance of flight connectivity and meteorological factors.
\cref{fig:applicationChiHist} shows that the extremal dependence of large delays is much stronger within clusters than between different clusters.
In the following we focus on the southern cluster around Texas,
consisting of $\valNAirportsCluster{}$ airports,
which is particularly stable even when modifying the number
$k$ of clusters and threshold value $p$;
similar analyses can be conducted for the other clusters.
The IATA codes of this cluster are given in \cref{table:IATAsChosen}.
The existence of a (regular) flight connection between two airports defines a natural graph, denoted by $G_{\text{flight}}$
and shown in the left panel of \cref{fig:applicationGraphs}.

\begin{figure}
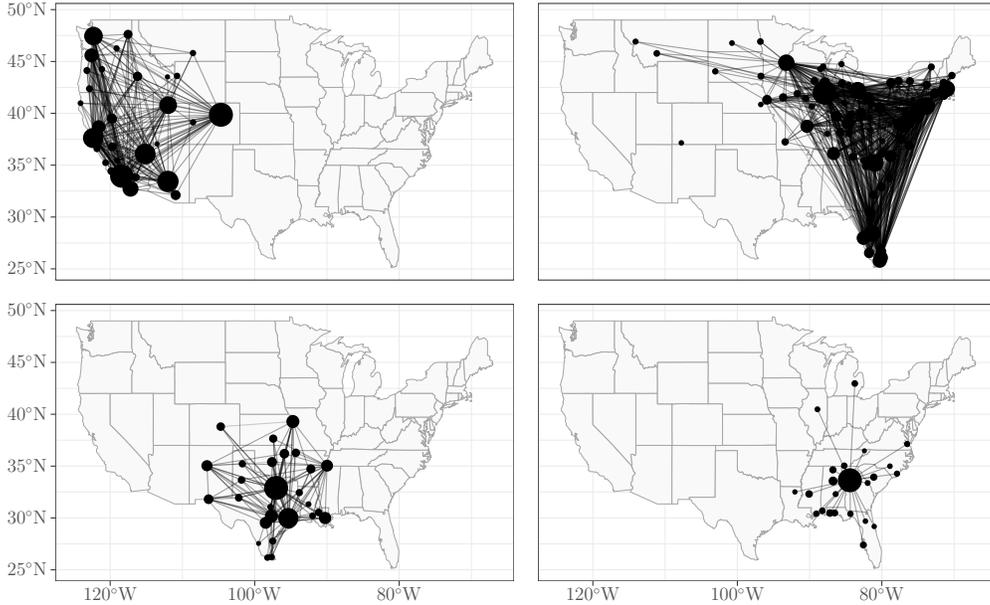

    \centering
    \inputTikz[0.9\textwidth]{application/clustering2}
    \setlength\abovecaptionskip{-0.5\baselineskip}
    \caption{
        Clusters of airports in the contiguous U.S.,
        using $k$-medoids clustering with
        empirical extremal correlation
        as dissimilarity.
        Edges represent airports connected by (regular) flights.
        Airport sizes are proportional to the average number of daily flights.
    }
    \label{fig:applicationClusters}
\end{figure}

\subsection{Graphical modeling}

\label{subsec:applicationGraphicalModeling}

The flight graph $G_{\text{flight}}$ is based on domain knowledge and is certainly a good first candidate for statistical modeling. It is however not necessarily the best graph in terms of conditional independence properties.
We therefore also consider extremal graph structures estimated from data.
To evaluate the fitted models out-of-sample, 
all rows with missing values were removed,
and the data was split into a training set with $\valNDaysEst{}$ observations (\valDateEstStart{} to \valDateEstEnd{}) for estimation, and a validation set with $\valNDaysVal{}$ observations (\valDateValStart{} to \valDateValEnd{}) for selection of tuning parameters and model comparisons.
As a base estimator, we use the empirical extremal variogram $\widehat \Gamma$ computed at probability threshold $p=\valUsedP$ on the training data.

The first, sparse graph is obtained as the minimum spanning tree with weight matrix~$\widehat \Gamma$ as proposed in \cite{engelke2020a}; this tree is denoted by $\widehat T$ and is shown in the center of \cref{fig:applicationGraphs}. Alternatively, the empirical extremal correlation $\widehat \chi$ as weight matrix can also recover an underlying tree~\citep{engelke2020a, hu2022}.

\begin{figure}
    \centering
    \inputTikz{application/graph_all}
    \setlength\abovecaptionskip{-1\baselineskip}
    \caption{
        The flight graph $\Gflight$ (left),
        the estimated tree graph $\widehat{T}$ (center),
        and the graph estimated using \eglearn{},
        $\widehat G_{\rho^*}$, for
        $\rho^* = \valRhoStar$ (right).
    }
    \label{fig:applicationGraphs}
\end{figure}

As a second family of estimated extremal graph structures $\widehat G_\rho$ we apply the \eglearn{} algorithm \citep{engelke2022a} to the flights data set with different regularization parameters $\rho\geq 0$. The latter governs the amount of sparsity in the estimated graph, where larger $\rho$ values correspond to sparser graphs. %
\eglearn{} produces a whole sequence of estimated graphical models, but without estimates of the corresponding parameter matrices.

We apply our methodology to fit a \HR{} multivariate generalized Pareto distribution with different extremal graphical structures. As before, the empirical variogram $\widehat \Gamma$ acts as base estimator.
Applying the completion operator $\comp[G]{}$ to the restriction of $\widehat \Gamma$ to a graph $G$ as in~\cref{eq:completionFull}, we obtain graph structured estimators $\widehat \Gamma^{\Gflight}$, $\widehat \Gamma^{\widehat T}$ and $\widehat \Gamma^{\widehat G_\rho}$ for the flight graph, the extremal tree and the $\rho$-regularized general graph, respectively.
We note that, previously, only parameter estimation on the tree was possible through the tree metric property~\cref{eq:tree_metric}, but not on the more general graphs.
Our method thus allows to estimate the \HR{} parameters of an extremal graphical model based on a graph learned from the data, the only restriction being that the graph needs to be connected.

A first benefit of our approach is that it complements structure learning methods such as \eglearn{} by allowing for a data-driven sparsity tuning.
One approach would be to compare the \HR{} log-likelihood values of the different parameter matrices $\widehat \Gamma^{\widehat G_\rho}$, for instance by AIC or BIC applied to the training data.
Instead, we directly compare the log-likelihood
values on the validation set and plot them against the tuning parameter $\rho$ in \cref{fig:applicationRhoVsAic}.
The likelihood is computed based on the \HRP{}
density, and 
since we are mainly interested in the dependencies between different airports,
the univariate marginals were normalized to the standard exponential scale
using the empirical distribution functions;
\cref{fig:applicationMarginals} shows the univariate shape parameters before normalization.
The structural estimators in this section are all part of the family of \HR{} distributions, allowing the comparison of the log-likelihood values, in the same was as normal distributions with different degrees of sparsity are compared in classical graphical modelling \citep[e.g.][]{friedmanEtAl2007}.
The best fitting model on the independent validation data set comes from \eglearn{}  at $\rho^*=\valRhoStar$, and the corresponding graph with $\valNEdgesEglearn$ edges is shown on the right-hand side of \cref{fig:applicationGraphs}. The flight graph performs significantly worse, and the tree seems to be too sparse.
The \HR{} precision matrix defined in \cref{sec:HrPrecisionMatrix} can be used to track how sparsity is induced in the \eglearn{} algorithm.
\cref{fig:applicationThetaToZero} shows the entries of the precision matrices $\widehat \Theta^{\widehat G_\rho}$ as a function of the tuning parameter $\rho$.
Similar to a usual lasso, we observe that they tend to zero in a possibly non-monotone way.

\begin{figure}
    \centering
    \begin{subfigure}[b]{.5\textwidth}
        \inputTikz{application/eglasso_rho_vs_loglik}
    \end{subfigure}
    \hspace{-0.05\textwidth}
    \begin{subfigure}[b]{.5\textwidth}
        \inputTikz{application/ThetaConvergenceToZero}
    \end{subfigure}
    \setlength\abovecaptionskip{-1\baselineskip}
    \caption{
        Left: Log-likelihood based on the validation data set
        for different regularization parameters $\rho$.
        Horizontal lines indicate the log-likelihoods
        of the complete graph ($\valNEdgesComplete$ edges, dashed line)
        and the flight graph ($\valNEdgesFlightGraph$ edges, dotted line).
        The log-likelihood of the tree graph ($\valNEdgesTree$ edges)
        is $\valTreeGraphLoglik$ and not shown.
        Right: Evolution of $\widehat \Theta^{\widehat G_\rho}$ entries.
        Only entries that vanish at $\rho = \valRhoMax$ are shown.
        Both plots show the number of edges in the corresponding graph 
        on top.
    }
    \label{fig:applicationRhoVsAic}
    \label{fig:applicationThetaToZero}
\end{figure}

The completed variogram matrices also serve to assess the fitted model and evaluate its goodness of fit. \cref{fig:applicationGamma} compares the values of the empirical extremal variogram $\widehat \Gamma$ to the variogram estimates implied by the fitted graphical model $\widehat \Gamma^G$ for graph $G$ being the flight graph, the extremal minimum spanning tree and the optimal \eglearn{} graph, respectively. The extremal correlations are $\chi = 2 - 2\Phi(\sqrt{\Gamma}/2)$ for a \HR{} distribution with variogram matrix $\Gamma$, where $\Phi$ is the standard normal distribution function and all functions are applied componentwise.
The results confirm the likelihood considerations and, in particular, convincingly show that a tree model is not flexible enough to capture all extremal dependencies in this data set.
It is worthwhile to note that this finding is in contrast to the results of the application to river flow data in \cref{sec:danube},
where the performance of the tree model is very similar to the \eglearn{} method,
outperforming both the complete graph and the physical flow graph of the river network. This stresses the usefulness and applicability of our methodology in various situations with different amounts of sparsity.

\begin{figure}
    \centering
    \inputTikz{application/empVsFitted_Gamma_all2}
    \setlength\abovecaptionskip{-1\baselineskip}
    \caption{
        Extremal correlations
        based on the empirical and fitted extremal variogram
        for
        the flight graph $\Gflight$ (left),
        the estimated tree graph $\widehat{T}$ (center),
        and the graph estimated using \eglearn{},
        $\widehat G_{\rho^*}$, for
        $\rho^* = \valRhoStar$ (right).
    }
    \label{fig:applicationGamma}
\end{figure}

\section*{Acknowledgments}

Financial support from the F.R.S-FNRS in the form of a PDR grant (T020321F, J. Segers)
and the Swiss National Science Foundation (S. Engelke)
is gratefully acknowledged.

\bibliography{literature}
\bibliographystyle{apalike}

\newpage

\begin{center}
{\large\bf SUPPLEMENTARY MATERIAL}
\end{center}

\etocdepthtag.toc{mtappendix}
\etocsettagdepth{mtchapter}{none}
\etocsettagdepth{mtappendix}{subsection}
\tableofcontents
\clearpage

\setcounter{section}{0}
\renewcommand{\thesection}{S.\arabic{section}}
\renewcommand{\theHsection}{S.\arabic{section}}
\setcounter{figure}{0}
\renewcommand{\thefigure}{S.\arabic{figure}}

\section{Background}
\subsection{\MP{} inverse}
\label{subsec:moorePenrose}

\begin{definition}
    \label{def:pseudoinverse}
    As shown for example in \citet[Theorem~1]{penrose1955},
    for any matrix $A \in \Rd[n \times m]$ there exists a unique
    matrix $B \in \Rd[m \times n]$ satisfying the equations
    \begin{align*}
        ABA
        &=
        A,
        &
		\lr{AB}\T
		&=
		AB,
        \\
        BAB
        &=
        B,
		&
        \lr{BA}\T
        &=
        BA
        .
    \end{align*}
    This solution is called the Moore--Penrose inverse
    or simply \quotes{pseudo-inverse} of $A$ and is denoted by $A^+ := B$.
\end{definition}
\begin{remark}
    The pseudo-inverse is defined in a similar way for matrices with complex-valued entries,
    using the conjugate transpose in place of the transpose.
\end{remark}

\begin{lemma}
    \label{lemma:propertiesMP}
    The pseudo-inverse has the following properties.
    \begin{enumerate}
        \item $\lr{A\pinv}\pinv = A$
        \item $\lr{AA\T}^+ = \lr{A^+}\T A^+$
            (note that in general $\lr{AB}^+ \neq B^+A^+$).
        \item $AA^+$ is the orthogonal projection onto the image of $A$.
        \item Using the singular value decomposition
            $A = U\Sigma V\T$, the pseudo-inverse can be computed as
            $A\pinv = V \Sigma\pinv U\T$, with
            \begin{align*}
                \Sigma\pinv
                =
                \diag{\Sigma_{11}\inv, \dots, \Sigma_{rr}\inv, 0, ... 0}
                , \quad
                r
                =
                \mathrm{rank}\lr{A}
                .
            \end{align*}
            \label{lemma:propertiesMP:SVD}
    \end{enumerate}
\end{lemma}
\begin{proof}
    The first statement follows from the interchangeability of $A$ and $B$
    in \cref{def:pseudoinverse}.
    The second statement can be proven by substituting the left- and right-hand side in the
    defining equations from \cref{def:pseudoinverse}.
    The last two statements are from \citet[Section~5.5.4]{golub1996}.
\end{proof}

\begin{lemma}
    \label{lemma:conditionsPinv}
    Let $A$ and $B$ be two symmetric matrices of equal size and
    let $P_A$ denote the orthogonal projection matrix onto the image of $A$.
    Then the following two conditions are sufficient and necessary
    for $A=B\pinv$:
    \begin{itemize}
        \item $\Image B \subseteq \Image A$
            (or equivalently $\ker A \subseteq \ker B$);
        \item $AB = P_A$.
    \end{itemize}
\end{lemma}
\begin{proof}
	In the first bullet point, the equivalence of $\Image B \subseteq \Image A$ and $\ker A \subseteq \ker B$ follows from the identity $(\ker M)^\perp = \Image M$ for symmetric matrices $M$.

    For $A = B\pinv$,
    the two conditions follow directly from \cref{lemma:propertiesMP}; note that by the lemma and by symmetry of $A$ and $B$, we have $AB = (AB)\T = B\T A\T = BA$.

    Conversely, using the fact that the projection matrix $P_A$ is symmetric,
    the defining equations from \cref{def:pseudoinverse}
    follow from the two conditions as follows:
    \begin{alignat*}{3}
        AB &= P_A = (P_A)\T = (AB)\T,
        \\
        BA &= B\T A\T = (AB)\T = AB = A\T B\T = (BA)\T,
        \\
        ABA &= P_A A = A,
        \\
        BAB &= (AB) B = P_A B = B,
    \end{alignat*}
    where we used $\Image B \subseteq \Image A$ in the last step.
\end{proof}

\subsection{Pseudo-determinant}
\label{subsec:pseudoDet}

\cref{def:pseudoDet} and \cref{lemma:pseudoDet1}
are from \citet[Section~2]{knill2013}
and only adapted in scope and notation for their use here.
\begin{definition}
    \label{def:pseudoDet}
    Let $A$ be a square matrix with eigenvalues $\set{\lambda_i}$.
    Then its pseudo-determinant,
    denoted by $\pdet{A}$,
    is defined as the product of its non-zero eigenvalues:
    \begin{align*}
        \pdet{A}
        =
        \prod_{\lambda_i \neq 0}
        \lambda_i
        .
    \end{align*}
    If all eigenvalues are zero, $\pdet{A} = 1$.
\end{definition}
\begin{lemma}
    \label{lemma:pseudoDet1}
    Let $A, B \in \Rdd$.
    \begin{enumerate}
        \item If $A$ is similar to $B$, then $\pdet{A} = \pdet{B}$.
        \item If $A$ is invertible then $\pdet{A} = \abs{A}$.
        \item $\pdet{A\T} = \pdet{A}$.
        \item For a normal matrix $A$, it holds that $\pdet{A\pinv} = 1/\pdet{A}$.
        \item $\pdet{A\T B} = \pdet{A B\T} = \pdet{B A\T}$.
        \item $\pdet{A} \neq 0$.
        \item If $A$ is block diagonal, i.e., if
            $A = \mathrm{Diag}\lr{A_1, \dots, A_k}$,
            then
            $\pdet{A} = \prod_i \pdet{A_i}$.
    \end{enumerate}
\end{lemma}

\begin{lemma}
    \label{lemma:pseudoDet2}
    Let $V$ be a linear subspace of $\Rd$ and put
    \begin{align*}
        \Qcal_V = \set{
            A \in \Rdd
            :
            A = A\T
            , \,
            \ker{A} = V
        }
        .
    \end{align*}
    Then for any $A_1, A_2 \in \Qcal_V$,
    \begin{align*}
        \pdet{A_1 A_2}
        =
        \pdet{A_1}
        \cdot
        \pdet{A_2}
        .
    \end{align*}
    Furthermore, for a sequence $A_i \in \Qcal_V$, $i \in \N$, with
    $\limit{i}{A_i} = A_0 \in \Qcal_V$,
    \begin{align*}
        \limit{i}{\pdet{A_i}}
        =
        \pdet{A_0}
        .
    \end{align*}
\end{lemma}
\begin{proof}
    Let $k = \dim V$, $l = d-k$, and let $\set{b_1, \dots, b_d}$ be a basis of $\Rd$ with
    $V = \mathrm{span}\lr{b_1, \dots, b_k}$.
    Let $M$ denote the corresponding basis change matrix and 
    $B_i = M\inv A_i M$ for $i \geq 0$.
    By construction, $B_i$ is of the form
    \begin{align*}
        \lr{
            \begin{matrix}
                \zeroVec_{k \times k} & \zeroVec_{k \times l}
                \\
                \zeroVec_{l \times k} & B_i'
            \end{matrix}
        }
    \end{align*}
    with $B_i' \in \Rdd[l]$ invertible
    and $\zeroVec_{a \times b} \in \Rd[a \times b]$ equal to zero in all entries.
    Using the properties from \cref{lemma:pseudoDet1} it follows that
    \begin{align*}
        \pdet{A_i}
        &=
        \pdet{B_i}
        =
        \pdet{\zeroVec_{k \times k}}
        \cdot
        \pdet{B_i'}
        =
        \abs{B_i'}
        .
    \end{align*}
    Since $M$ is chosen independently of $i$,
    \begin{align*}
        \pdet{A_1} \cdot \pdet{A_2}
        &=
        \abs{B_1'} \cdot \abs{B_2'}
        \\ &=
        \abs{B_1' B_2'}
        \\ &=
        \pdet{B_1 B_2}
        \\
        &=
        \pdet{M B_1 M\inv M B_2 M\inv}
        \\
        &=
        \pdet{A_1 A_2}
        .
    \end{align*}
    Furthermore,
    continuity of matrix multiplication implies that
    $B_0 = \limit{i}{B_i}$
	and hence
    $\limit{i}{B_i'} = B_0'$.
    Using the continuity of the regular determinant,
    it follows that
    \begin{align*}
        \limit{i}{\pdet{A_i}}
        &=
        \limit{i}{\abs{B_i'}}
        =
        \abs{B_0'}
        =
        \pdet{A_0}
        . \qedhere
    \end{align*}
\end{proof}

\subsection{Graph theory}
\label{sec:graphTheory}

All graphs considered in this paper are undirected, simple graphs without loops.
A graph $G=(V,E)$ is defined by a set of vertices
$V = \{\oneToX{d}\}$ and a set of edges
$E \subseteq V \times V$. %
Since we only consider undirected graphs without loops,
the set of all possible edges is $\allEdges := \tlr{V \times V} \setminus \setM{(i,i)}{i \in V}$,
and it must always hold that $(i,j) \in E \Rightarrow (j,i) \in E$.
The set of loops from each node to itself is denoted as
$\diagV = \setM{(i,i)}{i \in V}$.
These loops are no edges in the sense defined above,
but can be useful when identifying edges and subgraphs
with matrix entries and submatrices, respectively.
For a set of edges $E \subseteq \allEdges$,
the inclusion of the loops from each node to itself is denoted as
$\Ebar = E \cup \diagV$.

A graph is called complete if $E = \allEdges$.
A subgraph $G' = (V', E')$ of $G$ is a graph consisting of a subset of vertices
$V' \subseteq V$ and a subset of $E$ such that all endpoints lie in $V'$, i.e.,
$E' \subseteq E \cap \allEdges[V']$.
If $E'$ is maximal (i.e., the latter set inclusion is an equality),
$G'$ is called the subgraph induced by $V'$.
A subset of vertices $C \subseteq V$ is called complete if its induced subgraph is a complete graph.
A subset of vertices is called a clique if it is complete
and not a strict subset of another complete subset.

The neighborhood of a vertex $i$ is defined as
$\delta(i) = \set{j \in V: (i,j) \in E}$.
The neighborhood of a vertex including the vertex itself is denoted as
$\bar\delta(i) := \delta(i) \cup \set{i}$.
A path of length $m$ between vertices $i$ and $j$ is a sequence of $m+1$ distinct vertices
$p_0, p_1, \dots, p_m$ such that $p_0 = i$, $p_m = j$, and $(p_{i-1}, p_{i}) \in E$ for all $i=\oneToX{m}$.
If there exists a path between two vertices, they are said to be connected.
A graph is connected if any two of its vertices are connected.
A cycle of length $m$ is a sequence of $m$ distinct vertices
$p_1, \dots, p_m$ such that $(p_m, p_1) \in E$ and $(p_{i-1}, p_{i}) \in E$ for all $i=2,\ldots,m$.
A chord is an edge between two vertices of a cycle that is not itself part of the cycle.

\begin{definition}[Decomposable Graph]
    A decomposable graph is a graph in which all cycles of four or more vertices have a chord.
\end{definition}

A useful property of decomposable graphs is the following 
running intersection property \citep[see e.g.,][Proposition~2.17]{lauritzen1996}.

\begin{lemma}[Running intersection property]
    \label{def:runningIntersection}
    A graph $G$ is decomposable,
    if and only if
    its set of cliques $\Cli = \set{C_1, \dots, C_N}$
    can be
    ordered such that the running intersection property is fulfilled, that is,
    for all $i = 2,\ldots,N$ there exists $k(i) \in \{1,\ldots,i-1\}$ such that
    \[
    	D_i := C_i \cap (C_1 \cup \ldots \cup C_{i-1}) \subseteq C_{k(i)}.
    \]
    The multiset $\Sep = \set{D_2, \dots, D_N}$ 
    is independent of the chosen ordering of $\Cli$
    and its elements are called separators.
    For connected graphs,
    all separators are non-empty.
\end{lemma}

\begin{definition}[Block graph]
    A block graph is a decomposable graph in which all non-empty separators consist of single vertices:
    \begin{align*}
        \abs{D} \in \set{0,1}
        \quad \forall 
        D \in \Sep.
    \end{align*}
\end{definition}

\begin{definition}[Tree Graph]
    A tree graph or a tree is a connected graph that does not contain any cycle.
\end{definition}
\begin{remark}
    For $d \geq 2$, the set of trees is identical to the set of connected block graphs in which all cliques consist of exactly two vertices.
\end{remark}

\begin{definition}[Graph Laplacian]
    \label{def:graphLaplacian}
    For an undirected graph $G=(V,E)$
    the graph Laplacian matrix $L \in \Rdd$ is defined by
    \begin{alignat*}{3}
        L_{ij}
        =
        \myCases{
            \myCase{\deg(i)}{i=j,}
            \\
            \myCase{-1}{(i,j) \in E,}
            \\
            \myCase{0}{\text{otherwise,}}
        }
    \end{alignat*}
    where the degree $\deg(i)$ of a vertex $i$
    is defined as $\abs{\delta(i)}$.
\end{definition}

\section{Details on matrix completion}
\label{sec:combiningCliques}

\subsection{Matrix completion as likelihood optimization}

\label{subsec:completionAsOpt}

For a \HR{} distribution $Y$ with parameter matrix $\Gamma$, the random vectors $Y^{(k)}$, $k \in V$, defined in \cref{subsec:extremalCondInd} satisfy
\citep{engelke2014}
\begin{align}\label{datak}
  \tlr{Y^{(k)}_i - Y^{(k)}_k}_{i\neq k}   \sim
  \normal{-\half\operatorname{diag}(\Sigma^{(k)}), \Sigma^{(k)}}
  .
\end{align}
For an independent sample of size $n$ from $Y$, one can obtain samples from this $(d-1)$-dimensional normal distribution  by selecting only data with $Y_k > 1$ and following~\cref{datak}. We denote the corresponding empirical covariance matrix by $\widehat \Sigma\k$, augmented by a $k$th row and column of zeros to make a $d\times d$ matrix.
Ignoring the information on the parameter matrix $\Gamma$ in the mean vector, the (surrogate) log-likelihood of this model can be written in terms of our \HR{} precision matrix $\Theta$ as
\begin{align}\label{llhk}
  L(\Theta; \widehat{\Sigma}^{(k)})
  \propto
  \log \pdet{\Theta} - \trace{\widehat{\Sigma}^{(k)}\Theta},
\end{align}
where $\pdet{\,\cdot\,}$ is the pseudo-determinant.
Setting $\widehat \Gamma\k = \gamma(\widehat \Sigma\k)$ gives a nonparametric estimator of~$\Gamma$. By \cref{prop:ThetaReps} it holds that $\widehat{\Sigma}\k = \ID\tlr{-\half\widehat{\Gamma}\k}\ID$ and therefore the cyclic permutation property of the trace operator together with the fact that $\Theta \ID = \ID \Theta = \Theta$ for any $\Theta \in \Pcal_d^{\oneVec}$ shows that the right-hand side of~\cref{llhk} is equal to 
\begin{align*}
  \log \pdet{\Theta} + \half \trace{\overline{\Gamma} \Theta},
  \end{align*}
with $\overline{\Gamma} = \widehat \Gamma\k$. In order to use all data in the sample, we can consider this likelihood with $\overline{\Gamma}$ equal to a combined version of the variogram estimators over all $k$, defined as $\widehat \Gamma := d^{-1} \sum_{k=1}^d \widehat\Gamma\k$ and called the empirical variogram \citep{engelke2020a}.

A natural way to obtain a graphical model is therefore to maximize this surrogate \HR{} likelihood over $\Theta \in \Pcal_d^{\oneVec}$ under the constraint of a graph-structured precision matrix.
As shown in \cref{prop:completionOptProblem},
solving this optimization problem corresponds to our completion operator $\comp[G]{}$ from \cref{def:notationComp}.

\subsection{Estimation strategy for sparse graphs}
\label{subsec:estimationOnSparseGraphs}

Let $ \Cli$ be the collection of all cliques $C \subseteq V$ of a connected graph $G = (V, E)$ and suppose that for every $C\in \Cli$, $\widehat \Gamma_C$ is a consistent estimator of the entries $\Gamma_{ij}$ for $i,j \in C$. Note that $\widehat \Gamma_C$ can be computed using information only from the components in $C$, which reduces the computational cost drastically if the cliques are small, especially if censoring is used.
Since the cliques are overlapping on the separator sets, we need to combine different estimators to obtain a partial variogram matrix $\widehat\GammaO \in \ROdd$ on the whole graph. We do this by averaging the estimators on the intersections. Let $\Cli_{(i,j)} = \setM{C \in \Cli}{i,j \in C}$ denote the set of all cliques containing the edge $(i,j) \in E$ and put
\begin{align*}
    \widehat\GammaO_{ij}
    =
    \myCases{
        \myCase{
            \undefined
        }{
            (i,j) \notin \Ebar,
        }
        \\
        \myCase{
            \frac{1}{\abs{\Cli_{(i,j)}}}
            \sum_{C \in \Cli_{(i,j)}}
            \lr{
                \widehat{\Gamma}_C
            }_{ij}
        }{
            (i,j) \in \Ebar.
        }
    }
\end{align*}
This partial matrix is clearly a consistent estimator for
$\restrictTo{\Gamma}{G}$ as in~\cref{partial_consist},
since it is an average of consistent estimators.
For fixed sample size $n$, the estimator $\widehat\GammaO$ is not guaranteed to be a valid (partial) variogram matrix; this is the price to pay for the more efficient estimation using only information within each clique,
and by \cref{thm:matrixCompletionConsistent}, the probability of  $\widehat\GammaO$ being invalid converges to zero for $n \rightarrow \infty$.
If it exists, the completion
$\widehat{\Gamma}^G = \comp[G]{\widehat\GammaO}$
can be computed using \cref{prop:completeDecomposable}
or \cref{cor:completeGeneral},
and by \cref{thm:matrixCompletionConsistent}
it is a consistent estimator of $\Gamma$ with correct graph structure $G$.

A natural question is whether it is possible
to replace the arithmetic mean by another function
that guarantees a valid partial variogram.
The following example shows that this is not possible,
as long as entries that belong to only a single clique are not altered, as well.

\begin{example}
    Let
    \begin{align*}
        \GammaO
        =
        \begin{pmatrix}
            0 & 1 & 16 & \undefinedNoQuotes \\
            1 & 0 & x & 1 \\
            16 & x & 0 & 1 \\
            \undefinedNoQuotes & 1 & 1 & 0
        \end{pmatrix}
        .
    \end{align*}
    Then there are $x_1, x_2 \in \Rd[]$ such that with $x=x_1$
    the principal submatrix $\restrictTo{\GammaO}{\set{1,2,3}}$ is \CND{},
    and with $x=x_2$
    the principal submatrix $\restrictTo{\GammaO}{\set{2,3,4}}$ is \CND{},
    but there exists no $x$ such that both submatrices are \CND{} at the same time.
\end{example}
\begin{proof}
    \cite{gower1982} shows that $\Gamma$ can be interpreted as 
    a Euclidean distance matrix with $\Gamma_{ij} = d_{ij}^2$,
    where $d_{ij}$, for $i,j \in \set{\oneToX{d}}$, is the pairwise distance between 
    the points generating $\Gamma$ (see also \cref{subsubsec:proofNonCompletableGraphs}).
    The triangle inequality then requires
    $9 \leq x_1 \leq 25$ and $0 \leq x_2 \leq 4$,
    which can be satisfied for each submatrix but not simultaneously for a single value of $x$.
\end{proof}

\subsection{Simulation study}
\label{sec:simulation}

We perform a simulation study to
illustrate \cref{thm:matrixCompletionConsistent} and
the estimation strategies from \cref{subsec:estimationOnSparseGraphs}.
Throughout, the underlying graph $G$ is assumed to be known.

\subsubsection{Setup}
We randomly generate two decomposable graphs of sizes $d=6$ and $d=10$,
shown in \cref{fig:simulationGraph},
and two \HRP{} graphical models, parametrized by the following matrices
(rounded to two decimals):
\begin{gather*}
    \Theta_{d=6} = \lr{
        \begin{matrix}
     \phantom{-}10.06 & -4.09 & -1.65 & -5.51 & \phantom{-}2.37 & -1.17 \\
     -4.09 & \phantom{-}4.09 & 0 & 0 & 0 & 0 \\
     -1.65 & 0 & \phantom{-}11.78 & -9.10 & -1.69 & \phantom{-}0.65 \\
     -5.51 & 0 & -9.10 & \phantom{-}32.21 & -4.94 & -12.66 \\
     \phantom{-}2.37 & 0 & -1.69 & -4.94 & \phantom{-}4.27 & 0 \\
     -1.17 & 0 & \phantom{-}0.65 & -12.66 & 0 & \phantom{-}13.19
\end{matrix}
    }
    , \\
    \Theta_{d=10} = \lr{
        \tiny
        \begin{matrix}
     \phantom{-}24.15 & -3.95 & \phantom{-}2.45 & -10.61 & -1.42 & -4.85 & -4.51 & -1.25 & 0 & 0 \\
     -3.95 & \phantom{-}77.85 & 0 & 0 & 0 & 0 & 0 & 0 & \phantom{-}7.13 & -81.03 \\
     \phantom{-}2.45 & 0 & \phantom{-}22.58 & -15.78 & -8.16 & -1.09 & 0 & 0 & 0 & 0 \\
     -10.61 & 0 & -15.78 & \phantom{-}27.10 & -2.18 & \phantom{-}1.69 & \phantom{-}4.32 & -4.53 & 0 & 0 \\
     -1.42 & 0 & -8.16 & -2.18 & \phantom{-}11.76 & 0 & 0 & 0 & 0 & 0 \\
     -4.85 & 0 & -1.09 & \phantom{-}1.69 & 0 & \phantom{-}4.25 & 0 & 0 & 0 & 0 \\
     -4.51 & 0 & 0 & \phantom{-}4.32 & 0 & 0 & \phantom{-}3.97 & -3.77 & 0 & 0 \\
     -1.25 & 0 & 0 & -4.53 & 0 & 0 & -3.77 & \phantom{-}9.55 & 0 & 0 \\
     0 & \phantom{-}7.13 & 0 & 0 & 0 & 0 & 0 & 0 & \phantom{-}40.63 & -47.76 \\
     0 & -81.03 & 0 & 0 & 0 & 0 & 0 & 0 & -47.76 & \phantom{-}128.79
\end{matrix}
    }
    .
\end{gather*}

\begin{figure}
    \centering
    \inputTikz[0.4\textwidth]{simulation/graph_d6}
    \inputTikz[0.4\textwidth]{simulation/graph_d10}
    \caption{
        The graphs $G$ of each of the \HRP{} distributions that are sampled from in the simulation study (\cref{sec:simulation}).
    }
    \label{fig:simulationGraph}
\end{figure}

From each of these \HRP{} distributions we generate $n$ samples,
and then estimate the parameter matrix $\Gamma$,
using the following methods.
First, we use the empirical variogram \citep{engelke2020a}
to directly estimate all entries of $\Gamma$ (``Full variogram'').
This method can be modified as described in \cref{subsec:estimationOnSparseGraphs}
by applying the empirical variogram to each clique of $G$ and combining
the results through the matrix completions from \cref{sec:matcomprob} (``Clique-wise variogram'').

Second, we use maximum likelihood estimation
with respect to the probability density function
$f(y) = \lambda(y) / \LambdaVec(\zeroVec)$,
where $\lambda(y) = \lambda(y; \Theta)$ from \cref{prop:HrDensitySimple}.
Here, we consider three different approaches:
estimating all entries of $\Gamma$ directly (``Full MLE'');
estimating only the entries in $\Theta$ corresponding to edges of $G$,
while setting all others to zero (``Graphical MLE'');
and estimating the entries of $\restrictTo{\Gamma}{C}$ for each clique $C$ separately,
with respect to the ``marginals''
$f\I(y) = \lambda_I(y) / \LambdaVec_I(\zeroVec)$,
combining them as described in \cref{subsec:estimationOnSparseGraphs} (``Clique-wise MLE'').

\subsubsection{Results}
We compare the methods in terms of their computational cost,
and the quality of their estimate $\widehat\Gamma$.
The computational cost is expressed by the time taken to compute each estimate,
using our implementation in R.
The quality of the estimate
is measured by the mean squared error
\begin{align*}
    \mathrm{MSE}
    =
    \frac{2}{d(d-1)}
    \sum_{i<j}
    \lr{\Gamma_{ij} - \widehat\Gamma_{ij}}^2
    .
\end{align*}
This expression considers only the entries in the upper triangular part,
since any valid variogram matrix is symmetric with zero diagonal.

The results are illustrated in \cref{fig:simulationTimeVsErr},
\cref{table:simulation20} for $n=20$ and $d=6$, and \cref{table:simulation200} for $n=200$ and $d=10$;
other combinations of $n$ and $d$ exhibit the same general trends.
Throughout, the two methods based on the empirical variogram are very fast,
requiring less than a second.
The next fastest method is the clique-wise MLE,
followed by the graphical MLE,
with the full MLE being the slowest.
Notably, for $d=6$ the clique-wise MLE is faster than the full MLE by a factor of $10$,
whereas for $d=10$ it is faster by a factor of $200$.

In terms of MSE,
the two variogram based methods perform similarly,
with the clique-wise variogram slightly outperforming the full variogram.
In line with the significantly higher computational cost,
the MLE based methods perform better than the variogram based methods,
and quite similar to each other.

\begin{figure}
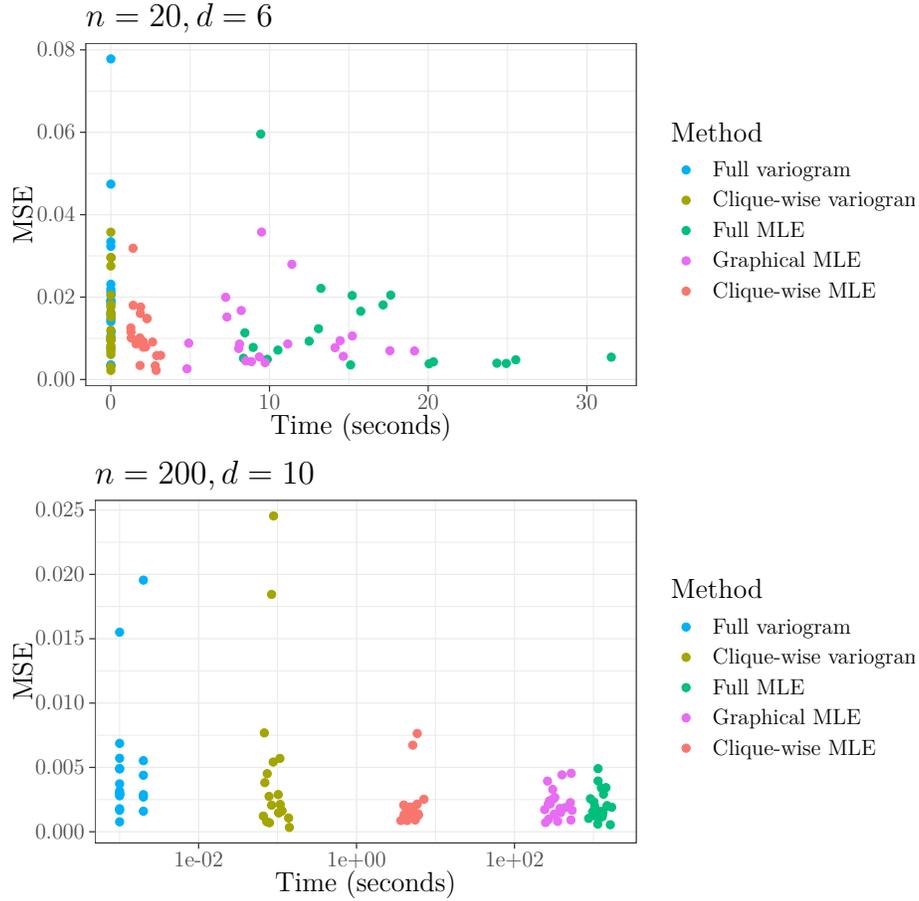

    \centering
    \inputTikz[0.8\textwidth]{simulation/timeVsErrScatter_n20_d6}

    \inputTikz[0.8\textwidth]{simulation/timeVsErrScatter_n200_d10}

    \caption{
        Time and MSE of different estimation methods
        and combinations of dimension ($d$) and sample size ($n$).
    }
    \label{fig:simulationTimeVsErr}
\end{figure}

\begin{table}
    \centering
    \begin{tabular}{lrr}
  \hline
Method & Time & MSE \\ 
  \hline
Full variogram & 2.00E-03 & 2.13E-02 \\ 
  Clique-wise variogram & 2.00E-03 & 1.50E-02 \\ 
  Clique-wise MLE & 2.04E+00 & 1.10E-02 \\ 
  Graphical MLE & 1.06E+01 & 1.09E-02 \\ 
  Full MLE & 1.61E+01 & 1.22E-02 \\ 
   \hline
\end{tabular}

    \caption{
        Average computation time and MSE for $n=20$, $d=6$.
    }
    \label{table:simulation20}
\end{table}

\begin{table}
    \centering
    \begin{tabular}{lrr}
  \hline
Method & Time & MSE \\ 
  \hline
Full variogram & 1.30E-03 & 4.87E-03 \\ 
  Clique-wise variogram & 9.40E-02 & 4.51E-03 \\ 
  Clique-wise MLE & 4.97E+00 & 2.07E-03 \\ 
  Graphical MLE & 3.57E+02 & 2.19E-03 \\ 
  Full MLE & 1.24E+03 & 2.09E-03 \\ 
   \hline
\end{tabular}

    \caption{
        Average computation time and MSE for $n=200$, $d=10$.
    }
    \label{table:simulation200}
\end{table}

\subsubsection{Discussion}

The simulation clearly illustrates advantages of sparse models
in terms of computational performance.
The estimation accuracy is highly comparable for all three MLE based methods,
while the graphical methods are significantly faster.
In fact, for larger dimensions $d$,
full MLE quickly becomes infeasible,
as the number of parameters that need to be optimized simultaneously is $d(d-1)/2$.
For the graphical MLE,
the number of parameters is identical to the number of edges $\abs{E}$,
which can be significantly smaller,
bounded below by $d-1$ for connected graphs.

In the clique-wise MLE method,
a separate optimization problem with
$\abs{C}(\abs{C}-1)/2$
parameters is solved for each clique $C \in \Cli$.
The resulting total number of parameters is
\begin{align*}
    \sum_{C \in \Cli(G)}
    \abs{C}(\abs{C}-1)/2
    ,
\end{align*}
which is generally larger than the number of edges $\abs{E}$.
However,
since the optimization problem for each clique is solved separately,
this can be faster to compute.
Furthermore, on suitable hardware significant speedups can be achieved by
solving the problems in parallel
(to allow a better comparison, we did not perform any parallel computations here).

In this simulation study we considered the underlying graph $G$
as known during the estimation of $\Gamma$.
In practice, this graph might need to be estimated in a separate step,
possibly leading to worse performance of the graphical methods.
Due to the modeling choices and hyperparameter-tuning involved,
we did not include structure estimation in the context of this simulation study
but refer to \cref{sec:applic} and \cref{sec:danube}.

Furthermore, we only considered decomposable graphs $G$,
for which the matrix completion step can be performed very quickly using
\cref{prop:completeDecomposable}.
A graph estimated from data is likely not decomposable,
implying that the iterative method in \cref{cor:completeGeneral} has to be invoked,
which is slower and has to be truncated at some point.

\section{Additional figures: Application to flight data}

\begin{figure}[H]
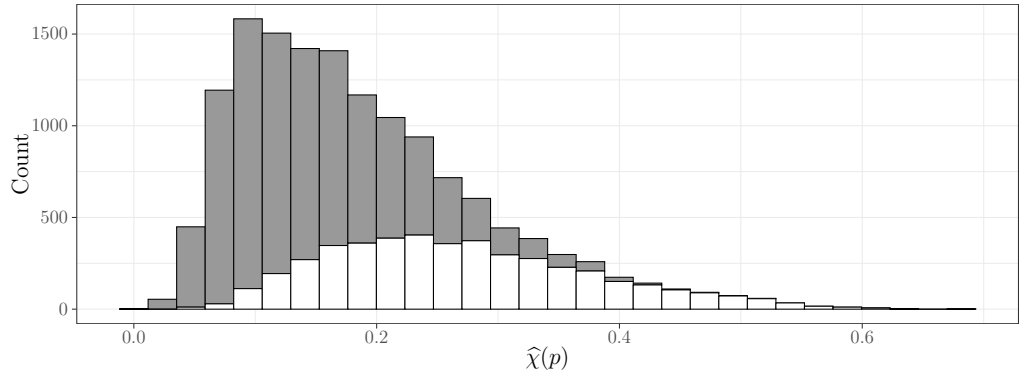

    \centering
    \inputTikz[0.9\textwidth]{application/chiEmp_hist}
    \caption{
        Histogram of empirical $\widehat\chi(p)$ for $p=\valUsedPClustering$.
        Based on all pairs of airports.
        Values corresponding to two airports within the same
        cluster are shown in white,
        and values corresponding to two airports in two
        distinct clusters in gray.
    }
    \label{fig:applicationChiHist}
\end{figure}

\begin{figure}[H]
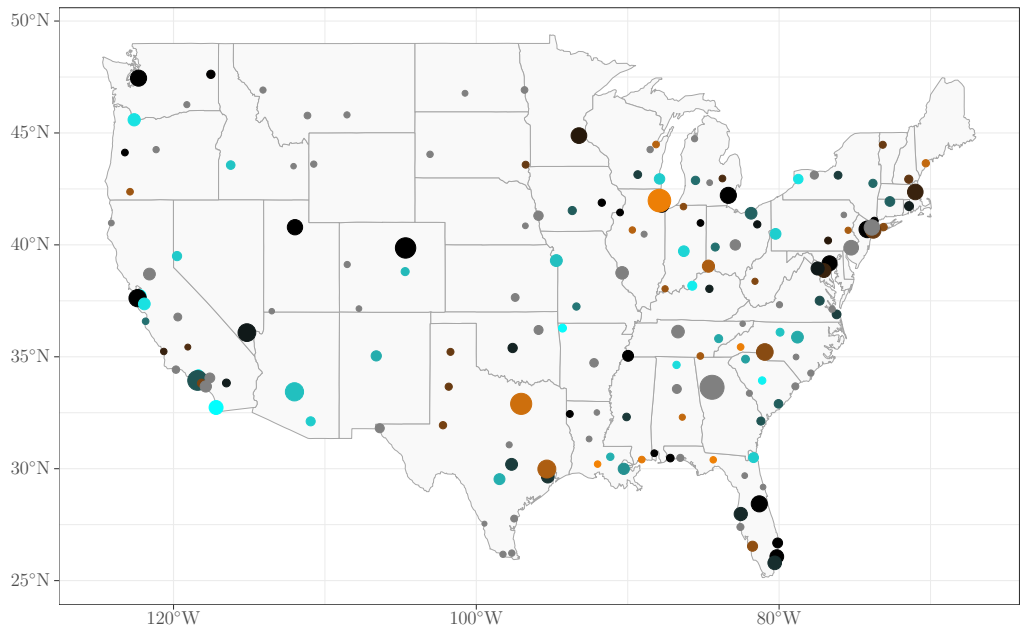

    \centering
    \inputTikz[0.9\textwidth]{application/marginals}
    \caption{
        Airports in the data set.
        Color indicates the shape parameter of
        univariate generalized Pareto distributions,
        fitted for each airport using maximum likelihood estimation
        to the observations above the $p=\valUsedPMarginals$ quantile
        (black at zero, shade of cyan for negative, and orange for positive values).
        The size of each circle is proportional to the average number of daily flights at that airport.
    }
    \label{fig:applicationMarginals}
\end{figure}

\begin{table}
    \begin{tabular}{lllllllllllll}
    \texttt{ABE} & \texttt{ABQ} & \texttt{ACV} & \texttt{AEX} & \texttt{AGS} & \texttt{ALB} & \texttt{AMA} & \texttt{ATL} & \texttt{ATW} & \texttt{AUS} & \texttt{AVL} & \texttt{AVP} & \texttt{BDL} \\
    \texttt{BFL} & \texttt{BHM} & \texttt{BIL} & \texttt{BIS} & \texttt{BMI} & \texttt{BNA} & \texttt{BOI} & \texttt{BOS} & \texttt{BTR} & \texttt{BTV} & \texttt{BUF} & \texttt{BUR} & \texttt{BWI} \\
    \texttt{BZN} & \texttt{CAE} & \texttt{CAK} & \texttt{CHA} & \texttt{CHS} & \texttt{CID} & \texttt{CLE} & \texttt{CLT} & \texttt{CMH} & \texttt{COS} & \texttt{CRP} & \texttt{CRW} & \texttt{CVG} \\
    \texttt{DAB} & \texttt{DAL} & \texttt{DAY} & \texttt{DCA} & \texttt{DEN} & \texttt{DFW} & \texttt{DRO} & \texttt{DSM} & \texttt{DTW} & \texttt{ELP} & \texttt{EUG} & \texttt{EVV} & \texttt{EWR} \\
    \texttt{FAR} & \texttt{FAT} & \texttt{FAY} & \texttt{FLL} & \texttt{FNT} & \texttt{FSD} & \texttt{FWA} & \texttt{GEG} & \texttt{GJT} & \texttt{GNV} & \texttt{GPT} & \texttt{GRB} & \texttt{GRK} \\
    \texttt{GRR} & \texttt{GSO} & \texttt{GSP} & \texttt{HOU} & \texttt{HPN} & \texttt{HRL} & \texttt{HSV} & \texttt{IAD} & \texttt{IAH} & \texttt{ICT} & \texttt{IDA} & \texttt{ILM} & \texttt{IND} \\
    \texttt{ISP} & \texttt{JAC} & \texttt{JAN} & \texttt{JAX} & \texttt{JFK} & \texttt{LAN} & \texttt{LAS} & \texttt{LAX} & \texttt{LBB} & \texttt{LEX} & \texttt{LFT} & \texttt{LGA} & \texttt{LGB} \\
    \texttt{LIT} & \texttt{LNK} & \texttt{LRD} & \texttt{MAF} & \texttt{MCI} & \texttt{MCO} & \texttt{MDT} & \texttt{MDW} & \texttt{MEM} & \texttt{MFE} & \texttt{MFR} & \texttt{MGM} & \texttt{MHT} \\
    \texttt{MIA} & \texttt{MKE} & \texttt{MLI} & \texttt{MLU} & \texttt{MOB} & \texttt{MRY} & \texttt{MSN} & \texttt{MSO} & \texttt{MSP} & \texttt{MSY} & \texttt{MYR} & \texttt{OAK} & \texttt{OKC} \\
    \texttt{OMA} & \texttt{ONT} & \texttt{ORD} & \texttt{ORF} & \texttt{PBI} & \texttt{PDX} & \texttt{PHF} & \texttt{PHL} & \texttt{PHX} & \texttt{PIA} & \texttt{PIT} & \texttt{PNS} & \texttt{PSC} \\
    \texttt{PSP} & \texttt{PVD} & \texttt{PWM} & \texttt{RAP} & \texttt{RDM} & \texttt{RDU} & \texttt{RIC} & \texttt{RNO} & \texttt{ROA} & \texttt{ROC} & \texttt{RSW} & \texttt{SAN} & \texttt{SAT} \\
    \texttt{SAV} & \texttt{SBA} & \texttt{SBN} & \texttt{SBP} & \texttt{SDF} & \texttt{SEA} & \texttt{SFO} & \texttt{SGF} & \texttt{SGU} & \texttt{SHV} & \texttt{SJC} & \texttt{SLC} & \texttt{SMF} \\
    \texttt{SNA} & \texttt{SRQ} & \texttt{STL} & \texttt{SYR} & \texttt{TLH} & \texttt{TPA} & \texttt{TRI} & \texttt{TUL} & \texttt{TUS} & \texttt{TVC} & \texttt{TYS} & \texttt{VPS} & \texttt{XNA}
\end{tabular}

    \caption{IATA codes of the airports considered for clustering}
    \label{table:IATAsAll}
\end{table}

\begin{table}
    \begin{tabular}{lllllllllllll}
    \texttt{ABQ} & \texttt{AEX} & \texttt{AMA} & \texttt{AUS} & \texttt{BTR} & \texttt{COS} & \texttt{CRP} & \texttt{DAL} & \texttt{DFW} & \texttt{ELP} & \texttt{GRK} & \texttt{HOU} & \texttt{HRL} \\
    \texttt{IAH} & \texttt{ICT} & \texttt{LBB} & \texttt{LFT} & \texttt{LIT} & \texttt{LRD} & \texttt{MAF} & \texttt{MCI} & \texttt{MEM} & \texttt{MFE} & \texttt{MSY} & \texttt{OKC} & \texttt{SAT} \\
    \texttt{SHV} & \texttt{TUL} & \texttt{XNA}
\end{tabular}

    \caption{IATA codes of the airports in the analyzed cluster}
    \label{table:IATAsChosen}
\end{table}

\clearpage
\section{Application to Danube data}
\label{sec:danube}

We perform a statistical analysis of the Danube river flow data from \cite{asadi2015extremes},
which has also been investigated for example in
\cite{engelke2020}, \cite{roettger2023}, \cite{gneccoEtAl2021}, \cite{MhallaEtAl2020} and \cite{hu2022}.
We follow the same procedure as in \cref{sec:applic},
but do not perform the initial clustering step
since there is significant extremal correlation
between all pairs of measuring stations over a large range of thresholds;
see \cref{fig:danubePairwiseChi}.

\begin{figure}
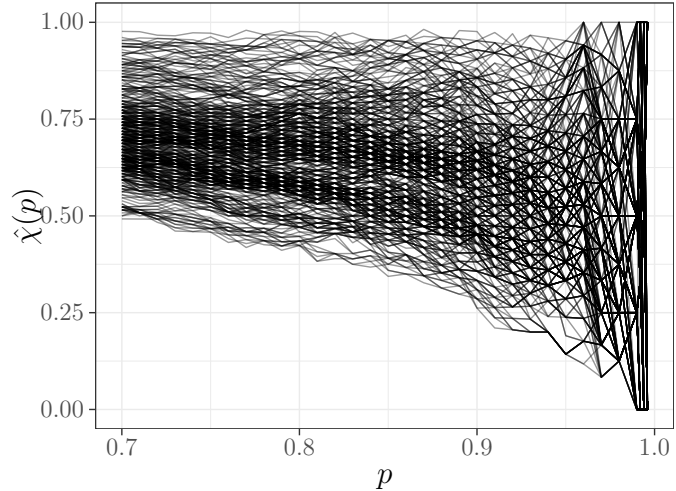

    \centering
    \inputTikz[0.6\textwidth]{danube/pairwiseChi}
    \caption{
        Empirical extremal correlation between pairs of measuring stations as function of the threshold $p$.
    }
    \label{fig:danubePairwiseChi}
\end{figure}

The data consist of clustered measurements from \danYearAllStart{} to \danYearAllEnd{},
representing flow volumes measured at \danNStations{} different stations
situated on tributaries of the Danube river,
yielding a total of \danNObsAll{} observations.
For details on the preprocessing, we refer to \citet{asadi2015extremes}.

\begin{figure}
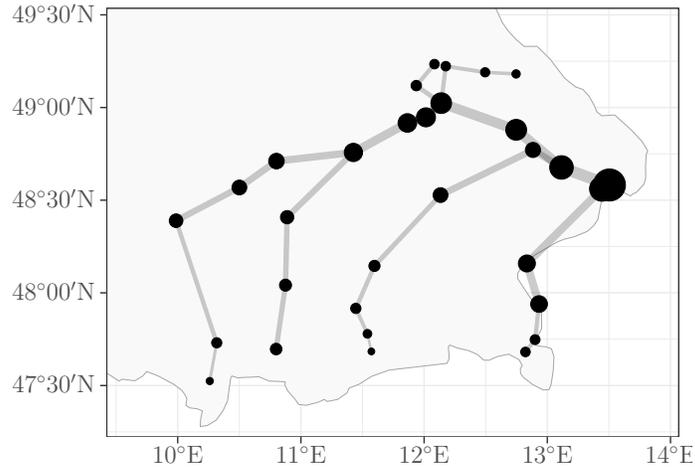

    \centering
    \inputTikz[0.6\textwidth]{danube/flowGraph}
    \caption{
        The physical flow graph.
        Circles correspond to measuring stations,
        edges correspond to stations connected by a river,
        width of circles and edges corresponds to the average flow volume.
        Flow direction is implied by increasing flow volumes.
        The outline is the country border of Germany.
    }
    \label{fig:danubeFlowGraph}
\end{figure}

The flow graph $\Gflow$, shown in the left panel of \cref{fig:danubeGraphs}, is based on domain knowledge and constitutes a good first candidate for the statistical graph structure of the data.
Furthermore, we consider extremal graph structures estimated in a data-driven way,
using the minimum spanning tree method suggested in \cite{engelke2020a}
and the \eglearn{} method from \cite{engelke2022a}.

In order to evaluate the fitted models out-of-sample, 
the data set was split into a training set with $\danNObsEst{}$ observations (\danYearEstStart{} to \danYearEstEnd{}) used for estimation,
and a validation set with $\danNObsVal{}$ observations (\danYearValStart{} to \danYearValEnd{}) used for tuning parameter selection and model comparisons.
The base estimator is the empirical extremal variogram $\widehat \Gamma$ computed at probability threshold $p=\danUsedP$ on the training data.
This choice of threshold value
follows \cite{engelke2020} and \cite{roettger2023}
and
is slightly lower than in \cref{sec:applic},
since the data set at hand is smaller and already consists of clustered maxima.
The \eglearn{} tuning parameter $\rho^*$ is chosen on the basis of the log-likelihood of the estimated models on the test set; see \cref{sec:applic} for further details on the applied methods.

\cref{fig:danubeGraphs} shows the estimated graph structures.
The minimum spanning tree $\widehat T$ is rather similar to the flow graph,
whereas \eglearn{} estimates a significantly denser graph structure $\widehat G_{\rho^*}$
with $\danNEdgesEglearn$ edges.
\cref{fig:danubeRhoVsAic} shows the log-likelihoods and edge counts
of the models estimated by \eglearn{} for different values of the tuning parameter $\rho$,
as well as the log-likelihoods of the models on $\widehat T$, $\widehat G_{\rho^*}$,
and the complete graph.
Notably, the log-likelihood corresponding to $\widehat T$ is quite close to that of \eglearn{},
outperforming both the flow graph and the complete graph.
This is in stark contrast to the performance of the minimum spanning tree in \cref{sec:applic},
indicating that this data set can be better described by a sparse model.
However, the denser graph estimated by \eglearn{} still performs best in terms of log-likelihood.
This trend is further supported by the comparison of
empirical and fitted extremal variogram entries,
shown in \cref{fig:danubeGamma}.

\begin{figure}
    \centering
    \inputTikz{danube/graph_all}
    \setlength\abovecaptionskip{-1\baselineskip}
    \caption{
        The flow graph $\Gflow$ (left),
        the estimated tree graph $\widehat{T}$ (center),
        and the graph estimated using \eglearn{},
        $\widehat G_{\rho^*}$, for
        $\rho^* = \danRhoStar$ (right).
    }
    \label{fig:danubeGraphs}
\end{figure}

\begin{figure}
    \centering
    \begin{subfigure}[b]{.5\textwidth}
        \inputTikz{danube/eglasso_rho_vs_loglik}
    \end{subfigure}
    \hspace{-0.05\textwidth}
    \begin{subfigure}[b]{.5\textwidth}
        \inputTikz{danube/ThetaConvergenceToZero}
    \end{subfigure}
    \setlength\abovecaptionskip{-1\baselineskip}
    \caption{
        Left: Log-likelihood based on the validation data set
        for different regularization parameters~$\rho$.
        Horizontal lines indicate the log-likelihoods of
        the flow graph ($\danFlowGraphLoglik$; $\danNEdgesFlowGraph$ edges, dotted line),
        and the tree graph ($\danTreeGraphLoglik$; $\danNEdgesTree$ edges, dash-dotted line).
        The log-likelihood of the complete graph ($\danNEdgesComplete$ edges)
        is~$\danCompleteGraphLoglik$ and not shown.
        Right: Evolution of $\widehat \Theta^{\widehat G_\rho}$ entries.
        Only entries that vanish at $\rho = \danRhoMax$ are shown.
        Both plots show the number of edges in the corresponding graph 
        on top.
    }
    \label{fig:danubeRhoVsAic}
    \label{fig:danubeThetaToZero}
\end{figure}

\begin{figure}
    \centering
    \inputTikz{danube/empVsFitted_Gamma_all}
    \setlength\abovecaptionskip{-1\baselineskip}
    \caption{
        Extremal correlations
        based on the empirical and fitted extremal variogram
        for
        the flow graph $\Gflow$ (left),
        the estimated tree graph $\widehat{T}$ (center),
        and the graph estimated using \eglearn{},
        $\widehat G_{\rho^*}$, for
        $\rho^* = \danRhoStar$ (right).
    }
    \label{fig:danubeGamma}
\end{figure}

\clearpage
\section{Proofs}
\label{sec:Proofs}

\subsection{Proof of \cref{prop:ThetaReps}}
\label{subsec:proofThetaReps}

Before proving the \namecref{prop:ThetaReps} we establish some auxiliary results.
The actual proof of \cref{prop:ThetaReps} is
a rather short combination of these results and is
given at the end of the \namecref{subsec:proofThetaReps}.
The auxiliary results in this \namecref{subsec:proofThetaReps}
are given in a sequential order,
in the sense that each result uses only previously established results.

In the following, for a given $\Gamma \in \Dcal_d$ we write
\begin{align*}
    \Sigma = \ID \lr{-\half \Gamma} \ID
    .
\end{align*}
We assume that the matrix $S \in \Rdd$ satisfies $\ID S \ID = \Sigma$,
the matrix $S$ not necessarily being symmetric.
\cref{lemma:farrisSeqG} shows that this is in fact a slightly weaker assumption than $\farris{S} = \Gamma$,
which is used in the \namecref{prop:ThetaReps}.
Furthermore, for any $t \in \Rd[]$ we introduce the abbreviation
\begin{align}
    \label{eq:definitionSigmat}
    \Sigma\t
    =
    t \oneVecTT
    +
    S
    .
\end{align}
The choices $S=-\half\Gamma$ or $S=\Sigma$ simplify many of the proofs significantly.

\begin{lemma}
    \label{lemma:farrisSeqG}
    Let $\Gamma \in \Dcal_d$ and recall the map $\farris{}$ in \cref{eq:FarrisTransform}.
    Then
    \begin{align*}
        \farris{S} = \Gamma
        \quad \Longleftrightarrow \quad
        S = S\T
        \, \land \,
        \ID S \ID = \ID \lr{-\half\Gamma} \ID
        .
    \end{align*}
\end{lemma}
\begin{proof}
    Recall that
	\begin{align*}
	\farris{S}
	=
	\oneVec \diag{S}\T
	+
	\diag{S} \oneVecT
	-
	2S.
	\end{align*}

    Since this map preserves (a)symmetry,
    and all $\Gamma \in \Dcal_d$ are symmetric,
    it follows that $S=S\T$ is a necessary condition for $\gamma\lr{S} = \Gamma$.
	Furthermore,
    on the one hand, if $\farris{S} = \Gamma$, then, since $\ID \oneVec = \zeroVec$, we get
	\[
        \Sigma
		=
		\ID \lr{-\half \Gamma} \ID
		=
		\ID \lr{-\half \farris{S}} \ID
		=
		\ID S \ID.
	\]
	On the other hand, for symmetric $S$ and writing $v:=S\eV$ with $\eV = d^{-1} \oneVec$, we have, by $\Proj = I_d - \eV \oneVec\T = I_d - \oneVec \eV\T$ and the linearity of $\farris{}$,
    \begin{align*}
        \farris{\Proj S \Proj}
        &=
        \farris{\Proj \lr{S - v\oneVec\T}} \\
        &=
        \farris{S - \oneVec v\T - v \oneVecT + \oneVec \eV\T v \oneVec\T}
        \\ &=
        \farris{S} - \farris{\oneVec v\T + v \oneVecT} + \farris{\oneVec \eV\T v \oneVec\T}
        \\ &=
        \farris{S}
        ,
    \end{align*}
    as $\diag{\oneVec v\T} = v = \diag{v \oneVec\T}$ and thus $\farris{\oneVec v\T + v \oneVec\T} = 0 \in \Rdd$.
    It follows that $\ID S \ID = \Sigma = \ID\lr{-\half\Gamma}\ID$ implies
    \begin{equation*}
        \farris{S}
        =
        \farris{\Proj S \Proj}
        =
        \farris{\Proj \lr{-\half \Gamma} \Proj}
        =
        \farris{-\half \Gamma}
        =
        \Gamma
        ,
    \end{equation*}
    since $\diag{\Gamma} = \zeroVec$.
\end{proof}

\begin{lemma}
    \label{lemma:SigmaPosSemiDef}
    The matrix $\Sigma$ is symmetric and \PSD{}, and its kernel is $\laSpan{\set{\oneVec}}$.
\end{lemma}
\begin{proof}
    Since $\ID = I_d - d^{-1} \oneVec \oneVec\T$ is the projection matrix onto the subspace orthogonal to $\oneVec$,
    it follows that $\Sigma \oneVec = \ID \lr{-\half\Gamma} \lr{\ID \oneVec} = \zeroVec$.
    Hence, $\laSpan{\set{\oneVec}} \subseteq \ker \Sigma$.
    Next, consider any $v \notin \laSpan{\set{\oneVec}}$,
    and let $w := \ID v$.
    This $w$ satisfies $w \neq \zeroVec$ and $w \perp \oneVec$,
    hence
    \begin{align*}
        v\T \Sigma v
        &=
        v\T \lr{\ID \lr{-\half \Gamma} \ID} v
        \\ &=
        -\half w\T \Gamma w
        \\
        &>
        0
        ,
    \end{align*}
    because of the conditional negative definiteness in the definition of a variogram matrix in~\cref{eq:defVariogramSet}.
    Symmetry follows from the symmetry of $\Gamma$ and $\ID$.
\end{proof}

\begin{lemma}
    \label{lemma:SigmatInvertible}
    There exists $t_0 \in \Rd[]$
    such that $\Sigma\t$ is invertible for all $t \neq t_0$.
    Specifically, $t_0 = -\tlr{\oneVec\T S^{-1} \oneVec}^{-1}$ if $S$ is invertible and $t_0 = 0$ otherwise.
\end{lemma}
\begin{proof}
    To prove the invertibility of $\Sigma\t$
    for $t \neq t_0$ 
    consider the following two cases.

    \underline{Case 1:}
    $S$ is invertible.
    Let $v_0 = S\inv \oneVec$.
    Then $\oneVecT v_0 \neq 0$, because otherwise, if $\oneVecT v_0 = 0$,
    then $v_0$ would be orthogonal to the kernel of $\Sigma$,
    and we would have $\ID v_0 = v_0$, and therefore
    \begin{align*}
        \zeroVec
        \neq
        \Sigma v_0
        =
        \ID S \ID v_0
        =
        \ID S v_0
        =
        \ID \oneVec
        =
        \zeroVec
        ,
    \end{align*}
    a contradiction.
    Next,
    let $t_0 = - \tlr{\oneVecT v_0}\inv$.
    Consider $t \neq t_0$ and suppose that $\Sigma\t$ is singular,
    i.e., there exists $v \in \Rd\without{\zeroVec}$ such that $\zeroVec\T
    = v\T \lr{t \oneVecTT + S}$. Then
    \begin{align*}
        0
        &=
        v\T
        \lr{
            t \oneVecTT
            +
            S
        }
        v_0
        \\
        &=
        t v\T \oneVecTT v_0
        +
        v\T S v_0
        \\
        &=
        t
        \lr{v\T \oneVec}
        \lr{\oneVecT v_0}
        +
        \lr{v\T \oneVec}
        \\
        &=
        \lr{v\T \oneVec}
        \lr{-t \cdot t_0\inv + 1}
        .
    \end{align*}
    For $t\neq t_0$, the second factor cannot be zero,
    so it must hold that $\oneVecT v = 0$.
    However, this leads to the contradiction
    \begin{align*}
        \zeroVec
        =
        \lr{
            t \oneVecTT
            +
            S
        }
        v
        =
        Sv
        \neq
        \zeroVec
        .
    \end{align*}
    Furthermore, it holds that
    \begin{align*}
        \Sigma\t[t_0]v_0
        =
        \lr{
            S - \frac{\oneVecTT}{\oneVecT S\inv \oneVec}
        } S\inv \oneVec
        =
        \oneVec
        -
        \oneVec
        =
        \zeroVec
        .
    \end{align*}

    \underline{Case 2:}
    $S$ is singular.
    Then $S$ must be of rank $d-1$,
    since the rank of $\Sigma = \ID S \ID$ is $d-1$ by \cref{lemma:SigmaPosSemiDef}.
    Let $v_0 \neq \zeroVec$ be such that
    $\ker S = \mathrm{span}\lr{\set{v_0}}$.
    Then, $v_0 \notin \set{\oneVec}^\perp$,
    because otherwise
    \begin{align*}
        \zeroVec
        \neq
        \Sigma v_0
        =
        \ID S \ID v_0
        =
        \ID S v_0
        =
        \zeroVec
        .
    \end{align*}
    Furthermore, $\oneVec \notin \Image S$,
    because otherwise
    there must be a vector $u \notin \laSpan{\set{v_0}}$,
    such that $Su = \oneVec$.
    Then,
    $\lr{\ID S} v_0 = \lr{\ID S} u = \zeroVec$,
    implying that the rank of $\ID S \ID$ is at most $d-2$, which is a contradiction.

    Setting $t_0 := 0$,
    and using
    $\Image S = S \Rd = S \tlr{\set{\oneVec}^\perp \oplus v_0\Rd[]} = S \set{\oneVec}^\perp$,
    the image of $\Sigma\t$ for
    $t \neq t_0$
    can then be checked to be
    \begin{align*}
        \Image \Sigma\t
        &=
        \Sigma\t \Rd
        \\
        &=
        \Sigma\t \lr{
            \set{\oneVec}^\perp
            \oplus
            v_0 \Rd[]
        }
        \\
        &=
        \lr{\Sigma\t \set{\oneVec}^\perp}
        \oplus
        \lr{\Sigma\t v_0 \Rd[]}
        \\
        &=
        \lr{S \set{\oneVec}^\perp}
        \oplus
        \lr{t \oneVec \oneVecT v_0 \Rd[]}
        \\
        &=
        \Image S
        \oplus
        \oneVec \Rd[]
        \\
        &=
        \Rd
        ,
    \end{align*}
    implying that $\Sigma\t$ is invertible.
    Again, it holds that $\Sigma\t[t_0]v_0 = S v_0 = \zeroVec$.
\end{proof}
\begin{lemma}
    \label{lemma:SigmaTPosDef}
    $\Sigma\t$ is \PD{}
    (not necessarily symmetric)
    for all $t > t_0$.
\end{lemma}

\begin{proof}
    A matrix $A \in \Rdd$ is positive (semi-)definite if and only if its symmetric part $\check{A} = \half\tlr{A+A\T}$ is so; indeed, $x\T A x = x\T A\T x = x\T \check{A} x$ for any $x \in \Rd$. The symmetric part of $\Sigma\t$ is $\check\Sigma\t = \half\tlr{\Sigma\t + \tlr{\Sigma\t}\T} = t \oneVec \oneVec\T + \check{S}$ where $\check{S} = \half\tlr{S + S\T}$ is the symmetric part of $S$. Further, since $\Sigma = \ID \tlr{-\half\Gamma} \ID$ is symmetric, the matrix $S \in \Rdd$ satisfies $\ID S \ID = \Sigma$ if and only if $\ID S\T \ID = \Sigma$ and thus if and only if $\ID \check{S} \ID = \Sigma$. Hence, to show that $\Sigma\t$ is positive definite for $t > t_0$, we can without loss of generality assume $S$ is symmetric (or more precisely, replace $S$ by its symmetric part).

	So assume $S$ is symmetric and let $v_0$ be the vector from the proof of \cref{lemma:SigmatInvertible}.
    First, we show that $\Sigma\t[t_0]$ is \PSD; to do so, it is sufficient to show that its non-zero eigenvalues are positive.
    To this end, recall that $\Sigma\t[t_0] v_0 = \zeroVec$
    and consider an eigenvector $v_1 \not\in \laSpan{\{v_0\}}$ of $\Sigma\t[t_0]$
    with eigenvalue $\alpha_1 \ne 0$.
    The vector
    \begin{align*}
        u
        =
        \lr{\oneVecT v_0}
        v_1
        -
        \lr{\oneVecT v_1}
        v_0
    \end{align*}
    satisfies $u \perp \oneVec$ and thus $\ID u = u$. Moreover, since $\oneVecT v_0 \neq 0$,
    it follows that $u \neq \zeroVec$.
    Hence, since $\Gamma$ is \CND{} and since $v_0$ and $v_1$ are orthogonal (as eigenvectors of the symmetric matrix $\Sigma\t[t_0]$ associated to distinct eigenvalues), we have
    \begin{align*}
        0
        &<
        u\T
        \lr{
            - \half \Gamma
        }
        u
        = u\T \ID \lr{-\half\Gamma} \ID u = u\T S u
        =
        u\T
        \Sigma\t[t_0]
        u \\
        &= \lr{\oneVec\T v_0}^2 \norm{v_1}^2 \alpha_1
    \end{align*}
    implying $\alpha_1 > 0$, as required.
    A similar argument shows that $\Sigma\t[t_0]$ has rank $d-1$: otherwise, we could find an eigenvector $v_1 \not\in \laSpan{\{v_0\}}$ of $\Sigma\t[t_0]$ orthogonal to $v_0$ with eigenvalue $0$ as well, and this would lead to a contradiction by the same calculation as above.
    For $t>t_0$, the matrix $\Sigma\t$ is invertible by \cref{lemma:SigmatInvertible},
    and since it is the sum of the two \PSD{} matrices
    $\Sigma\t[t_0]$ and $\lr{t - t_0}\oneVecTT$,
    it is also \PSD{} and hence \PD{}.
\end{proof}

\begin{lemma}
    \label{lemma:SigmaTLimit}
    The limit
    $\limit{t} \tlr{\Sigma\t}\inv$
    exists,
    is symmetric,
    and its kernel contains $\oneVec$.
\end{lemma}
\begin{proof}
    To prove the existence and claimed properties of
    $\limit{t} \tlr{\Sigma\t}\inv$,
    let $\set{v_1, \dots, v_{d-1}}$
    be a basis of $\set{\oneVec}^\perp$ in $\Rd$.
    Recall from the proof of \cref{lemma:SigmaTPosDef} that $\Sigma\t[t_0]$ has rank $d-1$,
    and let $\tilde v_0$ be its left null vector,
    i.e., $\tilde v_0\T \Sigma\t[t_0] = \zeroVec\T$.
    Repeating the arguments in the proof of \cref{lemma:SigmatInvertible} for $\tilde S = S\T$
    shows that $\tilde v_0\T \oneVec \neq 0$.
    Define the matrices $W, \tilde W \in \Rdd$ and $V \in \Rd[d \times (d-1)]$ by
    \begin{align*}
        V
        &=
        \lr{
            v_1, \dots, v_{d-1}
        }
        \\
        W
        &=
        \lr{
            v_0, v_1, \dots, v_{d-1}
        }
        \\
        \tilde W
        &=
        \lr{
            \tilde v_0, v_1, \dots, v_{d-1}
        }
    \end{align*}
    Note that $\ID V = V$.
    Then
    \begin{align*}
        \tilde W\T \oneVecTT W
        &=
        \begin{pmatrix}
            c & \zeroVec\T \\
            \zeroVec & \zeroVec \zeroVec\T
        \end{pmatrix}
        , \\
        \tilde W\T \Sigma\t[t_0] W
        &=
        \begin{pmatrix}
            0 & \zeroVec\T \\
            \zeroVec & C
        \end{pmatrix}
        ,
    \end{align*}
    for some constant $c \neq 0$ and $C \in \Rdd[(d-1)]$ satisfying
    \begin{align*}
        C
        &=
        V\T
        \Sigma\t[t_0]
        V
        \\ &=
        V\T \ID
        \Sigma\t[t_0]
        \ID V
        \\ &=
        V\T
        \Sigma
        V
        ,
    \end{align*}
    which is symmetric \PD{}, since $\Sigma$ is symmetric \PSD{} with kernel $\laSpan{\set{\oneVec}}$.
    Hence, using $\Sigma\t = (t-t_0) \oneVec \oneVec\T + \Sigma\t[t_0]$, we have
    \begin{align*}
        \Sigma\t
        &=
        (\tilde W\T)\inv
        \begin{pmatrix}
            c \lr{t-t_0} & \zeroVec\T \\
            \zeroVec & C
        \end{pmatrix}
        W\inv
        , \\
        \Longrightarrow \quad
        \lr{\Sigma\t}\inv
        &=
        W
        \begin{pmatrix}
            c\inv \lr{t-t_0}\inv & \zeroVec\T \\
            \zeroVec & C\inv
        \end{pmatrix}
        \tilde W\T
        , \\
        \Longrightarrow \quad
        \limit{t}
        \lr{\Sigma\t}\inv
        &=
        W
        \begin{pmatrix}
            0 & \zeroVec\T \\
            \zeroVec & C\inv
        \end{pmatrix}
        \tilde W\T
        .
    \end{align*}
    The latter matrix is symmetric since $W$ and $\tilde W$ differ only in the first column,
    and since $v_j \perp \oneVec$ for all $j = \oneToX{d-1}$,
    its kernel contains $\oneVec$.
\end{proof}

\begin{lemma}
    \label{lemma:PSPeqLim}
    For $t>t_0$, let
    $\Theta\t := \tlr{\Sigma\t}\inv$.
    Then
    \begin{align*}
        \Sigma\pinv
        =
        \limit{t}{\Theta\t}
        .
    \end{align*}
\end{lemma}

\begin{proof}
	We will check the two criteria in \cref{lemma:conditionsPinv}.
	Since $\ID = I_d - d^{-1} \oneVec \oneVec\T$ and since the vector $\oneVec$ belongs to the kernel of the symmetric matrix $\limit{t}{\Theta\t}$, we have $\ID \limit{t}{\Theta\t} = \limit{t}{\Theta\t}$ and thus
    \begin{align*}
        \Sigma
        \limit{t} \Theta\t
        &=
        \ID S \ID
        \limit{t} \Theta\t
        \\
        &=
        \ID
        \limit{t}
        S \Theta\t
        \\
        &=
        \ID
        \limit{t}
        \lr{
            S + \oneVecTT t
            - \oneVecTT t
        }
        \lr{
            S + \oneVecTT t
        }\inv
        \\
        &=
        \ID
        -
        \limit{t}
        \ID \oneVecTT t
        \lr{
            S + \oneVecTT t
        }\inv
        \\
        &=
        \ID
        .
    \end{align*}
    Since $\Sigma = \ID S \ID$, it follows that $\Image \ID = \Image \Sigma$ and thus that $\ID$ is the projection matrix onto the image of $\Sigma$.
    Further, the results from \cref{lemma:SigmaPosSemiDef}
    and \cref{lemma:SigmaTLimit} yield
    \begin{align*}
        \ker \Sigma
        =
        \laSpan{\set{\oneVec}}
        \subseteq
        \ker \limit{t}{\Theta\t}
        .
    \end{align*}
    Hence, the claimed equality follows from \cref{lemma:conditionsPinv} with $A = \Sigma$ and $B = \limit{t}{\Theta\t}$.
\end{proof}

\begin{lemma}
    \label{lemma:GammaLimTheta}
    For $\Theta$ from \cref{def:Theta}
    it holds that
    \begin{align*}
        \limit{t} \lr{t \oneVecTT - \half \Gamma}\inv
        =
        \Theta
        .
    \end{align*}
\end{lemma}
\begin{proof}
    For $d=2$, all valid variogram matrices are of the form
    $\Gamma_{11} = \Gamma_{22} = 0$ and
    $\Gamma_{12} = \Gamma_{21} =: \eta > 0$,
    and the expressions can directly be checked to be equal: we have $\Sigma^{(1)} = \Sigma^{(2)} = \gamma$ and
    \begin{align*}
    	\lr{t \oneVecTT - \half \Gamma}\inv
    	&= \begin{bmatrix} t & t -\half \eta \\ t-\half \eta & t \end{bmatrix}\inv \\
    	&= \frac{1}{t^2 - (t-\half\eta)^2}
    	\begin{bmatrix} t & -t + \half \eta \\ -t + \half\eta & t \end{bmatrix} \\
    	&= \frac{1}{\half\eta (2t - \half\eta)}     	\begin{bmatrix} t & -t + \half \eta \\ -t + \half\eta & t \end{bmatrix} \\
    	&\to \frac{1}{\half\eta} \begin{bmatrix} \phantom{-}\half & -\half \\[1ex] -\half & \phantom{-}\half \end{bmatrix}
    	= \begin{bmatrix} \phantom{-}1/\eta & -1/\eta \\ -1/\eta & \phantom{-}1/\eta \end{bmatrix} = \Theta, \qquad t \to \infty.
    \end{align*}

    For $d \ge 3$,
    write $V = \set{\oneToX{d}}$,
    consider $\Sigma\t = t \oneVecTT - \half \Gamma$ (which is \PD{} for large enough $t$),
    let ${Y}\t \sim \normal{0, \Sigma\t}$,
    and let ${Y}\tk$ be the random vector
    ${Y}\t$ conditioned on the event $\tset{{Y}\t_k = y_k}$
    for some $y_k \in \mathbb{R}$.
    Then the $(d-1)$-dimensional random vector $Y_{V\without{k}}\tk = (Y_i\tk)_{i \in V\without{k}}$ is multivariate normal with covariance matrix $\Sigma\tk$ given by
    \begin{align*}
        \Sigma\tk
        &=
        \Sigma\t_{\backslash k}
        -
        \Sigma\t_{\cdot, k}
        \lr{
            \Sigma\t_{k,k}
        }\inv
        \Sigma\t_{k, \cdot}
        \\
        &=
        \lr{
            \oneVecTT t
            -
            \half \Gamma_{\backslash k}
        }
        -
        \lr{
            \oneVec t
            -
            \half \Gamma_{\cdot, k}
        }
        \frac{1}{t}
        \lr{
            \oneVecT t
            -
            \half \Gamma_{k, \cdot}
        }
        \\
        &=
        \oneVecTT t
        -
        \half \Gamma_{\backslash k}
        - \oneVecTT t
        +
        \half \oneVec \Gamma_{k, \cdot}
        +
        \half \Gamma_{\cdot, k} \oneVecT
        +
        \frac{1}{4t} \Gamma_{\cdot, k} \Gamma_{k, \cdot}
        \\
        &=
        \half \lr{
            \oneVec \Gamma_{k, \cdot}
            +
            \Gamma_{\cdot, k} \oneVecT
            -
            \Gamma_{\backslash k}
        }
        +
        \littleO{1}
        \\
        &=
        \Sigma\k + \littleO{1}
        , \qquad t \to \infty,
    \end{align*}
    with $\Sigma\k$ as in \cref{def:HRdensity}.
    Since $\Sigma\k$ is \PD{}, so is $\Sigma\tk$ for sufficiently large $t$, and the respective precision matrices satisfy
    \begin{align*}
        \Theta\tk
        &:=
        \tlr{\Sigma\tk}\inv
        \longrightarrow
        \tlr{\Sigma^{(k)}}\inv
        =:
        \Theta^{(k)}
        , \quad
        t
        \rightarrow
        \infty
        .
    \end{align*}
    Recall $\Theta\t = \tlr{\Sigma\t}\inv$. Classical properties of the Schur complement of a block matrix \citep[see, e.g.][Theorem~2.5]{rue2005} imply
    \begin{align*}
        \Theta\tk = \tlr{\Theta\t}_{\lr{V\without{k}} \times \lr{V\without{k}}}
    \end{align*}
    and hence
    \begin{align*}
        \Theta\t_{ij}
        \rightarrow
        \Theta\k_{ij}
        ,\qquad
        i,j \neq k
        , \quad
        t \rightarrow \infty
        .
    \end{align*}
    Since $\Theta$ is defined by $\Theta_{ij} = \Theta\k_{ij}$,
    this implies
    \begin{equation*}
        \limit{t} \lr{t \oneVecTT - \half \Gamma}\inv
        =
        \Theta
        .
        \qedhere
    \end{equation*}
\end{proof}

\begin{proofOf}[prop:ThetaReps]
    The implication
    $\gamma\lr{S} = \Gamma \Rightarrow \Proj S \Proj = \Proj \lr{-\half \Gamma} \Proj$
    follows directly from \cref{lemma:farrisSeqG}.
    Combining the results above yields
    \begin{alignat*}{4}
        \Theta
        &=
        \limit{t} \lr{t \oneVecTT - \half \Gamma}\inv
        && \qquad \text{(\cref{lemma:GammaLimTheta})}
        \\ &=
        \lr{\Proj \lr{-\half \Gamma} \Proj}\pinv
        && \qquad \text{(\cref{lemma:PSPeqLim})}
        \\ &=
        \lr{\Proj S \Proj}\pinv
        && \qquad \text{(\cref{lemma:farrisSeqG})}
        \\ &=
        \limit{t} \lr{t \oneVecTT + S}\inv
        .
        && \qquad \text{(\cref{lemma:PSPeqLim})}
      \end{alignat*}
      Using the invertibility of the map that sends a matrix to its \MP{} inverse,
      the equality $\Theta = \lr{\Proj \lr{-\half \Gamma} \Proj}\pinv$  immediately implies that
      $\ID S \ID = \ID \lr{-\half \Gamma} \ID$
      is a necessary condition for $\lr{\ID S \ID}\pinv = \Theta$.
\end{proofOf}

\subsection{\texorpdfstring{$S$ as generalized inverse of $\Theta$}{S as generalized inverse of Theta}}

The following result yields an interesting alternative characterization of
the matrices $S$ that can be used in \cref{prop:ThetaReps},
in terms of generalized inverses,
of which the \MP{} inverse is a specific instance.

\begin{lemma}
    \label{lemma:SisGeneralizedInverse}
    Let $\Gamma \in \Dcal_d$, $\Sigma = \ID (-\half\Gamma) \ID$,
    and $\Theta = \Sigma\pinv$.
    For $S \in \Rdd$,
    we have $\ID S \ID = \Sigma$ if and only if
    $S$ is a generalized inverse of $\Theta$, i.e.,
    \begin{align}
        \label{eq:lemmaGenInv2}
        \Theta S \Theta &= \Theta
        .
    \end{align}
\end{lemma}

\begin{proof}
    First, suppose $\ID S \ID = \Sigma = \Theta\pinv$.
    Since $\ID$ is the projection matrix onto the image of $\Sigma$ and thus also onto the image of $\Theta = \Sigma\pinv$ (see proof of \cref{lemma:PSPeqLim}),
    \begin{align*}
        \Theta S \Theta
        =
        \Theta \ID S \ID \Theta
        =
        \Theta \Theta\pinv \Theta
        =
        \Theta
        .
    \end{align*}

    Second, let \cref{eq:lemmaGenInv2} be satisfied. Since $\Theta\pinv \Theta = \Theta \Theta\pinv$ is the projection matrix onto the image of $\Theta$ and thus equal to $\ID$, we find
    \[
    	\ID S \ID
    	= \Theta\pinv \Theta S \Theta \Theta\pinv
    	= \Theta\pinv \Theta \Theta\pinv
    	= \Theta\pinv
    	= \Sigma.
    	\qedhere
    \]
\end{proof}

\subsection{Proof of \cref{prop:thetaHomeo}}

\begin{definition}
    \label{def:gamma}
    The covariance transform $\gamma$ is defined for
    arbitrary square matrices as
    \begin{alignat*}{4}
        \myFunctionDefinition{\gamma}{\Rdd}{\Rdd}{S}{
            \oneVec \diag{S}\T
            +
            \diag{S} \oneVec\T
            -
            2 S
        }
        .
    \end{alignat*}
\end{definition}
This definition does not require $S$ to have any special properties
but also does not guarantee many useful properties of $\gamma\lr{S}$.
A more useful mapping can be obtained by restricting the domain as follows.
\begin{lemma}
    \label{lemma:gammaMapsToCND}
    Recall $\Dcal_d$ from \cref{eq:defVariogramSet}.
    Then $\gamma\lr{\Sigma} \in \Dcal_d$ for any
    symmetric, \PSD{}
    $\Sigma \in \Rdd$ satisfying
    \begin{align*}
        \ker\Sigma \cap \set{\oneVec}^\perp = \set{\zeroVec}
        .
    \end{align*}
\end{lemma}
\begin{proof}
    For $\Sigma$ satisfying the conditions above
    let $\Gamma = \gamma\lr{\Sigma}$.
    Then clearly
    $\Gamma = \Gamma\T$ and $\diag{\Gamma} = \zeroVec$.
    Moreover, for $\zeroVec \neq v \in \set{\oneVec}^\perp$,
    we have
    \begin{align*}
    	v\T \Gamma v
    	= v\T \farris{\Sigma} v
    	= v\T \tlr{\oneVec \diag{\Sigma}\T + \diag{\Sigma} \oneVec\T - 2 \Sigma} v
    	= -2 v\T \Sigma v < 0
        ,
    \end{align*}
    implying that $\Gamma \in \Dcal_d$.
\end{proof}

\begin{proofOf}[prop:thetaHomeo]
    \cref{lemma:SigmaPosSemiDef,lemma:gammaMapsToCND}
    show that $\sigma$ and $\gamma$
    do indeed map between $\Dcal_d$ and $\Pcal_d^\oneVec$.
    Since both maps consist only of elementary matrix operations
    and the continuous map $\diag{\cdot}$,
    they are also continuous.
    Using the intermediate results from the proof of \cref{lemma:farrisSeqG},
    it can be seen that
    \begin{align*}
        \farris{\sigma\lr{\Gamma}}
        &=
        \farris{\ID \lr{-\half \Gamma} \ID}
        =
        \Gamma
        , \qquad
        \forall \Gamma \in \Dcal_d
        , \\
        \sigma\lr{\farris{\Sigma}}
        &=
        \ID \lr{-\half \farris{\Sigma}} \ID
        =
        \Sigma
        , \qquad
        \forall \Sigma \in \Pcal_d^\oneVec
        ,
    \end{align*}
    implying that $\sigma$ is bijective between $\Dcal_d$ and $\Pcal_d$ with continuous inverse $\sigma\inv = \farris{}$.

    By \cref{lemma:propertiesMP}, \cref{lemma:propertiesMP:SVD},
    the \MP{} inverse of a symmetric \PSD{} matrix is again symmetric \PSD{}.
    Hence,
    $\theta$ also maps $\Dcal_d$ into $\Pcal_d^\oneVec$ and
    can be written as the composition $\varphi \circ \sigma$,
    with
    \begin{alignat*}{3}
        \myFunctionDefinition{\varphi}{\Pcal_d^\oneVec}{\Pcal_d^\oneVec}{\Sigma}{\Sigma\pinv}
        .
    \end{alignat*}
    As shown in \cite{rakocevic1997},
    the map $\varphi$ is continuous, and hence is a continuous, idempotent bijection; note that it is surjective since every $\Sigma \in \Pcal_d^\oneVec$ can be written as $\Sigma = \varphi(\varphi(\Sigma))$.
    Therefore,
    $\theta$ is also continuous with continuous inverse $\theta\inv = \farris{} \circ \varphi$.
\end{proofOf}

\subsection{Proof of \cref{prop:HrDensitySimple}}

\label{subsec:proofDensity}

Before showing \cref{prop:HrDensitySimple}, we establish a formula that allows to simplify the proportionality constant occurring in the density expression.

\begin{lemma}
	\label{lem:constantdiagonal}
	Let $\Gamma \in \mathcal{D}_d$ and put $\Theta = \lr{\ID \lr{-\half\Gamma} \ID}\pinv$. Then
	\begin{align*}
		\lr{-\half\Gamma} \Theta \lr{-\half\Gamma}
		&= -\half\Gamma + c(\Gamma) \oneVec \oneVec\T \qquad \text{where} \\
		c(\Gamma) &= \eV\T \lr{-\half\Gamma} \Theta \lr{-\half\Gamma} \eV - \eV\T \lr{-\half\Gamma} \eV.
	\end{align*}
	Since $\Gamma$ has a zero diagonal, this implies that $\Gamma \Theta \Gamma$ has a constant diagonal.
\end{lemma}

\begin{proof}
    Recall the definitions $\Sigma = \ID \lr{-\half \Gamma} \ID$, $\eV = d^{-1} \oneVec$ and $\ID = I_d - \oneVec \eV\T = I_d - \eV \oneVec\T$. Since $\ID \Theta = \Theta \ID = \Theta$, we have	
	\begin{align*}
		\lr{-\half\Gamma} \Theta \lr{-\half\Gamma}
		&= \lr{\ID + \oneVec \eV\T} \lr{-\half\Gamma} \ID \Theta \lr{-\half\Gamma} \\
		&= \ID \lr{-\half\Gamma} \ID \Theta \lr{-\half\Gamma}
		+ \oneVec \eV\T \lr{-\half\Gamma} \ID \Theta \lr{-\half\Gamma} \\
		&= \ID \lr{-\half\Gamma} + 
		\oneVec \eV\T \lr{-\half\Gamma} \Theta \lr{-\half\Gamma} \\
		&= -\half\Gamma - \oneVec \eV\T \lr{-\half\Gamma}
		+ \oneVec \eV\T \lr{-\half\Gamma} \Theta \lr{-\half\Gamma}.
	\end{align*}
	Since $\Theta \ID \lr{-\half\Gamma} \ID = \Theta \Sigma = \ID$, we have
	\begin{align*}
		\oneVec \eV\T \lr{-\half\Gamma} \Theta \lr{-\half\Gamma}
		&= \oneVec \eV\T \lr{-\half\Gamma} \Theta \ID \lr{-\half\Gamma} \lr{\ID + \eV \oneVec\T} \\
		&= \oneVec \eV\T \lr{-\half\Gamma} \Theta \ID \lr{-\half\Gamma} \ID +
		\oneVec \eV\T \lr{-\half\Gamma} \Theta \lr{-\half\Gamma} \eV \oneVec\T \\
		&= \oneVec \eV\T \lr{-\half\Gamma} \ID
		+ \underbrace{\lr{\eV\T \lr{-\half\Gamma} \Theta \lr{-\half\Gamma} \eV}}_{\text{scalar}} \oneVec \oneVec\T
	\end{align*}
	with
	\begin{align*}
		\oneVec \eV\T \lr{-\half\Gamma} \ID
		&= \oneVec \eV\T \lr{-\half\Gamma} - \oneVec \eV\T \lr{-\half\Gamma} \eV \oneVec\T \\
		&= \oneVec \eV\T \lr{-\half\Gamma} - \underbrace{\lr{\eV\T \lr{-\half\Gamma} \eV}}_{\text{scalar}} \oneVec \oneVec\T.
	\end{align*}
	Adding up, we get the stated identity.
\end{proof}

\begin{proofOf}[prop:HrDensitySimple]
    For $k \in \set{1,\ldots,d}$, let $\lambda\k$ denote the density expression
    from \cref{def:HRdensity} with that specific choice of $k$,
    and recall the notation $\norm[M]{v}^2 = v\T M v$.
    Extending the notation used there,
    let
    \begin{align*}
        \tilde\mu\k
        =
        \lr{
            -\half
            \Gamma_{ik}
        }_{i = \oneToX{d}}
        \in
        \Rd
        .
    \end{align*}
    Note that
    \begin{align*}
        \tlr{
            y - \oneVec y_k - \tilde\mu\k
        }_k
        = 0
        ,
    \end{align*}
    so that $\lambda\k$ can be rewritten as
    \begin{align*}
        \lambda\k(y; \Gamma)
        &=
        \frac{
            \expF{-y_k}
        }{
            \sqrt{\lr{2\pi}^{d-1}\abs{\Sigma\k}}
        }
        \expF{
            -\half
            \tnorm[\Theta]{
                y - \oneVec y_k - \tilde\mu\k
            }^2
        }
        .
    \end{align*}
    From \citet[][Proof of Lemma~5.1]{roettger2023}
    and \cref{lemma:pseudoDet1}
    it follows that
    $\abs{\Sigma\k} = \abs{\Theta\k}\inv = \lr{d\inv \pdet{\Theta}}\inv = d \pdet{\Sigma}$
    for all $k \in V$.
    Using $\Theta\oneVec = \zeroVec$
    and applying logarithms, the above equation becomes
    \begin{align}
        \log\lambda\k\lr{y; \Gamma}
        &=
        \log c_1
        -y_k
        -\half
        \tnorm{
            y - \tilde\mu\k
        }^2,
    	\quad \text{where}
        \nonumber
        \\
    	c_1
        &=
        \lr{2\pi}^{-(d-1)/2}\lr{d\inv\pdet{\Theta}}^{1/2}
        .
        \label{eq:c1}
    \end{align}
    Then, since the value of $\lambda\k$ is the same for all $k \in \set{1,\ldots,d}$, we find
    \begin{align}
    \nonumber
        \log\lambda\lr{y; \Gamma}
        &=
        \sum_{k=1}^d d\inv \log\lambda\k\lr{y; \Gamma}
        \\ &=
        \sum_{k=1}^d d\inv \lr{
            \log c_1 - y_k - \half
            \tnorm[\Theta]{
                y - \tilde\mu\k
            }^2
        }
   \nonumber
        \\ &=
        \log c_1
        -
        \eV\T y
        - \half
        \sum_{k=1}^d
        d\inv
        \tnorm[\Theta]{
            y - \tilde\mu\k
        }^2
        ,
        \label{eq:auxlambda}
    \end{align}
    where $\eV = d\inv \oneVec$.
    Using the notation $\bilin[\Theta]{v,w} = v\T\Theta w$,
    the last sum can be rewritten as
    \begin{align*}
    \nonumber
        \sum_{k=1}^d d\inv
        \tnorm[\Theta]{
            y - \tilde\mu\k
        }^2
        &=
        \norm[\Theta]{y}^2
        +
        \sum_{k=1}^d d\inv
        \bilin[\Theta]{y, \Gamma_{\cdot k}}
        +
        \sum_{k=1}^d d\inv
        \norm[\Theta]{\half \Gamma_{\cdot k}}^2
        \\ &=
        \norm[\Theta]{y}^2
        +
        \bilin[\Theta]{y, \Gamma \eV}
        +
        \frac{1}{d}
        \tr\lr{
            \lr{\half\Gamma}
            \Theta
            \lr{\half\Gamma}
        }
        .
    \end{align*}
    Using \cref{lem:constantdiagonal},
    the last term can be expressed as
    \begin{align*}
        \frac{1}{d}
        \tr\lr{
            \lr{-\half\Gamma}
            \Theta
            \lr{-\half\Gamma}
        }
        &=
        \frac{
            \tr\lr{\oneVecTT}
        }{
            d
        }
        \lr{
            \eV\T \lr{\half\Gamma} \Theta \lr{\half\Gamma} \eV + \eV\T \lr{\half\Gamma} \eV
        }
        \\
        &=
        \eV\T \lr{
            \lr{\half\Gamma} \Theta \lr{\half\Gamma}
            +
            \half\Gamma
        } \eV
    	.
    \end{align*}

    Plugging back into \cref{eq:auxlambda} yields
    \begin{align*}
        \log \lambda \lr{y; \Gamma}
        &=
        \log c_1
        - \eV\T y
        - \half \norm[\Theta]{y}
        - \half \bilin[\Theta]{y, \Gamma\eV}
        - \half
        \eV\T \lr{
            \lr{\half\Gamma} \Theta \lr{\half\Gamma}
            +
            \half\Gamma
        } \eV
        \\
        \lambda \lr{y; \Gamma}
        &=
        c_1
        \cdot
        c_2
        \cdot
        \expF{
            -\half
            y\T \Theta y
            -
            y\T\eV
            +
            y\T \Theta\lr{-\half\Gamma}\eV
        }
    \end{align*}
    with $c_1$ as in \cref{eq:c1} and
    \begin{align*}
        c_2
        &=
        \expF{
            -\half
            \eV\T \lr{
                \lr{\half\Gamma} \Theta \lr{\half\Gamma}
                +
                \half\Gamma
            } \eV
        }
        .
        \qedhere
    \end{align*}
\end{proofOf}

\begin{remark}
    \label{remark:densitiesForDifferentMargins}
    We consider multivariate generalized Pareto distributions $Y$ with standard exponential univariate ``marginals'' $Y\I[\set{i}]$, $i \in V$.
    These are supported on $\L_Y = \{x \in \Rd : x \not\leq \zeroVec \}$ with exponent measure $\Lambda_Y$.
    Another common choice for the univariate ``marginals'' is the standard Pareto distribution \citep[e.g.,][]{engelke2020},
    in which case the multivariate Pareto distribution is that of $Z = \exp(Y)$,
    supported on $\L_Z = \{ x \in [0, \infty)^d :  x \not\leq \oneVec\}$ with exponent measure $\LambdaVec_Z(x) = \LambdaVec_Y( \log x)$ and density
    \begin{align*}
        \lambda_Z\lr{x}
        =
        \lr{\prod_{i = 1}^d x_i\inv}
        \lambda_Y\lr{\log x}, \quad x \in \L_Z
        .
    \end{align*}
    All results in this paper are still applicable with these (or other) marginals since they only concern the dependence structure.
\end{remark}

\subsection{Proofs of matrix completion results}

\subsubsection{Decomposable graphs}
\label{subsubsec:proofsDecomposableGraphs}

\begin{lemma}[\cite{bakonyi2011}, Theorem~2.2.3]
    \label{lemma:completeTwoCliques}
    The \PD{} block matrix completion problem
    \begin{equation}
    \label{eq:SigmaABC}
        \Sigma
        =
        \lr{
            \begin{matrix}
                \Sigma_{AA} & \Sigma_{AB} & \undefinedMatrixNoQuotes \\
                \Sigma_{BA} & \Sigma_{BB} & \Sigma_{BC} \\
                \undefinedMatrixNoQuotes & \Sigma_{CB} & \Sigma_{CC}
            \end{matrix}
        }
        , \qquad
        \lr{\Sigma\inv}_{AC}
        =
        0
    \end{equation}
    has the unique solution $\Sigma_{AC}=\Sigma_{AB}\Sigma_{BB}\inv \Sigma_{BC}$,
    and $\Sigma_{AC}=0$ in the case of an empty separator $B=\emptyset$.
    Here, $\undefinedMatrix$ denotes a matrix of adequate size with all entries $\undefined$.
\end{lemma}
\begin{proof}
    In the case $B = \emptyset$, the result follows immediately from the block-wise inversion of block diagonal matrices.
    For $B \neq \emptyset$, consider a normal random vector $X$ with covariance matrix $\Sigma$.
    The condition $\tlr{\Sigma\inv}_{AC}=0$ implies the conditional independence $X_A \perp X_C \vert X_B$,
    and hence $\Cov{X_A, X_C \mid X_B} = 0$.
    Using e.g. \citet[Section~2.1.7]{rue2005} to compute this conditional covariance yields
    \begin{align*}
        \Sigma_{AC} - \Sigma_{AB}\Sigma\inv_{BB}\Sigma_{BC} = 0,
    \end{align*}
    from which the \namecref{lemma:completeTwoCliques} immediately follows.
\end{proof}

\begin{lemma}
    Let $G=(V,E)$, for $V = \set{\oneToX{d}}$, be an undirected graph,
    with some node $k \in V$ being connected to all other nodes.
    Let $\GammaO$ be a matrix specified on $G$
    that is \CND{} on all fully specified principal submatrices.
    Let $\SigmaO\k = \varphi_k(\GammaO)$ for $\varphi_k$ in \cref{def:HRdensity},
    for ease of notation indexed by $V\without{k}$.

    Then $\SigmaO\k$ is \PD{} on all fully specified principal submatrices
    and preserves specified entries in the sense that
    \begin{align*}
        \GammaO_{ij} \neq \undefined
        \,\Longleftrightarrow\,
        \SigmaO\k_{ij} \neq \undefined
        , \qquad \forall
        i,j \neq k
        .
    \end{align*}
\end{lemma}
\begin{proof}
    To show that specified entries are preserved,
    recall that
    \begin{align*}
        \SigmaO\k_{ij}
        =
        \half \lr{
            \GammaO_{ik} + \GammaO_{jk} - \GammaO_{ij}
        }
        , \qquad
        i,j \neq k
        .
    \end{align*}
    Since $k$ is connected to all other nodes in $G$,
    $\GammaO_{ik}$ is specified for all $i \in V\without{k}$,
    yielding the claimed equivalence.

    To show positive definiteness of fully specified principal submatrices,
    let $M$ be the index set of such a submatrix and observe
    \begin{align*}
        \SigmaO\k_{M,M}
        =
        \lr{\varphi_k\tlr{\GammaO}}_{M,M}
        =
        \varphi_k\lr{\GammaO_{M \cup \set{k},M \cup \set{k}}}
        ,
    \end{align*}
    which is positive definite.
\end{proof}

\begin{proofOf}[lemma:completeTwoCliquesGamma]
    First, note that $\varphi_k$ is in fact bijective with inverse
    \begin{align*}
        \varphi_k\inv:
        \Sigma\k
        \mapsto
        \lr{
            \SigmaT\k_{ii} + \SigmaT\k_{jj} - 2 \SigmaT\k_{ij}
        }_{i,j = \oneToX{d}}
        =
        \gamma\lr{\SigmaT\k}
        ,
    \end{align*}
    where $\SigmaT\k \in \Rdd$ is defined as $\Sigma\k$ with a zero-valued $k$th row and column added,
    and $\gamma$ is from \cref{def:gamma}.
    Upon a permutation of the indices, the $(d-1) \times (d-1)$ matrix $\SigmaO\k$ takes the form~\cref{eq:SigmaABC} with blocks of indices $A = C_1 \without{k}$, $B = D_2 \without{k}$ (possibly empty), and $C = C_2 \without{k}$.
    \cref{lemma:completeTwoCliques} permits to find the unique \PD{} completion of $\SigmaO\k$, say $\Sigma\k$.
    The matrix
    \begin{align*}
        \Gamma := \varphi_k\inv \lr{\Sigma\k} = \gamma\lr{\SigmaT\k}  
    \end{align*}
    is thus well-defined,
    and since $\SigmaT\k$ is symmetric \PSD{}, with $\ker \SigmaT\k = \laSpan{\e_k}$,
    it follows from \cref{lemma:gammaMapsToCND} that $\Gamma$ is \CND{}.

    Since $\varphi_k$ and $\varphi_k\inv$ preserve specified entries,
    the condition
    \begin{align*}
        \Gamma_{ij} = \GammaO_{ij}
        , \qquad \forall\,
        (i,j) \in \Ebar,
    \end{align*}
    is satisfied by construction.
    Furthermore, from \cref{def:Theta} and \cref{lemma:completeTwoCliques} it follows that
    \begin{align*}
        \Theta_{ij} = \Theta\k_{ij} = 0
        , \qquad \forall\,
        (i,j) \notin \Ebar
        .
    \end{align*}
    The uniqueness of this completion follows from \cref{cor:completionUnique} below.
\end{proofOf}

\begin{lemma}
    \label{lemma:reverseMarginalPreservesZeros}
    Let $Y$ be a random variable following a \multGP{} distribution
    with positive, continuous exponent measure density $\lambda$.
    Let $G=(V,E)$ be a connected, decomposable graph $G=(V,E)$,
    consisting of two cliques $C_1$, $C_2$ separated by $D_2=C_1 \cap C_2$.
    Let $G'=(V',E')$ be a connected, decomposable graph
    with $V'=C_1$ and $E' \supseteq \restrictTo{E}{D_2}$.
    Suppose $Y$ satisfies the (extremal) pairwise Markov property relative to $G$,
    and the marginal $Y_{C_1}$
    conditionally on $\max_{i \in C_1} Y_i > 0$
    satisfies the (extremal) pairwise
    Markov property relative to $G'$.

    Then $Y$ satisfies the (extremal) pairwise Markov property relative to the graph
    $G'' = (V, E'')$ with $E'' = \restrictTo{E}{C_2} \cup E'$.
\end{lemma}

\begin{remark}
    Since the pairwise Markov property only requires
    \begin{align*}
        (i,j) \notin E
        \implies
        \exInd{Y_i}{Y_j}{Y_{\without{i, j}}}
        ,
    \end{align*}
    but not vice versa,
    a distribution can satisfy this for a number of distinct graphs.
    In particular, adding edges to a graph strictly weakens the condition above.
\end{remark}

\begin{proof}
    Since all graphs involved in the lemma are connected and decomposable,
    Theorem~1 from \cite{engelke2020} can be used to show
    \begin{align*}
        \lambdaM{y}
        &=
        \frac{
            \lambdaM[C_1]{y}
            \lambdaM[C_2]{y}
        }{
            \lambdaM[D_2]{y}
        }
        \\
        &=
        \frac{
            \lambdaM[C_2]{y}
        }{
            \lambdaM[D_2]{y}
        }
        \frac{
            \prod_{C \in \Cli'}
            \lambdaM[C]{y}
        }{
            \prod_{D \in \Sep'}
            \lambdaM[D]{y}
        }
        \\
        &=
        \frac{
            \prod_{C \in \Cli''}
            \lambdaM[C]{y}
        }{
            \prod_{D \in \Sep''}
            \lambdaM[D]{y}
        }
        ,
    \end{align*}
    with $\Cli'', \Sep''$ being the cliques and separators of $G''$,
    and $\Cli', \Sep'$ those of $G'$.
    Here, the function $\lambda_{I}$, for non-empty $I \subset V$, is the exponent measure density corresponding to the $I$-th marginal $Y_I$ conditionally on $\max_{i \in I} Y_i > 0$.
    By the same theorem,
    the above decomposition of $\lambda(y)$ implies that $Y$ satisfies the pairwise Markov property
    relative to $G''$.
\end{proof}

\begin{proofOf}[prop:completeDecomposable]
    Setting $\Gamma = \Gamma_N$,
    the condition $\Gamma_{ij} = \GammaO_{ij}$ for all $(i,j) \in \Ebar$
    is satisfied by construction,
    and \cref{lemma:completeTwoCliquesGamma} guarantees that each $\Gamma_n$ is \CND{}.
    The condition $\Theta_{ij} = 0$ for all $(i,j) \notin \Ebar$ is satisfied too,
    since \cref{lemma:reverseMarginalPreservesZeros} can be applied in each step.
    The uniqueness of this completion follows from \cref{cor:completionUnique}.
\end{proofOf}

\subsubsection{General graphs}
\label{subsubsec:proofsGeneralGraphs}

The proofs shown here closely follow the proofs in \cite{speedKivery1986}
and are adjusted where necessary to hold for \CND{} variogram matrices
instead of \PD{} covariance matrices. Recall the map $\sigma(\,\cdot\,)$ in \cref{eq:thetaSigma}.

\begin{definition}
    \label{def:KLdivGamma}
    Let $\Gamma_1, \Gamma_2 \in \Dcal_d$ and
    $\Sigma_1 = \proj{\Gamma_1}$,
    $\Sigma_2 = \proj{\Gamma_2}$.
    Let $P_1, P_2$ denote the probability measures of two
    degenerate normally distributed random vectors
    with mean $\zeroVec$ and covariance matrices $\Sigma_1, \Sigma_2 \in \Pcal_d^\oneVec$,
    and let $p_1, p_2$ be the corresponding densities on $\set{\oneVec}^\perp$ with respect to the $(d-1)$-dimensional Lebesgue measure.
    The \KL{} divergence $\KLdiv{}$ of these matrices is defined as
    \begin{align*}
        \KLdiv{\Gamma_1 \vert \Gamma_2}
        :=
        \KLdiv{\Sigma_1 \vert \Sigma_2}
        :=
        \KLdiv{P_1 \vert P_2}
        =
        \E[P_1]{\log p_1(x) - \log p_2(x)}
        .
    \end{align*}
\end{definition}

\begin{lemma}
    \label{lemma:divergence1}
    With $\Gamma_i$, $\Sigma_i$, $i=1,2$ as above,
    and $\pdet{\,\cdot\,}$ denoting the pseudo-determinant
    (see \cref{subsec:pseudoDet}),
    we have
    \begin{align*}
        \KLdiv{\Gamma_1 \vert \Gamma_2}
        =
        - \half \lr{
            \log \pdet{
                \Sigma_2\pinv \Sigma_1
            }
            +
            d - 1
            -
            \trace{
                \Sigma_2\pinv \Sigma_1
            }
        }
        .
    \end{align*}
\end{lemma}
\begin{proof}
	Since $\E[P_1]{x x\T} = \Sigma_1$ and since $p_i(x) \propto \pdet{\Sigma_i}^{-1/2} \expF{-\half x\T \Sigma_i\pinv x}$ with proportionality constant not depending on $i \in \{1, 2\}$, we have
    \begin{align*}
        \KLdiv{P_1 \vert P_2}
        &=
        \E[P_1]{\log p_1(x) - \log p_2(x)}
        \\
        &=
        \half \E[P_1]{
            - \log \pdet{\Sigma_1}
            - x\T \Sigma_1\pinv x
            + \log \pdet{\Sigma_2}
            + x\T \Sigma_2\pinv x
        }
        \\
        &=
        \half \log \frac{\pdet{\Sigma_2}}{\pdet{\Sigma_1}}
        +
        \half \E[P_1]{
            - x\T \Sigma_1\pinv x
            + x\T \Sigma_2\pinv x
        }
        \\
        &=
        - \half \log \pdet{
            \Sigma_2\pinv \Sigma_1
        }
        +
        \half \E[P_1]{
            - \trace{\Sigma_1\pinv x x\T}
            + \trace{\Sigma_2\pinv x x\T}
        }
        \\
        &=
        - \half \log \pdet{
            \Sigma_2\pinv \Sigma_1
        }
        -
        \half \trace{
            \Sigma_1\pinv \Sigma_1
        }
        +
        \half \trace{
            \Sigma_2\pinv \Sigma_1
        }.
    \end{align*}
    The statement then follows since $\Sigma_1\pinv\Sigma_1$
    is equal to the projection matrix $\ID$ onto $\set{\oneVec}^\perp$ and the trace of a projection matrix is equal to its rank.
\end{proof}

\begin{lemma}
    \label{lemma:pdetLogConcave}
    Consider a fixed $\Gamma_2 \in \Dcal_d$.
    Then the function $\KLdiv{\,\cdot\, \vert \Gamma_2}$ is convex on $\Dcal_d$.
\end{lemma}
\begin{proof}
    Since $\tr(\cdot)$ and $\sigma(\cdot)$ are linear maps,
    it remains to be shown that,
    for a fixed $\Sigma_2 \in \Pcal_d^\oneVec$,
    the map $A \mapsto \log \pdet{\Sigma_2\pinv A}$ is concave on $\Pcal_d^\oneVec$.
    Since $\Sigma_2\pinv$ and $A$ are symmetric with the same kernel,
    \cref{lemma:pseudoDet2} shows that 
    \begin{align*}
        \log \pdet{\Sigma_2\pinv A}
        =
        \log \pdet{\Sigma_2\pinv}
        +
        \log \pdet{A}
        ,
    \end{align*}
    which is a concave function in $A$ if and only if $A \mapsto \log \pdet{A}$ is concave.

    To show this,
    recall the notation from \cref{lemma:pseudoDet2} with $V = \laSpan{\set{\oneVec}}$,
    and let $s \in \brackets{0, 1}$,
    and $A_1, A_2 \in \Qcal_V$ \PSD{}.
    Then indeed
    \begin{align*}
        \log\pdet{
            s A_1 + (1-s) A_2
        }
        &=
        \log\abs{
            s B'_1 + (1-s) B'_2
        }
        \\ &\geq
        s \log\abs{
            B'_1
        }
        +
        (1-s) \log\abs{
            B'_2
        }
        \\ &=
        s \log\pdet{
            A_1
        }
        +
        (1-s) \log\pdet{
            A_2
        }
        ,
    \end{align*}
    where the inequality follows from
    the log-concavity of the determinant for \PD{} matrices
    \citep[see e.g.][p.~73f]{boydVandenberghe2004}.
\end{proof}

In the following results, let
$S \subseteq V \times V$ be such that $(i,j) \in S$ implies $(j,i) \in S$,
i.e., $S$ corresponds to the edges of an undirected graph on $V$,
optionally extended by some pairs $(i,i)$, for $i \in V$.

\begin{lemma}
    \label{lemma:divergence2}
    For
    $\Gamma_i, \Gamma'_i \in \Dcal_d$,
    $\Sigma_i = \proj{\Gamma_i}$,
    $\Theta_i = \ppinv{\Gamma_i}$,
    with $i \in \mathbb{N}$,
    the divergence $\KLdiv{\,\cdot\,\vert\,\cdot\,}$ has the following properties.
    \begin{enumerate}
        \item \label{enum:KL1}
            $\KLdiv{\Gamma_1 \vert \Gamma_2} \geq 0$,
            with equality if and only if $\Gamma_1=\Gamma_2$.
        \item \label{enum:KL2}
            Given $\Gamma_1, \Gamma_3 \in \Dcal_d$, if there exists a $\Gamma_2 \in \Dcal_d$ such that
            \begin{enumerate}
                \item \label{enum:KL2a}
                    $\lr{\Gamma_2}_{ij} = \lr{\Gamma_1}_{ij}$, for $(i,j) \in S$, and
                \item \label{enum:KL2b}
                    $\lr{\Theta_2}_{ij} = \lr{\Theta_3}_{ij}$, for $(i,j) \notin S$,
            \end{enumerate}
            then
            \begin{align*}
                \KLdiv{\Gamma_1 \vert \Gamma_3}
                =
                \KLdiv{\Gamma_1 \vert \Gamma_2}
                +
                \KLdiv{\Gamma_2 \vert \Gamma_3}
                .
            \end{align*}
            If such a $\Gamma_2$ exists, it is unique.
        \item \label{enum:KL3}
            If $\set{\Gamma_n}$ and $\set{\Gamma'_n}$ are sequences contained in compact
            subsets of $\Dcal_d$, then
            $\KLdiv{\Gamma_n \vert \Gamma'_n} \rightarrow 0$
            implies
            $\Gamma_n - \Gamma'_n \rightarrow 0$ as $n \to \infty$.
    \end{enumerate}
\end{lemma}
\begin{proof}
    Statement \ref{enum:KL1} is a well-known property of the \KL{} divergence,
    which can be applied since $\sigma$ is injective (see \cref{prop:thetaHomeo}),
    and each $\Sigma \in \Pcal_d^\oneVec$ characterizes a different distribution.

    To show statement \ref{enum:KL2}, use the multiplicative property of the pseudo-determinant on $\Pcal_d^{\oneVec}$ to compute
    \begin{align*}
        &\hphantom{=}
        -2 \lr{
            \KLdiv{\Gamma_1 \vert \Gamma_2}
            +
            \KLdiv{\Gamma_2 \vert \Gamma_3}
        }
        \\
        &=
        \log \pdet{
            \Sigma_2\pinv \Sigma_1
        }
        +
        \trace{
            \Sigma_1\pinv \Sigma_1
        }
        -
        \trace{
            \Sigma_2\pinv \Sigma_1
        }
        +
        \log \pdet{
            \Sigma_3\pinv \Sigma_2
        }
        +
        \trace{
            \Sigma_2\pinv \Sigma_2
        }
        -
        \trace{
            \Sigma_3\pinv \Sigma_2
        }
        \\
        &=
        \log \pdet{\Sigma_3\pinv \Sigma_2 \Sigma_2\pinv \Sigma_1}
        +
        \trace{
            \Sigma_1\pinv \Sigma_1
            -
            \Sigma_2\pinv \Sigma_1
            +
            \Sigma_2\pinv \Sigma_2
            -
            \Sigma_3\pinv \Sigma_2
        }
        \\
        &=
        \log \pdet{\Sigma_3\pinv \Sigma_1}
        +
        \trace{
            \lr{
                \Sigma_2\pinv - \Sigma_3\pinv
            }
            \lr{
                \Sigma_2 - \Sigma_1
            }
            -
            \Sigma_3\pinv \Sigma_1
            +
            \Sigma_1\pinv \Sigma_1
        }
        \\
        &=
        -2
        \KLdiv{\Gamma_1 \vert \Gamma_3}
        +
        \trace{
            \lr{
                \Sigma_2\pinv - \Sigma_3\pinv
            }
            \lr{
                \Sigma_2 - \Sigma_1
            }
        }
        .
    \end{align*}
    Since $\Sigma_i = \ID\tlr{-\half\Gamma_i}\ID$ and since $\ID \Theta_i = \Theta_i \ID = \Theta_i$, the invariance of the trace operator under cyclic permutations permits writing
    \begin{align*}
        \trace{
            \lr{
                \Sigma_2\pinv - \Sigma_3\pinv
            }
            \lr{
                \Sigma_2 - \Sigma_1
            }
        }
        &=
        -\half \trace{
            \lr{
                \Theta_2 - \Theta_3
            }
            \ID
            \lr{
                \Gamma_2 - \Gamma_1
            }
        	\ID
        }
        \\
        &=
        -\half
        \trace{
            \lr{
                \Theta_2 - \Theta_3
            }
            \lr{
                \Gamma_2 - \Gamma_1
            }
        }
    	\\
    	&=
    	-\half
    	\sum_{(i,j) \in V \times V} \lr{\Theta_2-\Theta_3}_{ij} \lr{\Gamma_2-\Gamma_1}_{ij}
        .
    \end{align*}
    Each term in the sum is zero by \ref{enum:KL2a} and \ref{enum:KL2b}.

    Uniqueness can be shown as in the \PD{} case by considering
    $\Gamma_2, \Gamma_2'$ that satisfy \ref{enum:KL2a} and \ref{enum:KL2b}.
    If $\Gamma_2, \Gamma_2'$ both satisfy the said properties with respect to $\Gamma_1$ and $\Gamma_3$, then they obviously also satisfy those properties with respect to $\Gamma_1' = \Gamma_3' = \Gamma_2$.
    But then,
    property~\ref{enum:KL1} and the first part of property~\ref{enum:KL2} yield
    \begin{align*}
        &
        0
        =
        \KLdiv{\Gamma_2 \vert \Gamma_2}
        =
        \KLdiv{\Gamma_2 \vert \Gamma_2'}
        +
        \KLdiv{\Gamma_2' \vert \Gamma_2}
        \\
        \implies \quad &
        \KLdiv{\Gamma_2 \vert \Gamma_2'}
        =
        \KLdiv{\Gamma_2' \vert \Gamma_2}
        =
        0
        \\
        \implies \quad &
        \Gamma_2 = \Gamma_2'
        .
    \end{align*}

    Lastly, suppose $\set{\Gamma_n}$ and $\set{\Gamma'_n}$ are 
    as in statement \ref{enum:KL3},
    but $\Gamma_n - \Gamma'_n \not\rightarrow 0$.
    Then there are convergent subsequences
    $\Gamma_{n_i} \rightarrow \Gamma^*$ and 
    $\Gamma'_{n_i} \rightarrow \Gamma^{*'}$
    with $\Gamma^* \neq \Gamma^{*'}$.
    By continuity of $\KLdiv{}$,
    $
        \KLdiv{\Gamma_{n_i}\mid \Gamma'_{n_i}}
        \rightarrow
        \KLdiv{\Gamma^* \mid \Gamma^{*'}}
        \neq
        0
    $,
    which is a contradiction.
\end{proof}

\begin{corollary}
    \label{cor:completionUnique}
    Let $G=(V,E)$ be a connected graph,
    $\GammaO$ a \PCND{} matrix on $G$,
    and $\Gamma$ a graphical completion of $\GammaO$,
    in the sense of \cref{def:matrixCompletionProblem}.
    Then
    \cref{lemma:divergence2}, \cref{enum:KL2},
    with $\Gamma_1 = \Gamma_3 = \Gamma$
    and $S = \Ebar$
    shows the uniqueness of the completion from
    \cref{lemma:completeTwoCliquesGamma,prop:completeDecomposable,prop:completeGeneral}.
\end{corollary}

Recall $S \subseteq V \times V$ as introduced prior to \cref{lemma:divergence2}.

\begin{lemma}
    \label{lemma:prop3SpeedKivery}
    Let $S_1, \dots, S_m \subseteq S$ be such that their union is $S$.
    Let
    $\Gamma \in \Dcal_d$,
    $\Theta \in \Pcal_d^{\oneVec}$ (not necessarily $\Theta = \theta(\Gamma)$),
    and
    $\Gamma_n \in \Dcal_d$,
    $\Theta_n = \ppinv{\Gamma_n}$,
    for $n \in \mathbb{N}_0$,
    with $\Gamma_0 = \Gamma$.

    If each $\Gamma_n$, $n \geq 1$, satisfies
    \begin{alignat*}{2}
        \lr{\Theta_n}_{ij}
        &=
        \Theta_{ij}
        && \qquad \mathrm{if} \,
        (i,j) \in S_t
        ,
        \\
        (\Gamma_{n-1})_{ij}
        &=
        (\Gamma_n)_{ij}
        && \qquad \mathrm{if} \,
        (i,j) \notin S_t
        ,
    \end{alignat*}
    with $t \equiv n \mod m$,
    then the sequence $\lr{\Gamma_n}_n$ converges to the unique
    $\Gamma^S \in \Dcal_d$ that, writing $\Theta^S = \ppinv{\Gamma^S}$, satisfies
    \begin{alignat}{2}
        \label{eq:prop3speedKiiveri1}
        \Theta^S_{ij}
        &=
        \Theta_{ij}
        && \qquad \mathrm{if} \,
        (i,j) \in S
        ,
        \\
        \label{eq:prop3speedKiiveri2}
        \Gamma^S_{ij}
        &=
        \Gamma_{ij}
        && \qquad \mathrm{if} \,
        (i,j) \notin S
        .
    \end{alignat}
\end{lemma}
\begin{proof}
    Using \cref{lemma:divergence2}, \cref{enum:KL2},
    decompose the \KL{} divergence $\KLdiv{\Gamma_0|\gamma\lr{\Theta\pinv}}$
    as follows:
    \begin{align*}
        \KLdiv{\Gamma_{n-1} | \gamma\lr{\Theta\pinv}}
        &=
        \KLdiv{\Gamma_{n-1} | \Gamma_n}
        + \KLdiv{\Gamma_{n} | \gamma\lr{\Theta\pinv}}
        \\
        \Longrightarrow \quad
        \KLdiv{\Gamma_{0} | \gamma\lr{\Theta\pinv}}
        &=
        \KLdiv{\Gamma_n | \gamma\lr{\Theta\pinv}}
        + \sum_{k=1}^n \KLdiv{\Gamma_{k-1} | \Gamma_k}
        .
    \end{align*}
    From here
    it follows that $\sum_{k=1}^\infty \KLdiv{\Gamma_{k-1}|\Gamma_k}$
    is a convergent series so that
    $\KLdiv{\Gamma_{k-1}|\Gamma_k} \rightarrow 0$ for $k \rightarrow \infty$,
    and also that
    \begin{align*}
        \set{\Gamma_k}_k
        \subseteq
        \setM{
            \Gamma' \in \Dcal_d
        }{
            \KLdiv{\Gamma' \vert \gamma\lr{\Theta\pinv}}
            \leq
            \KLdiv{\Gamma_0 \vert \gamma\lr{\Theta\pinv}}
        }
        =: A
        ,
    \end{align*}
    which is a compact set due to the convexity of
    $\KLdiv{\Gamma|\gamma\lr{\Theta\pinv}}$
    as a function of $\Gamma$ (see \cref{lemma:pdetLogConcave})
    and the facts that
    $\Dcal_d$ is a convex cone
    and $\Gamma \mapsto \KLdiv{\Gamma|\gamma\lr{\Theta\pinv}}$ attains its minimum uniquely at $\Gamma = \gamma\lr{\Theta\pinv}$.
    
    Hence, the sequence $v_s := (\Gamma_{sm+1}, \dots, \Gamma_{sm+m})$, $s \in \N$,
    has a convergent subsequence indexed by $s \in \N^* \subseteq \N$
    with limit $(\Gamma_1^*, \dots, \Gamma_m^*)$.
    For any $2 \leq t \leq m$,
    the entries of this limit satisfy
    \begin{align*}
        (\Gamma_t^* - \Gamma_{t-1}^*)
        &=
        (\Gamma_t^* - \Gamma_{sm+t})
        +
        (\Gamma_{sm+t} - \Gamma_{sm+t-1})
        +
        (\Gamma_{sm+t-1} - \Gamma_{t-1}^*)
        .
    \end{align*}
    For $s \in \N^* \rightarrow \infty$,
    the first and last term in this sum converge to zero by the definition of $\Gamma^*_t$,
    and the second term converges to zero by \cref{lemma:divergence2}, \cref{enum:KL3}.
    Hence, $\Gamma_t^* = \Gamma_{t-1}^* =: \Gamma^*$ for all $2 \leq t \leq m$.
    Since the elements of the sequence $\Gamma_{sm+t}$, $s \in \N^*$, $t = \oneToX{m}$,
    satisfy \cref{eq:prop3speedKiiveri2} (always),
    and \cref{eq:prop3speedKiiveri1} infinitely often,
    the limit $\Gamma^*$ satisfies both conditions, as well.

    The above argument can be repeated for all other convergent subsequences
    of $(v_s)_{s \in \N \setminus \N^*}$
    and \cref{lemma:divergence2}, \cref{enum:KL2},
    then shows that the limits of all these subsequences must be identical,
    hence $(\Gamma_k)_{k \in \N}$ converges to the unique
    $\Gamma^S$ satisfying \cref{eq:prop3speedKiiveri1,eq:prop3speedKiiveri2}.
\end{proof}

\begin{proofOf}[prop:completeGeneral]
    The result follows from \cref{lemma:prop3SpeedKivery} with
    $S_i = E_i^c$ and $\Theta$ being the Laplacian matrix of the graph $G$ (see \cref{def:graphLaplacian}).
    Uniqueness follows from \cref{cor:completionUnique}.
\end{proofOf}

\subsubsection{Example of non-completable matrix}
\label{subsubsec:proofNonCompletableGraphs}

\begin{proofOf}[exmp:noCompletion]
    A useful property of variogram matrices is the fact
    that they can equivalently be interpreted as Euclidean distance matrices.
    For instance, \cite{gower1982} shows that
    a matrix is (strictly) \CND{} if and only if
    it is the $d \times d$ matrix consisting of the squared distances
    $\norm{p_i - p_j}^2$
    of a set of points
    $p_1, \dots, p_d \in \Rd[d-1]$ that do not lie in a lower dimensional affine hyper-plane.

    If there was a completion of the matrix $\GammaO$ from the \namecref{exmp:noCompletion},
    then the above interpretation as a (partial) Euclidean distance matrix 
    would imply the existence of a set of points $p_1, \dots, p_d$,
    with distances $\norm{p_i - p_{i+1}} = 1$ for $i = \oneToX{d-1}$,
    and $\norm{p_1 - p_d} = 2d$.
    However, by the triangle inequality this leads to the contradiction
    \begin{align*}
        2d
        =
        \norm{p_1 - p_d}
        \leq
        \sum_{i=1}^{d-1} \norm{p_i - p_{i+1}}
        =
        d
        .
    \end{align*}
    \cite{bakonyiJohnson1995} show that such a counter-example can be constructed
    for any non-decomposable graph.
\end{proofOf}

\subsection{Proofs of statistical inference related results}

\subsubsection{Matrix completion as likelihood optimization}

\begin{proofOf}[prop:completionOptProblem]
    Let
    $U: \Rdd \rightarrow \Rd[d(d-1)/2]$
    denote the mapping that maps a matrix to the vector containing the entries in its upper triangular part (excluding the diagonal).
    Note that the restriction $\restrictTo{U}{\Pcal_d}$ is invertible,
    since the lower triangular part of $\Theta \in \Pcal_d$ is defined by symmetry
    and the diagonal is such that the row sums are zero.
    In the computations below,
    we consider $P \in U(\Pcal_d)$ and
    write $\Theta :\equiv \Theta(P)$ for notational convenience.

    Consider the function $f(P) = \log\pdet{\Theta} + \half \trace{\overline\Gamma\Theta}$.
    \citet[][Proposition~A.5]{roettger2023}
    show that
    \begin{align*}
        \nabla_P \log \pdet{\Theta}
        =
        U(
            -\gamma\lr{\Theta\pinv}
        )
        .
    \end{align*}
    Furthermore, by symmetry of the matrices $\Theta$ and $\overline\Gamma$,
    \begin{align*}
        \nabla_P
        \trace{\Theta\overline\Gamma}
        =
        \nabla_P
        \sum_{i,j \in \oneToX{d}}
        \Theta_{ij}
        \overline\Gamma_{ij}
        =
        \nabla_P
        \sum_{\substack{
            i,j \in \oneToX{d}
            \\
            i<j
        }}
        2 \Theta_{ij} \overline\Gamma_{ij}
        =
        U(
            2 \overline\Gamma
        )
        ,
    \end{align*}
    and hence
    \begin{align*}
        \nabla_P f(\Theta)
        =
        U(\overline\Gamma - \gamma(\Theta))
        .
    \end{align*}
    The map $\Theta \mapsto \trace{\overline\Gamma\Theta}$
    is linear
    and the proof of
    \cref{lemma:pdetLogConcave}
    shows that the map $\Theta \mapsto \log\pdet{\Theta}$
    is concave.
    Hence, the maximizer of $f$ under the constraint
    $\Theta_{ij} = 0$ for $(i,j) \notin \Ebar$
    satisfies
    \begin{align*}
        \gamma\lr{\Theta}_{ij} = \overline\Gamma_{ij}
        , \quad
        \forall (i,j) \in \Ebar
        .
    \end{align*}
    \cref{prop:completeGeneral} shows that the unique solution 
    satisfying this condition and the constraint is $\theta\lr{\comp[G]{\overline\Gamma}}$.
    
\end{proofOf}

\subsubsection{Matrix completion is continuous}

\begin{proofOf}[lemma:matrixCompletionContinuous]
	Enumerate the edges as $E=\set{e_1, \hdots, e_m}$ with $m=\abs{E}$ the cardinality of $E$.
    Let $\ppinv{}$ denote the mapping sending a variogram matrix $\Gamma$ to its precision matrix $\Theta = \ppinv{\Gamma}$
    as constructed in \cref{prop:ThetaReps}.
    Let $\pi_G$ be the restriction map
    \begin{alignat*}{3}
        \myFunctionDefinition{\pi_G}{
            \Rdd
        }{
            \Rd[m]
        }{
            M
        }{
            v
        }
    \end{alignat*}
    where $v$ has entries $v_k$ for $k=1,\dots, m$, with
    \begin{align*}
        v_k = M_{ij}
        \text{~for~}
        e_k = (i,j),
        i<j
        .
    \end{align*}
    In words, $\pi_G$ extracts from a matrix $M$ the $m$ elements $M_{ij}$ corresponding to the edges $(i, j) \in E$ and stacks them in a vector.

    Let $\Qcal_G$ be the set of symmetric $d \times d$ matrices with zero row sums and with zeros in off-diagonal positions corresponding to non-edges of $G$:
    \[
    	\Qcal_G := \setM{M \in \Rdd}{
    	M = M\T, \;
    	M \oneVec = \zeroVec, \;
    	M_{ij} = 0 \; \forall (i, j) \not\in \Ebar
    }.
    \]
    Clearly, $M \in \Qcal_G$ is determined by its elements $M_{ij}$ for $(i, j) \in E$ such that $i < j$, i.e., by $\pi_G(M)$.
    Indeed, the elements of $M$ below the diagonal are determined by the symmetry constraint $M = M\T$ and its diagonal elements are determined by the zero row-sum constraint $M \oneVec = \zeroVec$.
    Since all restrictions imposed on $M$ are linear, $\Qcal_G$ is a linear subspace of $\Rdd$ of dimension $m$.
    The space $\Qcal_G$ is thus isomorphic with $\Rd[m]$: identify $M \in \Qcal_G$ with the elements $M_{ij}$ for index pairs $(i, j) \in E$ such that $i < j$.
    Furthermore, the map $\pi_G$ above, when restricted to $\Qcal_G$, is a linear isomorphism between $\Qcal_G$ and $\Rd[m]$; in particular, it is continuous.

	The set of precision matrices $\Theta \in \Pcal_d^1$ with $\Theta_{ij} = 0$ for $(i, j) \not\in \Ebar$, i.e., those that figure in the matrix completion problem in \cref{def:matrixCompletionProblem}, is equal to
	\[
		\Pcal_G
		:= \setM{M \in \Qcal_G}{
					\ker M = \laSpan{\set{\oneVec}}, \;
					M ~ \text{\PSD}
				}.
	\]
	Let $\lambda_1(S) \ge \ldots \ge \lambda_d(S)$ be the $d$ real eigenvalues of a symmetric matrix $S \in \Rdd$, counted with multiplicities and ordered decreasingly. We have
    \[
    	\Pcal_G
    	= \setM{M \in \Qcal_G}{\lambda_{d-1}(M) > 0}.
    \]
    The functions $\lambda_j$ are well-known to be Lipschitz on $\setM{S \in \Rdd}{S = S\T}$.
    It follows that $\Pcal_G$ is open in $\Qcal_G$. Upon the above identification of $\Qcal_G$ with $\Rd[m]$ via $\pi_G$, we can thus view $\Pcal_G$ as an open subset of $\Rd[m]$. Formally, this subset is denoted as
    \[
    	\mathring\Pcal_G := \pi_G(\Pcal_G) \subset \Rd[m].
    \]
    Let $f_1$ denote the inverse of the restriction of $\pi_G$ to $\Pcal_G$, that is,
    \begin{alignat*}{3}
	\myFunctionDefinition{f_1}{
		\mathring\Pcal_G
	}{
		\Pcal_G
	}{
		v
	}{
		M
		=
		\pi_G\inv\lr{v}
	}
	\end{alignat*}
	Since $\pi_G$ was continuously invertible on the set $\Qcal_G$ that contains $\Pcal_G$, the function $f_1$ is a continuous bijection too.

    The set of variogram matrices $\Gamma$ corresponding to precision matrices $M$ in $\Pcal_G$ is
    \[
    	\Dcal_G
    	:= \theta^{-1}\tlr{\Pcal_G}
    	= \setM{\Gamma \in \Dcal_d}{\theta(\Gamma) \in \Pcal_G}
    	= \set{\theta^{-1}\tlr{M} : M \in \Pcal_G}.
    \]
    By definition, these are exactly the variogram matrices $\Gamma \in \Dcal_d$ of which the associated precision matrices $M = \theta(\Gamma)$ have zero elements $M_{ij} = 0$ for $(i, j) \not\in \Ebar$. In other words, $\Dcal_G$ corresponds to all possible solutions of the matrix completion problem in \cref{def:matrixCompletionProblem} with respect to $G$.

    By \cref{prop:completeGeneral}, a variogram matrix $\Gamma$ in $\Dcal_G$ is uniquely determined by the values of $\Gamma_{ij}$ for $(i, j) \in \Ebar$. Since $\Gamma$ has zero diagonal, this means that $\Gamma \in \Dcal_G$ is uniquely determined by $\pi_G(\Gamma) = (\Gamma_{ij} : (i, j) \in E, \, i < j)$. Another way to say the same thing is that the map $\pi_G$ restricted to $\Dcal_G$ is injective. For clarity, let
    \[
    	\mathring\Dcal_G := \pi_G(\Dcal_G) \subset \Rd[m]
    \]
    denote the image of $\Dcal_G$ under $\pi_G$ and let $f_2$ denote the restriction of $\pi_G$ to $\Dcal_g$:
    \begin{alignat*}{3}
    \myFunctionDefinition{f_2}{
    	\Dcal_G
    }{
    	\mathring\Dcal_G
    }{
    	\Gamma
    }{
    	\pi_G\lr{\Gamma}
        .
    }
    \end{alignat*}
    The inverse mapping of $f_2$ is the completion map $\comp[G]{}$ in~\cref{def:notationComp}.\footnote{Actually, this comes with some abuse of notation: in \cref{def:notationComp}, the set $\mathring{D}_G$ and the completion map $\comp[G]{}$ are defined with the placeholder $\undefined$ for pairs $(i, j) \not\in \Ebar$, whereas in this proof, these unknown elements are omitted altogether.}
    Knowing that $f_2$ is continuous (since $\pi_G$ is continuous), we need to show that $\comp[G]{}$ is continuous as well.

    To do this,
    consider the mapping
    \begin{alignat*}{3}
        \myFunctionDefinition{f}{
            \mathring \Pcal_G
        }{
            \mathring \Dcal_G
        }{
            v
        }{
            \lr{
                f_2 \circ \theta^{-1} \circ f_1
            } \lr{
                v
            }
        }
        .
    \end{alignat*}
    The map $f$ is represented schematically in \cref{fig:f}.
    It sends the restriction $v = \pi_G(M) \in \mathring\Pcal_G \subset \Rd[m]$
    of a precision matrix $M \in \Pcal_G$ to the restriction
    $\pi_G(\Gamma) \in \mathring\Dcal_G \subset \Rd[m]$ of the associated variogram matrix $\Gamma = \theta^{-1}(M)$;
    indeed, we have $M = f_1\tlr{v}$ and thus $f\tlr{v} = \pi_G\tlr{\theta^{-1}\tlr{M}}$.
    The map $f$ is a composition of three continuous bijections and thus a continuous bijection as well.
    The domain of $f$ is $\mathring\Pcal_G$, which was shown to be an open subset of $\Rd[m]$.
    The image of $f$ is $\mathring\Dcal_G$, a subset of $\Rd[m]$ as well.
    By the Brouwer Invariance of Domain Theorem \citep[see e.g.][]{kulpa1998}, $\mathring\Dcal_G$ is open and $f$ is a homeomorphism.
    But since
    \[
    	\comp[G]{} = f_2^{-1} = \theta^{-1} \circ f_1 \circ f^{-1},
    \]
    it follows that $\comp[G]{}$ is continuous too, as required.
    \begin{figure}
        \centering
        \inputTikz[]{completionContinuous}
        \caption{\label{fig:f} Illustration of the definition of $f$}
    \end{figure}
\end{proofOf}

\begin{proofOf}[thm:matrixCompletionConsistent]
    First, we show that
    with probability tending to one
    $\widehat\GammaO$ allows a completion.
    To this end, recall the definition of the set of \CND{} matrices
    from \cref{eq:defVariogramSet},
    \begin{align}
        \label{eq:defVariogramSetProof}
        \Dcal_d
        =
        \setM{
            M \in \Rdd
        }{
            M = M\T
            \,\land\,
            \diag{M} = \zeroVec
            \,\land\,
            v\T M v < 0
            \,\forall\,
            \zeroVec \neq v \perp \oneVec
        }
        ,
    \end{align}
    and let
    $K = \setM{v \in \Rd}{v \perp \oneVec, \norm[\infty]{v} = 1}$.
    The set $K$ is compact and, hence,
    the value
    $\Delta_M = \max_{v \in K} v\T M v$
    exists.
    Since for $v \neq \zeroVec$ we have
    \begin{align*}
        v\T M v
        =
        \norm[\infty]{v}^2
        \cdot
        \lr{v \norm[\infty]{v}\inv}\T
        M
        \lr{v \norm[\infty]{v}\inv}
        ,
    \end{align*}
    with $\lr{v \norm[\infty]{v}\inv} \in K$,
    the inequality condition in \cref{eq:defVariogramSetProof}
    is equivalent to $\Delta_M < 0$.

    Next, 
    let $G$, $\Gamma$, and $\widehat{\GammaO}$ be as in
    the \namecref{thm:matrixCompletionConsistent}.
    Let $\varepsilon = -\Delta_\Gamma/\tlr{2d^2} > 0$,
    and consider the ball
    $B_\varepsilon = \setM{M \in \Rdd}{\norm[\infty]{M - \Gamma} \leq \varepsilon}$,
    where $\norm[\infty]{M}$
    denotes the infinity norm applied to $M$
    interpreted as a $d^2$-dimensional vector.
    Then any $M \in B_\varepsilon$ satisfies
    \begin{align*}
        \max_{v \in K}
        v\T M v
        &=
        \max_{v \in K}
        \sum_i
        \sum_j
        M_{ij} v_i v_j
        \\ &\leq
        \max_{v \in K}
        \sum_i
        \sum_j
        \Gamma_{ij} v_i v_j
        +
        \varepsilon \abs{v_i v_j}
        \\ &\leq
        \max_{v \in K}
        v\T \Gamma v
        - \Delta_\Gamma / 2
        \\ &=
        \Delta_\Gamma / 2
        \\ &<
        0
        .
    \end{align*}
    Let $\Rcal_d = \setM{M \in \Rdd}{M = M\T \,\land\, \diag{M}=\zeroVec}$
    and observe that
    $\Rcal_d \cap B_\varepsilon \subseteq \Dcal_d$.
    For a set of matrices $S \subseteq \Rdd$
    denote $\restrictTo{S}{G} = \setM{\restrictTo{M}{G}}{M \in S}$,
    with $\restrictTo{\cdot}{G}$ as in \cref{eq:defRestrictToGraph},
    and observe that
    \begin{align*}
        \restrictTo{\Rcal_d}{G}
        \cap
        \restrictTo{B_\varepsilon}{G}
        \subseteq
        \restrictTo{\Dcal_d}{G}
        =
        \mathring\Dcal_G
        .
    \end{align*}
    By assumption,
    any realization of
    $\widehat\GammaO$ is always in $\restrictTo{\Rcal_d}{G}$
    and hence
    \begin{align*}
        \P{\widehat\GammaO \in \mathring\Dcal_G}
        =
        \P{\widehat\GammaO \in \restrictTo{B_\varepsilon}{G}}
        =
        \P{\max_{(i,j) \in E} \abs{\widehat\GammaO_{ij} - \Gamma_{ij}} \leq \varepsilon}
        \longrightarrow
        1
        .
    \end{align*}
    Together with the continuity from \cref{lemma:matrixCompletionContinuous}
    this completes the proof.
\end{proofOf}

\end{document}